\documentclass[12pt]{article}
\usepackage{amsthm}
\usepackage{amssymb}
\usepackage{amsmath}

\usepackage[utf8]{inputenc}
\usepackage[T1]{fontenc}

\usepackage[polish, english]{babel}
\usepackage{a4wide}

\usepackage[shortcuts]{extdash}

\usepackage[hidelinks]{hyperref}

\usepackage{latexsym}
\usepackage[dvipsnames]{xcolor}
\usepackage{tikz}
\usepackage{xspace}
\usepackage{thm-restate}

\usepackage[capitalize,nameinlink]{cleveref}
\usepackage{mathtools}
\usepackage{environ}
\usepackage{todonotes}
\usepackage[noadjust,nocompress]{cite}

\usepackage{proof,paralist}

\newtheorem{theorem}{Theorem}[section]
\newtheorem{lemma}[theorem]{Lemma}
\newtheorem{fact}[theorem]{Fact}
\newtheorem{definition}[theorem]{Definition}
\newtheorem{problem}[theorem]{Problem}
\newtheorem{proposition}[theorem]{Proposition}

\newtheorem{corollary}[theorem]{Corollary}

\newtheorem{remark}[theorem]{Remark}

\newtheorem{claim}[theorem]{Claim}

\numberwithin{equation}{section}
\numberwithin{figure}{section}

\newcommand{\ident}[1]{\ensuremath{\texorpdfstring{\mathrm{\mathnm{#1}}}{#1}}\xspace}

\newcommand{\domain}[1]{\ensuremath{\mathbb{#1}}\xspace}

\newcommand{\powerset}{\mathsf{\mathnm{P}}}

\newcommand{\N}{\domain{N}}

\newcommand{\R}{\domain{R}}

\DeclareMathOperator*{\prob}{\mathbb P}

\newcommand{\eqdef}{\stackrel{\text{def}}=}

\renewcommand{\mod}[1]{\allowbreak\mkern10mu({\operator@font mod}\,\,#1)}

\newcommand{\mathcalsym}[1]{\ensuremath{\mathcal{#1}}\xspace}

\newcommand{\CA}{\mathcalsym{A}}
\newcommand{\CB}{\mathcalsym{B}}

\newcommand{\CD}{\mathcalsym{D}}

\newcommand{\CI}{\mathcalsym{I}}

\newcommand{\CP}{\mathcalsym{P}}

\newcommand{\CR}{\mathcalsym{R}}
\newcommand{\CS}{\mathcalsym{S}}

\newcommand{\CU}{\mathcalsym{U}}
\newcommand{\CV}{\mathcalsym{V}}
\newcommand{\CW}{\mathcalsym{W}}

\newcommand{\comment}[1]{}

\newcommand{\from}{\colon}

\newcommand{\dom}{\ident{dom}}

\newcommand{\lang}{\ident{L}}

\usetikzlibrary{positioning}
\usetikzlibrary{decorations.pathreplacing}
\usetikzlibrary{automata,positioning}
\usetikzlibrary{decorations.pathmorphing}
\usetikzlibrary{decorations.markings}
\usetikzlibrary{decorations}
\usetikzlibrary{arrows}
\usetikzlibrary{patterns}
\usetikzlibrary{calc}
\usetikzlibrary{shapes}
\usetikzlibrary{fadings,	shadings}

\usetikzlibrary{arrows.meta}
	\tikzstyle{ubrace} = [draw, thick, decoration={brace, amplitude=7pt, mirror, raise=0.0cm}, decorate,
	    every node/.style={anchor=north, yshift=-0.15cm}]
	\tikzstyle{rbrace} = [draw, thick, decoration={brace, amplitude=7pt, mirror, raise=0.0cm}, decorate,
	    every node/.style={anchor=west, xshift= 0.15cm}]

	\tikzstyle{obrace} = [draw, thick, decoration={brace, amplitude=7pt, raise=0.0cm}, decorate,
	    every node/.style={anchor=south, yshift= 0.15cm}]
	\tikzstyle{lbrace} = [draw, thick, decoration={brace, amplitude=7pt, raise=0.0cm}, decorate,
	    every node/.style={anchor=east, xshift=-0.15cm}]

\newcommand{\eval}[2]{%
	\pgfmathparse{#2}%
	{\global\edef#1{\pgfmathresult}}%
}

\newcommand{\evalInt}[2]{%
	\pgfmathparse{int(#2)}%
	{\global\edef#1{\pgfmathresult}}%
}

\newcommand{\dL}{\mathtt{{\scriptstyle L}}}
\newcommand{\dR}{\mathtt{{\scriptstyle R}}}

\newcommand{\mathnm}[1]{#1}

\newcommand{\restr}{{\upharpoonright}}
\newcommand{\sufcut}{{\downharpoonleft}}

\tikzstyle{dot} = [draw,shape=circle,fill, minimum size=1mm, inner sep=0pt, outer sep=0pt]
\tikzstyle{edge}  = [draw, thick,->]

\evalInt{\zQtity}{5}
\eval{\xScale}{0.5}
\eval{\yScale}{1.0}

\eval{\xZero}{0}
\eval{\yZero}{0}

\newcommand{\getXofR}[2]{\xZero + \xScale * ((#1 + 1.0) * (\zQtity + 1.0) + (#2))}
\newcommand{\getYofR}[1]{\yZero - \yScale * (#1)}

\newcommand{\showT}[2]{
	\eval{\xDist}{#1}
	
	\eval{\xL}{\getXofR{\minD}{\minD}}
	\eval{\xR}{\getXofR{\maxD}{\maxD}}
	\eval{\y}{\getYofR{-0.2}+0.3}
	
	\node at (\xL-\xDist,\y) {#2};
	\node at (\xR+\xDist,\y) {#2};

	\foreach \i in {\minD,...,\maxD} {
		\eval{\mx}{\getXofR{\i}{\zQtity/2}}

		\node[scale=1] at (\mx,\y) {\i};
	}
}

\newcommand{\showLabel}[1]{
	\eval{\lx}{\getXofR{\minD}{0}-0.4}
	\eval{\ly}{\getYofR{\lev-0.5}}

	\node[scale=1, anchor=mid east] at (\lx,\ly) {#1};
}

\newcommand{\showR}[2]{
	\evalInt{\lev}{#1}
	\evalInt{\n}{#2}

	\eval{\lx}{\getXofR{0}{\dPar}}
	\eval{\ly}{\getYofR{\lev}}

	\foreach \i in {\minD,...,\maxD} {
		\eval{\xL}{\getXofR{\i}{0}}
		\eval{\xR}{\getXofR{\i}{\zQtity}}
		\eval{\y}{\getYofR{\lev}}

		\node[dot] (c\lev-\i-L) at (\xL,\y) {};
		\node[dot] (c\lev-\i-R) at (\xR,\y) {};
		
		\draw[line width=0.5mm] (c\lev-\i-L) -- (c\lev-\i-R);
	}

	\eval{\lx}{\getXofR{1}{0}}
	\eval{\rx}{\getXofR{\dPar}{\dPar}}
	\eval{\y}{\getYofR{\lev}}

}

\newcommand{\showCopy}[4]{
	\evalInt{\si}{#1}
	\evalInt{\sj}{#2}

	\evalInt{\ti}{#3}
	\evalInt{\tj}{#4}

	\eval{\uy}{\getYofR{\lev-0.9}}
	\eval{\dy}{\getYofR{\lev-0.2}}

	\eval{\sx}{\getXofR{\si}{\sj-0.5}}
	\eval{\tx}{\getXofR{\ti}{\tj-0.5}}

	\draw[-{Latex[width=0.6mm, length=0.6mm]}, line width=0.01mm] (\sx,\dy) -- (\tx,\uy);
}

\newcommand{\showDots}[3]{
	\evalInt{\n}{#1}
	\evalInt{\si}{#2}
	\evalInt{\sj}{#3}
	
	\eval{\uy}{\getYofR{\lev-0.9}-0.05}

	\eval{\mx}{\getXofR{\si}{\j-0.5}}
		
	\node[scale=0.8, anchor=center] at (\mx,\uy) {\ifthenelse{\isodd{\n}}{\bf $\bot$}{\bf $\top$}};
}

\newcommand{\showBid}[2]{
	\evalInt{\lev}{#1}
	\evalInt{\n}{#2}
	\evalInt{\no}{\n+1}
	\evalInt{\nt}{\n+2}
	\showLabel{$\Bid_{\n}$}

	\ifthenelse{\n < \maxD}{
		\foreach \i in {\no,...,\maxD} {
			\foreach \j in {1,...,\zQtity} {
				\showCopy{\i}{\j}{\i}{\j}
			}
		}
	}{}
	
	\ifthenelse{\nt > \maxD}{
		\foreach \i in {\minD,...,\n} {
			\foreach \j in {1,...,\zQtity} {
				\showDots{\n}{\i}{\j}
			}
		}
	}{
		\foreach \i in {\minD,...,\n} {
			\foreach \j in {1,...,\zQtity} {
				\showCopy{\nt}{\j}{\i}{\j}
			}
		}
	}
}

\newcommand{\showCut}[2]{
	\evalInt{\lev}{#1}
	\evalInt{\n}{#2}
	\evalInt{\nn}{\n+1}

	\showLabel{$\Cut_{\n}$}

	\ifthenelse{\n < \maxD}{
		\foreach \i in {\nn,...,\maxD} {
			\foreach \j in {1,...,\zQtity} {
				\showCopy{\i}{\j}{\i}{\j}
			}
		}
	}{}

	\foreach \i in {\minD,...,\n} {
		\foreach \j in {1,...,\zQtity} {
			\showCopy{\n-1}{\j}{\i}{\j}
		}
	}
}

\newcommand{\showTran}[2]{
	\evalInt{\i}{#1}
	\evalInt{\j}{#2}

	\eval{\trade}{0.15}

	\eval{\lu}{\getXofR{\i}{\trade}}
	\eval{\ru}{\getXofR{\i}{\zQtity-\trade}}

	\eval{\lb}{\getXofR{\i}{\trade}}
	\eval{\rb}{\getXofR{\j}{\zQtity-\trade}}

	\eval{\uy}{\getYofR{\lev-0.8}}
	\eval{\by}{\getYofR{\lev-0.2}}

	\ifthenelse{\i=0}{\path[draw=black!60, pattern=north east lines, pattern color=black!40]
		(\lu,\uy) -- (\ru,\uy) -- (\rb,\by) -- (\lb,\by) -- cycle;}{
	\ifthenelse{\i=1}{\path[draw=black!60, pattern=vertical lines  , pattern color=black!40]
		(\lu,\uy) -- (\ru,\uy) -- (\rb,\by) -- (\lb,\by) -- cycle;}{
	\ifthenelse{\i=2}{\path[draw=black!60, pattern=north west lines, pattern color=black!40]
		(\lu,\uy) -- (\ru,\uy) -- (\rb,\by) -- (\lb,\by) -- cycle;}{
	\ifthenelse{\i=3}{\path[draw=black!60, pattern=horizontal lines, pattern color=black!40]
		(\lu,\uy) -- (\ru,\uy) -- (\rb,\by) -- (\lb,\by) -- cycle;}{
					  \path[draw=black!60, pattern=north east lines, pattern color=black!40]
		(\lu,\uy) -- (\ru,\uy) -- (\rb,\by) -- (\lb,\by) -- cycle;
	}%
	}%
	}%
	}
}

\newcommand{\showDelta}[1]{
	\evalInt{\lev}{#1}

	\showLabel{$\Delta$}

	\foreach \i in {\minD,...,\maxD} {
		\showTran{\i}{\maxD}
	}
}

\newcommand{\showLim}[4]{		
	\evalInt{\ulev}{#2}
	\evalInt{\dlev}{#1}

	\eval{\lx}{\getXofR{\maxD}{\zQtity}+0.6+0.4*(#3)}
	
	\eval{\uy}{\getYofR{\ulev}+0.2}
	\eval{\dy}{\getYofR{\dlev}-0.2}
	\eval{\cy}{\getYofR{(\ulev+\dlev)*0.5}}

	\draw (\lx-0.1,\uy) -- (\lx,\uy) -- (\lx, \dy) -- (\lx-0.1,\dy);

	\ifthenelse{\equal{#4}{U}}{
		\draw[-{Latex}] (\lx+0.2,\cy-0.3*\yScale) -- (\lx+0.2,\cy+0.3*\yScale);
	}{
		\draw[-{Latex}] (\lx+0.2,\cy+0.3*\yScale) -- (\lx+0.2,\cy-0.3*\yScale);
	}
}	

\newcommand{\showLimU}[3]{
	\showLim{#1}{#2}{#3}{U}
}

\newcommand{\showLimD}[3]{
	\showLim{#1}{#2}{#3}{D}
}

\newcommand{\showForm}[4]{
	\evalInt{\ulev}{#2}
	\evalInt{\dlev}{#1}

	\eval{\lx}{\getXofR{\minD}{0}-1.6-0.6*(#3)}
	
	\eval{\uy}{\getYofR{\ulev}+0.2}
	\eval{\dy}{\getYofR{\dlev}-0.2}
	\eval{\cy}{\getYofR{(\ulev+\dlev)*0.5}}

	\draw (\lx+0.1,\uy) -- (\lx,\uy) -- (\lx, \dy) -- (\lx+0.1,\dy);

	\node[anchor=mid east] at (\lx+0.09,\cy) {#4};
}

\eval{\ssX}{+1.0}
\eval{\ssY}{+1.0}
\eval{\ddX}{+1.0}
\eval{\ddY}{-1.0}

\tikzstyle{form} = [anchor=mid, align=center]
\tikzstyle{semm} = [form, scale=1.2]
\tikzstyle{edde} = []
\tikzstyle{dasm} = [loosely dashed, line cap=round, line width=0.3mm]

\newcommand{\showTreeInd}[1] {
	\evalInt{\i}{#1}
	\evalInt{\j}{#1 - 1}
	\eval{\bx}{-3 * \i * \ddX}
	\eval{\by}{-3 * \i * \ddY}

	\global\gdef\limI{\ifthenelse{\equal{\intcalcMod{\i}{2}}{0}}%
        {\limD}%
        {\limU}}

	\node[semm](r\i)
		at (\bx + 0.00 * \ddX, \by + 0.00 * \ddY)
		{\textbf{;}};

	\node[form](b\i)
		at (\bx + 0.00 * \ddX, \by + 1.50 * \ddY)
		{$\Bid_\i$};

	\node[form](a\i)
		at (\bx + 0.75 * \ddX, \by + 0.75 * \ddY)
		{$\limI{}$};

	\node[semm](s\i)
		at (\bx + 1.50 * \ddX, \by + 1.50 * \ddY)
		{\textbf{;}};

	\node[form](d\i)
		at (\bx + 3.00 * \ddX, \by + 1.50 * \ddY)
		{$\limI{\Delta}$};

	\node[form](c\i)
		at (\bx + 1.50 * \ddX, \by + 3.00 * \ddY)
		{$\Cut_\i$};

	\draw[edde] (r\i) -- (b\i);
	\draw[edde] (r\i) -- (a\i);
	\draw[edde] (a\i) -- (s\i);
	\draw[edde] (s\i) -- (d\i);
	\draw[edde] (s\i) -- (r\j);
	\draw[edde] (s\i) -- (c\i);

	\draw[dasm]
		(\bx - 0.85, -2.5) --
		(\bx - 0.85, \by) arc (180:90:0.85) --
		(+2.5, \by + 0.85);

	\node[form, scale=1.2] at (+2.25, \by + 0.5) {$\Phi_{\i}$};
}

\newcommand{\showTreeBas}[1] {
	\evalInt{\i}{#1}
	\eval{\bx}{0}
	\eval{\by}{0}

	\global\gdef\limI{\ifthenelse{\equal{\intcalcMod{\i}{2}}{0}}%
        {\limD}%
        {\limU}}

	\node[semm](r\i)
		at (\bx + 0.0 * \ddX, \by + 0.0 * \ddY)
		{\textbf{;}};

	\node[form](b\i)
		at (\bx + 0.0 * \ddX, \by + 1.5 * \ddY)
		{$\Bid_\i$};

	\node[form](d\i)
		at (\bx + 1.5 * \ddX, \by + 0.0 * \ddY)
		{$\limI{\Delta}$};

	\draw[edde] (r\i) -- (b\i);
	\draw[edde] (r\i) -- (d\i);

	\draw[dasm]
		(\bx - 0.85, -2.5) --
		(\bx - 0.85, \by) arc (180:90:0.85) --
		(+2.5, \by + 0.85);

	\node[form, scale=1.2] at (+2.25, \by + 0.5) {$\Phi_{\i}$};
}

\title{On the Computability of Measures of Regular Sets of Infinite Trees%
\thanks{Work supported by the National Science Centre, Poland (grant no.\@ 2021/\allowbreak41/\allowbreak B/\allowbreak ST6/\allowbreak03914).}}
\author{Damian Niwiński \and Paweł Parys \and Michał Skrzypczak}
\DeclareSymbolFont{yhlargesymbols}{OMX}{yhex}{m}{n}
\DeclareMathAccent{\yhwidehat}{\mathord}{yhlargesymbols}{"62}

\renewcommand{\widehat}{\yhwidehat}

\newcommand{\trees}{\ident{Tr}}

\newcommand{\ar}{\ident{ar}}

\newcommand{\vtau}{\bar{\tau}}
\newcommand{\vsigma}{\bar{\sigma}}

\newcommand{\Fix}{\mathrm{Fix}}
\newcommand{\Bid}{\ident{Bid}}
\newcommand{\Cut}{\ident{Cut}}

\newcommand{\limU}[1]{{#1}{\uparrow}}
\newcommand{\limD}[1]{{#1}{\downarrow}}

\newcommand{\dcolon}{\!::}
\newcommand{\comp}{;}
\newcommand{\Pcc}{\powerset_\mathsf{cc}}

\newcommand\set[1]{\{#1\}}
\renewcommand\phi\varphi
\renewcommand\epsilon\varepsilon

\newcommand{\parto}{\rightharpoonup}

\Crefname{lemma}{Lemma}{Lemmata}
\Crefname{fact}{Fact}{Facts}
\Crefname{claim}{Claim}{Claims}
\Crefname{equation}{Equation}{Equations}
\Crefname{inequality}{Inequality}{Inequalities}\creflabelformat{inequality}{\begingroup(#2#1#3)\endgroup}
\Crefname{expression}{Expression}{Expressions}\creflabelformat{expression}{\begingroup(#2#1#3)\endgroup}

\begin{document}

\maketitle

\begin{abstract}
	The Rabin tree theorem yields an~algorithm to solve the satisfiability problem for monadic second\=/order logic over infinite trees.
	Here we solve the probabilistic variant of this problem.
	Namely, we show how to compute the probability that a~randomly chosen tree satisfies a~given formula.
	We additionally show that this probability is an~algebraic number.
	This closes a~line of research where similar results were shown for formalisms weaker than the full monadic second\=/order logic.
\end{abstract}

\maketitle

\section{Introduction}

The Rabin tree theorem had an~essential impact on formal methods in computer science, thanks to both the large applicability of the result and the strength of techniques introduced in the proof~\cite{rabin_s2s}.
Basically, it reduces the satisfiability problem for monadic second\=/order logic (MSO) over the infinite binary tree
to the non\=/emptiness problem of finite automata running over labelled infinite trees, which yields decidability.
The acceptance criterion used in the automata generalises previous ideas of B\"uchi~\cite{buchi_decision}.
The result gave rise to the whole area of automata-based methods in verification, whose realm of applications goes far beyond the original theory of the tree,
see e.g.,~\cite{thomas_languages,2008thomas,Demri2016,Pin2021} for surveys.

The potential of automata on infinite trees for probabilistic questions in verification has been noted by several authors~\cite{chen_model_checking,mio_branching_games,CarayolHS14}.
The space of trees can be naturally equipped with a~Cantor\=/like topology as well as the probability measure,
where we assume that each letter is chosen independently at random (called \emph{coin\=/flipping measure}).
The following question can be viewed as a~probabilistic counterpart of the problem solved by Rabin.

\begin{problem}\label{problem}
	Given a~representation of a~regular tree language $L$, find the probability that a~randomly generated tree belongs to $L$.
\end{problem}

An analogous question for regular languages of $\omega$\=/words can be traced back to the early work on verification of finite-state probabilistic systems initiated, among others, by Vardi~\cite{Vardi85}.
An efficient solution follows from the work by Courcoubetis and Yannakakis~\cite{courcou-yanna-acm} (see also~\cite{Vardi99-proba}).
It has been noted by several authors that, for the coin\=/flipping measure, the probability in question is rational (see, e.g., \cite{chatterjee_stochastic_parity,takeuti-rational}).

In the case of infinite trees, it is not even evident that the probability in question always exists, as a~regular tree language need not in general be Borel~\cite{niwinski_gap}.
However, it was established by Gogacz et al.~\cite{michalewski_measure_final} that all regular languages of trees are measurable.

Chen et al.~\cite{chen_model_checking} addressed the question for trees in~the case where the tree language~$L$ is~recognised by~a~deterministic top\=/down parity automaton
and the measure is~induced by~a~stochastic branching process, which then makes also a~part of~the input data.
Their algorithm compares the probability with any~given rational number in~polynomial space and with $0$ or~$1$ in~polynomial time.

Michalewski and Mio~\cite{michalewski_comp_measure} gave an~algorithm for a~larger class of languages~$L$ given by~so\=/called \emph{game automata}.
Their algorithm (given for the coin\=/flipping measure) reduces the problem to~computing the value of~a~Markov branching play, and uses Tarski's decision procedure for the theory of~reals.
These authors also discovered that the measure of~a~regular tree language can be~irrational.
Later on, Niwiński et al.~\cite{niwinski_measures_wmso} solved the problem for weak alternating automata, which capture precisely weak MSO logic (with quantifiers restricted to finite sets).

In this paper we ultimately solve \cref{problem} for nondeterministic parity tree automata, which capture the whole MSO.

\begin{theorem}\label{thm:main-theorem}
	There is an~algorithm that inputs a~nondeterministic parity tree automaton $\CA$ and computes the coin\=/flipping measure of the language of trees recognised by $\CA$.
	Moreover, this measure is an~algebraic number.
\end{theorem}

As a~consequence, given additionally a~rational (or just algebraic) number $q$, we can decide if the resulting measure is equal, smaller, or greater than $q$.

Our algorithm is able to produce a~representation of the above measure in three\=/fold exponential time.
The decision problems with a~fixed rational number $q$ can be solved in two\=/fold exponential space.
These complexity upper bounds coincide with the ones provided by Przybyłko and Skrzypczak~\cite{przybylko_simple_sets} in the case of safety automata.

Instead of considering the coin\=/flipping measure over a full binary tree,
one may also be interested in trees generated by a given stochastic branching process~\cite{chen_model_checking,niwinski_measures_wmso}.
It was already observed by Niwiński et al.~\cite[Section~7.1]{niwinski_measures_wmso} that such a situation reduces to the basic case of the coin\=/flipping measure:
essentially, instead of an automaton $\CA$ reading a tree generated by a branching process~$\CP$,
we may consider an automaton being a product of $\CA$ and $\CP$, reading a binary tree chosen randomly according to the coin-flipping measure.
We thus have the following corollary, with the same complexity bounds as for \cref{thm:main-theorem}.

\begin{corollary}\label{cor:branching}
	There is an~algorithm that inputs a stochastic branching process $\CP$ and a~nondeterministic parity tree automaton $\CA$
	and computes the measure induced by $\CP$ of the language of trees recognised by $\CA$.
	Again, this measure is an~algebraic number.
\end{corollary}

It is worth to note that natural stronger probabilistic variants of the satisfiability problem are known to be undecidable.
For example, both the non\=/emptiness problem~\cite{Paz1971} and the value\=/$1$ problem~\cite{gimbert_probabilistic} for probabilistic finite automata, even for finite words, are undecidable.
Moreover, Baier et al.~\cite{BaierBG08} showed that it is undecidable whether there exists an $\omega$\=/word that is accepted by a~given probabilistic B\"uchi automaton with a~positive probability.
So presumably we are close to the frontiers of decidability here.

One should note that our result does not entail the Rabin tree theorem---the satisfiability problem does not directly reduce to its probabilistic variant.
Also, the core of the construction of Rabin was a~translation between an~MSO formula and an~automaton, whereas we begin our construction already with a~given automaton.

The main difficulty while considering probabilistic questions about automata lies in nondeterminism, which breaks a~correspondence between input objects and runs of the automaton.
It is known that in the case of infinite trees nondeterminism is inherent.
As a~remedy, we take the $\mu$\=/calculus perspective on tree automata.
Indeed, it is well-known that $\mu$\=/calculus (even in a~restricted form) interpreted over trees captures the full power of tree automata~\cite{niwinski88}
(the opposite direction follows easily from the Rabin tree theorem).

This helps to a~large degree, but not fully.
Having two random events $A$, $B$ we can speak about their union $A\cup B$ and intersection $A\cap B$,
but knowing only the probabilities $\prob(A)$, $\prob(B)$ we cannot say what is the probability of $\prob(A{\cup}B)$ or $\prob(A{\cap} B)$.
A~similar problem arises with other functions of arity greater than $1$.
In order to overcome this difficulty, we introduce a~new formalism: \emph{unary $\mu$\=/calculus}.
We believe that this formalism is of independent interest, and can potentially find other applications.

This paper extends a conference version published at LICS 2023~\cite{probabilistic-rabin}.

\section{Basic concepts}\label{sec:basic-concepts}

\subsection*{Trees and automata}

A \emph{tree} over an \emph{alphabet} $\Sigma$ is a~function $t\from\set{\dL,\dR}^\ast\to \Sigma$.
Here $\dL,\dR$ represent the directions leading from a \emph{node} $v \in \set{\dL,\dR}^\ast$ to its children $v\dL$ and $v\dR$.
The empty word $\epsilon$ is the \emph{root} of the tree.
A \emph{subtree} of a~tree $t$ starting in a~node $v$ is a~tree denoted $t.v$ defined by $t.v(w)=t(vw)$, for all $w\in\set{\dL,\dR}^\ast$.
The set of all trees over $\Sigma$ is denoted $\trees_\Sigma$.

A \emph{nondeterministic $\max$\=/parity tree automaton} is defined to be $\CA=\langle \Sigma,\allowbreak Q,\allowbreak q_I,\allowbreak\gamma,\allowbreak \Omega\rangle$, where
$\Sigma$ is a~finite input alphabet,
$Q$ a~finite set of \emph{states},
$q_I\in Q$ an \emph{initial state},
$\gamma\subseteq Q\times \Sigma\times Q\times Q$ a \emph{transition} relation,
and $\Omega\from Q \to \set{0,\ldots, d{-}1}$ assigns a \emph{priority} to each state.%
	\footnote{For the sake of readability, we decided to change the acceptance condition from $\min$\=/parity to $\max$\=/parity,
	when comparing to the conference version of this paper~\cite{probabilistic-rabin}.}
We assume without loss of generality that $d>0$ is an~even number.
A \emph{run} of such an~automaton $\CA$ over a~tree $t\in\trees_\Sigma$ is a~function $\rho\from\set{\dL,\dR}^\ast\to Q$ that assigns a~state to every node of $t$,
such that $(\rho(v),t(v),\rho(v\dL),\rho(v\dR))\in\gamma$ for all $v \in\set{\dL,\dR}^\ast$.
A run $\rho$ is \emph{accepting} if, for any branch $w\in\set{\dL,\dR}^\omega$, the maximal priority occurring infinitely often is even, that is,
$\limsup_{n \to \infty} \Omega\big(\rho(w_n)\big)$ is even, where $w_n$ is the prefix of $w$ of length~$n$.
A tree is \emph{accepted} by the automaton $\CA$ if it admits an~accepting run $\rho$ starting from the initial state, that is, $\rho(\epsilon)=q_I$.
The set of trees accepted by $\CA$ is denoted $\lang(\CA)$.

\subsection*{Orders and fixed points}

Let $(X,{\leq})$ be a~partial order.
The least upper bound of a~subset $A \subseteq X$ (if exists) is also called the \emph{supremum} of $A$ and denoted by $\sup A$,
and the greatest lower bound is the \emph{infimum} denoted by $\inf A$.
A~\emph{complete lattice} is a~partial order such that suprema and infima exist for all subsets.
In particular $\sup X = \top$
and $\inf X = \bot$ are the greatest and least elements of $X$.
The supremum and infimum of a~two-element set $\{a,b\}$ are denoted $a{\lor}b$ and $a{\land}b$, respectively.

The celebrated Knaster\=/Tarski theorem states that if $(X,{\leq})$ is a~complete lattice
then any monotone (i.e., order\=/preserving) function $f \from X \to X$ has the least ($\mu$) and the greatest ($\nu$) fixed points satisfying
\begin{alignat*}{6}
	\mu x.\,f(x) &= \inf\, &&\big\{ a &&\mid f(a) \leq a &&\big\}, \\
	\nu x.\,f(x) &= \sup\, &&\big\{ a &&\mid a \leq f(a) &&\big\}.
\end{alignat*}
For a~monotone function of several arguments, we can apply fixed-point operators successively.
For example, $\nu y.\,\mu x.\,g(x,y)$ is the greatest fixed point of the mapping $y \mapsto \mu x.\,g(x,y)$ (it is easy to observe that this mapping is also monotone).

A \emph{chain} in a~partial order $(X,{\leq})$ is any non\=/empty subset $C\subseteq X$ linearly ordered by ${\leq}$.
A subset $A\subseteq X$ is called a~\emph{chain\=/complete subset of $X$} if the supremum and the infimum of every chain $C\subseteq A$ exists and remains in $A$.
We write $\Pcc(X)$ for the set of all chain\=/complete subsets of $X$.

We say that a~monotone function $f\from X\to Y$ between two complete lattices, $(X,{\leq})$ and $(Y,{\leq})$, is \emph{chain\=/continuous}
if for every chain $C\subseteq X$ it holds $\sup(f(C))=f(\sup C)$ and $\inf(f(C))=f(\inf C)$.

Whenever we consider a~space of the form $X^n$ for a~partial order
$(X,{\leq})$ and a~natural number $n\geq 0$, we assume that $X^n$ is ordered coordinate\=/wise,
that is, $(x_0,\ldots,x_{n-1})\leq(y_0,\ldots,y_{n-1})$ if $x_i\leq y_i$ for all $i\in\set{0,\ldots,n{-}1}$.

\subsection*{Measures}

The set $\trees_\Sigma$ of trees over a~finite alphabet $\Sigma$ can be equipped with a~topology generated by~a~basis consisting of all the sets $U_f$,
where ${f}\from {\dom(f)} \to {\Sigma }$ is a~function with a~finite domain $\dom(f)\subset\{\dL,\dR\}^*$, and $U_f$ consists of all trees $t$ that coincide with $f$ on $\dom(f)$.
If $\Sigma$ has at least $2$ elements then this topology is homeomorphic to the Cantor discontinuum $\{0,1\}^\omega$ (see, e.g.,~\cite{perrin_pin_words}).

The set of trees can be further equipped with the coin\=/flipping measure, where a~letter in $\Sigma$ is chosen uniformly at random independently at each node.
This is the standard Lebesgue measure on the product space defined on the basis by $\prob\left(U_f\right) = \big|\Sigma \big|^{-| \dom(f)|}$.

We use the well\=/known facts that the measure $\prob$ is \emph{Borel} (i.e.,~every Borel set is $\prob$\=/measurable)
and that it is \emph{complete} (i.e.,~every subset of a~set of $\prob$\=/measure $0$ is measurable and of $\prob$\=/measure $0$).

\section{Fixed-point approach}

Our proof relies on fixed\=/point definitions of regular tree languages in the powerset lattice of $\trees_\Sigma$.
Such a definition can be derived from a~parity automaton (either nondeterministic or alternating) recognising the given language~\cite{niwinski88,niwinski_mu_calc_index}.
The operations involved are $a(\cdot,\cdot)$, for each $a\in\Sigma$, the binary set union $\cup$ and intersection $\cap$, and the least ($\mu$) and greatest ($\nu$) fixed\=/point operators.
The aforementioned operation $a(L_\dL,L_\dR)$ sends the languages $L_\dL$ and $L_\dR$ to the set of trees $t$, such that $t.\dL\in L_\dL$, $t.\dR\in L_\dR$, and $t(\epsilon)=a$.
For example, an~expression $\mu x.\,\nu y.\,a(x,x)\cup b(y,y)$ defines the set of trees over the alphabet $\{a,b\}$ where, on each branch, $a$ occurs only finitely often.
In general, it is convenient to use \emph{vectorial} fixed\=/point expressions of the form
\begin{align}\label[expression]{punkt-staly}
	\mu x_{d-1}.\,\nu x_{d-2}\ldots\nu x_2.\,\mu x_1.\,\nu x_0.\,F(x_0,x_1,\ldots,x_{d-1}),
\end{align}
where the vectorial variables $x_i=(x_{i,1},\ldots,x_{i,k})$, for $i\in\set{0,1,\ldots,d{-}1}$,
range over \emph{$k$\=/tuples} of subsets of $\trees_\Sigma$, and $F$ is a $k$\=/tuple of terms constructed from the basic operations mentioned above.
If we start with a~parity tree automaton, the parameter $k$ corresponds to the number of states, and the parameter~$d$ to the number of priorities (see, e.g.,~\cite{niwinski_rudiments}).
More specifically, if the $j$\=/th state has priority $i$, it becomes a~variable $x_{i,j}$.
Moreover, if the automaton is nondeterministic then no intersection $\cap$ is used in the expression.

\subsection{Non-emptiness for nondeterministic automata}

In the case of nondeterministic automata, this fixed\=/point characterisation leads to a simple, albeit not very efficient,
algorithm to decide if a~regular tree language is non\=/empty~\cite[Section~4.2]{niwinski_mu_calc_index}.

The idea is to take the $k$\=/th power of the Boolean lattice $\{0,1\}$, and evaluate \cref{punkt-staly} there,
with the operation $a(x,y)$ interpreted as $x \land y$ (for any symbol $a$) and $\cup$ as $\lor$.
Now the set $L_q$ of trees accepted from the state $q$ of the automaton is non\=/empty if and only if the respective component equals 1
(i.e., the component corresponding to the state $q$ in \cref{punkt-staly} evaluated in $\{0,1\}^k$).
The argument relies on good properties of a~mapping $h\colon \powerset(\trees_\Sigma) \to \{0,1\}$ that sends all non\=/empty sets to $1$ and the empty set to $0$.
We have $a\big(h(L_\dL),h(L_\dR)\big)=h\big(a(L_\dL,L_\dR)\big)$ and $h(L_1)\lor h(L_2)=h\big(L_1\cup L_2\big)$.
It can be further shown that $h$ preserves the fixed points in consideration as well;
here the argument relies on unravelling the involved expression $F$ into a~tree labelled by an~accepting run of the automaton.

However, an analogous construction fails when an automaton is alternating.
Indeed, the aforementioned mapping $h$ does not preserve
intersection $\cap$ if we interpret it as $\wedge$.
Moreover, the mapping $h$ is not \emph{chain\=/continuous}:
there can be a~sequence of non\=/empty sets $L_0\supseteq L_1\supseteq L_2\supseteq\ldots$ with $\bigcap_{n} L_n=\emptyset$, which means that $\forall n.\ h(L_n)=1$ and $h\big(\bigcap_{n} L_n\big)=0$.

\subsection{Unary \texorpdfstring{$\mu$\=/}{mu-}calculus}

Is it possible to apply a~fixed-point approach when, instead of non\=/emptiness, we are interested in the measure of a set?

On a~high level, we would like to find some measure\=/theoretic analogue of the Boolean lattice $\{0,1\}$, so that an~appropriate interpretation of \cref{punkt-staly} would yield the researched measure.
A natural candidate is the space $[0,1]$ along with the monotone mapping $\prob: L\mapsto \prob(L)$, which assigns to each (measurable) set $L$ the probability of $L$.
In analogy to the mapping $h$ from the previous section, the operation $a(\cdot,\cdot)$ can be interpreted over measures of languages.
Indeed, $\prob\big(a(L_\dL,L_\dR)\big)=|\Sigma|^{-1}\cdot\prob(L_\dL)\cdot \prob(L_\dR)$.
Moreover, continuity of measure implies that the mapping $\prob$ is chain\=/continuous: $\prob\big(\bigcap_{n} L_n\big)=\lim_{n\to\infty}\prob(L_n)$ whenever the languages $L_n$ form a~chain.

However, the measures of neither the union $\cup$ nor the intersection $\cap$ of two languages are determined by the measures of the languages themselves,
so it is not possible to interpret $\cup$ nor $\cap$ here.
This reflects a~more general phenomenon: the knowledge of distributions of two random variables does not determine their joint distribution.
Thus, when working with measures, we should avoid using functions with more than one argument.

To overcome this obstacle, we introduce the main technical invention of our paper: unary $\mu$\=/calculus.

In the classical $\mu$\=/calculus, an~application of a fixed\=/point operator, say $\mu$, to a mapping $F\colon X\to X$ returns an~element $\mu x.\,F(x)$ of $X$, and if applied to a~distinguished argument of a $(k{+}1)$\=/ary function, it returns a~$k$\=/ary function, for example $\mu x_0.\,G(x_0,x_1,\ldots,x_k)$, and so on.
In unary $\mu$\=/calculus, an~application of a~fixed\=/point operator
to a~unary monotone function returns another unary monotone function.
More precisely, when applied to a~function $F\colon X \to X$, it returns a (partial) function $\limU{F}\colon X \parto X$,
that sends an element $x\in X$ to the least fixed point of $F$ \emph{above} $x$ (if exists).
Similarly, $\limD{F}\colon X \parto X$ sends an element $x\in X$ to the greatest fixed point of $F$ \emph{below} $x$ (if exists).
Such fixed points need not always exist even if $X$ is a~complete lattice: although the respective subsets are complete sublattices, the function $F$ need not preserve them.
However, when $F(x)\geq x$ then $\limU{F}(x)$ exists and can be obtained by taking a~limit of an~ascending chain of approximations (as in Knaster\=/Tarski theorem);
similarly for $\limD{F}(x)$, whenever $F(x)\leq x$.

An~essential step of our construction is to define a~formula of unary $\mu$\=/calculus which, in the domain of tree languages, is equivalent to \cref{punkt-staly}.
Since transitions of the considered automaton are already black\=/boxed into a~single function $F$, the exact type of the automaton does not matter:
the construction works in the same manner for both nondeterministic and alternating automata.
We further show that the derived formula interpreted in an appropriate probabilistic domain yields the researched measure.

\subsection{Computing measure}

We now provide more details on our formula and its interpretation.

Viewing $\trees_\Sigma$ as a~sample space, a~measurable tree language $L \subseteq \trees_\Sigma$ can be identified with a~random variable, say $\widetilde{L}\from\trees_\Sigma\to\{0,1\}$,
which is the characteristic function of $L$; we search for $\prob(\{t\mid\widetilde{L}(t)=1\})$.
Given an automaton, it is apparent that we need to consider the product variable
\begin{align}\label[expression]{product-var}
	\big(\widetilde{L_q}\big){}_{q \in Q},
\end{align}
where $L_q$ is the set of trees accepted from the state $q$.
Note that the variables $\widetilde{L_q}$ are, in general, not independent (even pairwise).
In particular, the probability that a~random tree $t$ is accepted both from a~state $p$ and from a~state $q$ is not determined by the probabilities of these two events taken separately.
Therefore, if we would like to define the probability distribution of the variables $\widetilde{L_q}$ as a nested fixed point,
we cannot rely on a~direct product construction as in the non\=/emptiness problem (where we would just store a tuple of probabilities, one for each variable $\widetilde{L_q}$).

To proceed further, we find it convenient to present the product variable from \cref{product-var}
as a single random variable $\tau_{\CA} \colon \trees_\Sigma\to \powerset(Q)$ that sends a tree $t$ to the set $\{q\mid t\in L_q\}$, that is,
\begin{align}\label{profil-krata}
	\tau_{\CA} (t) = \{ q\mid \widetilde{L_q} (t) = 1 \}.
\end{align}
Clearly, the probability distribution of $\tau_{\CA}$ yields the measure of the language recognised by the automaton,
which amounts to the probability of the event $q_I \in \tau_{\CA}(t)$ (where $q_I$ is the initial state).
We aim for a formula of unary $\mu$\=/calculus that defines the probability distribution of the variable~$\tau_{\CA}$.

To this end, we consider the space of all functions $\tau\colon\trees_\Sigma\to\powerset(Q)$, which we call \emph{profiles}.
Clearly this space is a complete lattice with the component\=/wise ordering induced by the inclusion ordering on $\powerset(Q)$.
If we evoke the identification from \cref{profil-krata}, this ordering coincides with the component\=/wise inclusion on $(\powerset(\trees_\Sigma))^{|Q|}$.
The fixed-point characterisation of regular tree languages~\cite{niwinski_rudiments} gives us a formula as in \cref{punkt-staly} that defines the profile $\tau_{\CA}$.
By a subtle analysis of the structure of nested fixed points, we manage to transform it into a formula of unary $\mu$\=/calculus.

The next step is to define an~appropriate structure on the space of probability distributions of measurable profiles.
As they take values in $\powerset(Q)$, we use here a~general construction of \emph{probabilistic powerdomain}
which settles an~appropriate ordering on the set of probability distributions on a~complete lattice $\CR$.
This structure is in general not a~lattice, but it is chain\=/continuous, which is sufficient for unary $\mu$\=/calculus.

Now, our formula of unary $\mu$\=/calculus can be reinterpreted in the probabilistic powerdomain over a~single lattice $\CR$,
which in our case is $\powerset(Q \times \{ 0,1,\ldots,d{-}1\})$ (i.e., each state is multiplied by all possible priorities).
The new interpretation yields the desired probability distribution of $\tau_{\CA}$ and the actual values can be expressed in the Tarski's theory of reals.

\subsection*{Overview of the proof}

\cref{sec:spaces} introduces the space of profiles and recalls the fixed-point characterisation of regular tree languages in terms of profiles~\cite{niwinski_rudiments}.
\cref{sec:unary} introduces unary $\mu$\=/calculus along with a typing system that allows us to verify if an expression of this calculus is well-defined.
\crefrange{sec:basic-functions}{sec:realisations} are devoted to showing that the generic formula from \cref{punkt-staly} can be translated into unary $\mu$\=/calculus.
This is a general result on fixed-point definitions and the proof is carried out in an~abstract setting, not restricted to the interpretation in the domain of trees.
We come back to trees in \cref{sec:measurability} and verify that all the sets of interest for us are measurable.
This concerns not only the regular tree languages whose measurability is known~\cite{michalewski_measure_final}, but also the sets obtained in the iterative approximation of fixed points.
Knowing this, we can switch our perspective from profiles to probability distributions.
In \cref{sec:distributions} we accomplish the main technical result of our paper
showing that the purchased distribution of the random variable $\tau_{\CA}$ associated with the automaton
can be obtained by reinterpreting our formula of unary $\mu$\=/calculus in the probabilistic powerdomain in consideration.
The proof of \cref{thm:main-theorem} is finished in \cref{sec:tarski},
where we show that formulae of unary $\mu$\=/calculus interpreted in the probabilistic powerdomain can be expressed using Tarski's first-order theory over reals.
Therefore, decidability of this theory implies decidability of our question expressed by a~unary $\mu$\=/calculus formula.
Finally, in \cref{sec:branching} we treat the case of trees generated by a stochastic branching process, proving \cref{cor:branching}.

\section{Profiles}\label{sec:spaces}

As already mentioned, we work with two spaces: the space of profiles, and the space of probability distributions over profiles and vectors thereof.
For the sake of simplicity, we work with nondeterministic automata, see \cref{rem:alternating} for the case of alternating automata.
Thus, fix a~nondeterministic parity~automaton $\CA=\langle \Sigma,Q,q_I,\gamma,\Omega\rangle$.
Recall that we assume that the maximal priority $d$ is even.

The space of \emph{profiles} is $\CV=\big(\trees_\Sigma\to\powerset(Q)\big)$.
A~single profile $\tau\in\CV$ is thus a~function that assigns a~set of states to each tree $t$ in $\trees_\Sigma$.
Ultimately, we are interested in the unique profile $\tau_\CA\in\CV$ that to each tree $t$
assigns the set of those states from which there exists an~accepting run of $\CA$ on $t$.
In the meantime, we also consider other profiles.

We fix the order on $\CV$ defined for $\tau,\tau'\in\CV$ by $\tau\leq \tau'$ if $\tau(t)\subseteq \tau'(t)$ for all trees $t\in\trees_\Sigma$.
It is easy to see that $(\CV,{\leq})$ is a~complete lattice: a~supremum (infimum) can be obtained by taking the union (intersection) for every tree $t\in\trees_\Sigma$.

Let us now define a~function $\delta\from \CV^d\to\CV$ corresponding to the transition function of the automaton.
For each $a\in\Sigma$ we define $\delta_a\from\powerset(Q)\times\powerset(Q)\to\powerset(Q)$: for $K^\dL,K^\dR\in\powerset(Q)$ we take
\begin{align*}
	\delta_a(K^\dL,K^\dR)\eqdef\big\{q\in
	Q\mid \exists(q,a ,q^\dL,q^\dR)\in\gamma.\,q^\dL\in K^\dL\land q^\dR\in K^\dR\big\}.
\end{align*}
Then, for $\vtau=(\tau_0,\ldots,\tau_{d-1})\in\CV^d$ and $t\in\trees_\Sigma$ we define
\begin{align*}
	K^{\dL}_{t}(\vtau)&\eqdef\big\{q\in Q\mid q\in\tau_{\Omega(q)}(t.\dL)\big\},\\
	K^{\dR}_{t}(\vtau)&\eqdef\big\{q\in Q\mid q\in\tau_{\Omega(q)}(t.\dR)\big\},\\
	\delta(\vtau)(t)&\eqdef\delta_{t(\epsilon)}\Big(K^\dL_t(\vtau),K^\dR_t(\vtau)\Big).
\end{align*}
The intuition is: a state $q$ belongs to the profile $\delta(\vtau)$ in the root of a~hypothetical tree $t$
if after a~single transition from a~state $q$
the two successor states $q^\dL$ and $q^\dR$ belong respectively to the profiles $\tau_{\Omega(q^\dL)}$ and $\tau_{\Omega(q^\dR)}$ applied to the two subtrees of $t$ starting in the children of the root.
Note that $\delta$ ignores most of the input profiles.
Namely, for every state $q$ we check whether $q$ belongs to $\tau_i(t)$ only for $i=\Omega(q)$, while this information is meaningless for $i\neq\Omega(q)$.
In such a situation one may wonder if we can ``compress'' all the input profiles $\tau_0,\dots,\tau_{d-1}$ into a single one.
This would be problematic, however, because different components of the input $\tau_0,\dots,\tau_{d-1}$ should be treated differently with respect to fixed-point operators
(cf.\@ the formula from \cref{tau-automaton} below).

\begin{lemma}\label{lem:delta-mono-chain}
	The function $\delta\from \CV^d\to\CV$ is monotone and chain\=/continuous.
\end{lemma}

\begin{proof}
	Monotonicity of $\delta$ follows from monotonicity of $\delta_a$.
	To see that $\delta$ is chain\=/continuous, consider a~chain $C\subseteq\CV^d$,
	its supremum $\vsigma\in\CV^d$, and the supremum $\sigma'\in\CV$ of the chain $\set{\delta(\vtau)\mid\vtau\in C}\subseteq\CV$.
	We have to prove that $\delta(\vsigma)=\sigma'$.
	To this end, fix a~tree $t\in\trees_\Sigma$.
	Notice that the pair $\big(K^\dL_t(\vtau),K^\dR_t(\vtau)\big)$ belongs to the finite set $\powerset(Q)^2$, and that the functions $K^\dL_{t}$, $K^\dR_{t}$ are monotone.
	Because the order in $\CV$ is defined for each tree separately, we have that $\big(K^\dL_t(\vsigma), K^\dR_t(\vsigma)\big)$ is the supremum of the finite chain
	$\big\{(K^\dL_t(\vtau),K^\dR_t(\vtau))\mid\vtau\in C\big\}$;
	likewise $\sigma'(t)$ is the supremum of the finite chain $\big\{\delta(\vtau)(t)\mid\vtau\in C\big\}$.
	The supremum of a~finite chain is just its greatest element, so there exists some $\vtau\in C$ such that $K^\dL_t(\vsigma)=K^\dL_t(\vtau)$, $K^\dR_t(\vsigma)=K^\dR_t(\vtau)$, and $\sigma'(t)=\delta(\vtau)(t)$.
	Next, notice that $\delta(\vtau)(t)$ depends only on $K^\dL_t(\vtau)$ and $K^\dR_t(\vtau)$;
	in other words, $\big(K^\dL_t(\vsigma),K^\dR_t(\vsigma)\big)=\big(K^\dL_t(\vtau),K^\dR_t(\vtau)\big)$ implies $\delta(\vsigma)(t)=\delta(\vtau)(t)=\sigma'(t)$.
	We thus have the equality $\delta(\vsigma)(t)=\sigma'(t)$ for all
	$t\in\trees_\Sigma$, which simply means that $\delta(\vsigma)=\sigma'$,
	as required.

	The proof for the infimum of $C$ is analogous.
\end{proof}

Recall the special profile $\tau_\CA\in\CV$ corresponding to the automaton $\CA$ (note that $d$ is assumed to be even).

\begin{proposition}\label{tau-automaton}
	We have
	\begin{align*}
		\tau_\CA=\mu x_{d-1}.\,\nu x_{d-2}.\,\mu x_{d-3}.\,\nu x_{d-4}\ldots\mu x_{1}.\,\nu x_{0}.\,\delta(x_0,\ldots,x_{d-1}).
	\tag*{\qed}\end{align*}
\end{proposition}

This proposition expresses a~well\=/known equivalence between $\mu$\=/calculus and tree automata first observed by Niwiński~\cite{niwinski88}.
A proof based on parity games can be found for instance in a~book by Arnold and Niwiński~\cite[Proposition~7.3.5]{niwinski_rudiments},
showing that the two parity games corresponding to the automaton and to the formula, respectively, are equivalent.
The cited argument is done in the complete lattice $\powerset(\trees_\Sigma)$, but this domain corresponds to the space $\CV$
since a~profile $\tau\from\trees_\Sigma\to\powerset(Q)$ (with $Q=\{q_0,\ldots,q_{m-1}\}$, say)
can be identified with an $m$\=/vector of tree languages $(L_0,\ldots,L_{m-1})$, where $L_i=\{t\in\trees_\Sigma\mid q_i\in\tau(t)\}$.

Note that in this work we use functional $\mu$\=/calculus rather than the modal one.
An analogous formula for the modal $\mu$\=/calculus has been given by Janin and Walukiewicz~\cite{JaninW95}.

\subsection{Alternating automata}

An alternating tree automaton is defined analogously to a nondeterministic one, except that the transition relation $\gamma$ assigns to each state $q\in Q$ a~formula $\gamma(q)$ built using the syntax
\[ \theta ::= a(q_\dL,q_\dR) \mid \theta_1\land\theta_2 \mid \theta_1\lor\theta_2,\]
where $a\in\Sigma$ and $q_\dL,q_\dR\in Q$.
In other words, a~transition is a~positive Boolean combination of \emph{deterministic} transitions of the form $a(q_\dL,q_\dR)$.
A~nondeterministic automaton can be seen as a~special case of an~alternating one, where the formulae $\gamma(q)$ are disjunctions of deterministic transitions.

The semantics of alternating automata is given in terms of acceptance game, see, for instance, \cite{Kirsten-alter-auto}, \cite{Julian-Igor-hand}, or \cite{niwinski_rudiments}.
In particular, $\tau_\CA(t)$ can again be defined as the set of states from which $\CA$ accepts $t$.
Our goal is to provide an~alternative definition of the function $\delta_a\from \powerset(Q)\times\powerset(Q)\to\powerset(Q)$, which is now defined for an~alternating automaton.
The remaining definitions of $K^{\dL}_{t}(\vtau)$, $K^{\dR}_{t}(\vtau)$, and $\delta(\vtau)(t)$ remain the same.

To define $\delta_a(K^\dL,K^\dR)$ we proceed inductively over the structure of the formula $\gamma(q)$, defining auxiliary sets $K_\theta$ for all subformulae of $\gamma(q)$ as follows:
\begin{align*}
	K_{a(q^\dL,q^\dR)}&\eqdef\big\{q\in Q\mid q^\dL\in K^\dL \land q^\dR\in K^\dR\big\},\\
	K_{b(q^\dL,q^\dR)}&\eqdef\emptyset&\text{for $b\neq a$},\\
	K_{\theta_1\lor\theta_2}&\eqdef K_{\theta_1}\cup K_{\theta_2},\\
	K_{\theta_1\land\theta_2}&\eqdef K_{\theta_1}\cap K_{\theta_2}.
\end{align*}
Finally, we put $\delta_a(K^\dL,K^\dR)\eqdef K_{\gamma(q)}$.

The correctness of this construction relies on the following fact.

\begin{fact}
	Assume that $\tau_\CA(t.\dL)=K^\dL$, $\tau_\CA(t.\dR)=K^\dR$, and $a=t(\epsilon)$.
	Then $\tau_\CA(t)=\delta_a(K^\dL,K^\dR)$.
\qed\end{fact}

\begin{remark}\label{rem:alternating}
	Assume that the considered automaton is alternating and $\delta$ is defined based on $\delta_a$ as above.
	Then, \cref{lem:delta-mono-chain} and the equivalence from \cref{tau-automaton} hold.
	Moreover, the rest of our proof, relating $\tau_\CA$ to the constructed formula of unary $\mu$\=/calculus and then showing how to compute measures, remains also valid;
	it works for any monotone function $\delta_a$.
\end{remark}

\section{Unary \texorpdfstring{$\mu$\=/}{mu-}calculus}\label{sec:unary}

We now proceed to define the notion of unary $\mu$\=/calculus.
This is done in generality, for any partial order $(X,{\leq})$.

\subsection{Intensional definition}\label{ssec:intensional}

Consider a~partial order $(X,{\leq})$.
A \emph{partial function} $F\from X\parto X$ is a~function $F\from\dom(F)\to X$ for some $\dom(F)\subseteq X$.
Let $F,G\from X\parto X$ be monotone partial functions.
We write $F\comp G$ for their composition, defined by
\begin{align*}
	(F\comp G)(x)&\eqdef G(F(x))\hspace{9em}\mbox{with}\\
	\dom(F\comp G)&\eqdef \big\{x\in\dom(F)\mid F(x)\in \dom(G)\big\}.
\end{align*}
We also define $\Fix(F)$ as the set of fixed points of $F$, that is, $\Fix(F)\eqdef\set{x\in\dom(F)\mid F(x)=x}$.
Then, we define a~partial function $\limU{F}\from X\parto X$ taking
\begin{align*}
	\limU{F}(x) &\eqdef \min \big\{y\in \Fix(F)\mid y\geq x\big\}\hspace{9em}\text{with}\\
	\dom(\limU{F})&\eqdef\big\{x\in X\mid\exists y\in \Fix(F).\, \big(y\geq x\land
		\forall y'\in\Fix(F).\, (y'\geq x\Rightarrow y'\geq y)\big)\big\},
\end{align*}
That is, $x\in\dom(\limU{F})$ if there exists the least fixed point of $F$ that is greater or equal $x$, and then $\limU{F}(x)$ is this fixed point.
Note that the domain of $\limU{F}$ need not be contained in the domain of $F$.
Dually, we define
\begin{align*}
	\limD{F}(x) &\eqdef \max \big\{y\in \Fix(F)\mid y\leq x\big\},
\end{align*}
and $\dom(\limD{F})$ is defined accordingly.

We fix a~finite set of symbols $\mathfrak{F}$ that we call a~\emph{unary functional signature} (\emph{signature} for short).
A~partial order $(X,{\leq})$ is expanded to a~\emph{structure over $\mathfrak{F}$} if each symbol $H\in\mathfrak{F}$ is interpreted by a~monotone function $H_X\from X \to X$.
We call the functions $H_X$ the \emph{basic functions}.

A~formula of \emph{unary $\mu$\=/calculus} over a~signature $\mathfrak{F}$ is any term $F$ obtained using the grammar
\[ F ::= H \mid F_1\comp F_2 \mid \limU{F} \mid \limD{F},\]
where $H\in\mathfrak{F}$.

Assume that $(X,{\leq})$ is a~structure over $\mathfrak{F}$.
We now define, for each formula $F$ of unary $\mu$\=/calculus, a~partial function $F_X\from X\parto X$:
\begin{align*}
	(H)_X\eqdef H_X,\qquad
	(F\comp G)_X\eqdef F_X\comp G_X,\qquad
	(\limU{F})_X\eqdef\limU{(F_X)},\qquad
	(\limD{F})_X\eqdef\limD{(F_X)}.
\end{align*}

The following fact follows easily from the definition.

\begin{lemma}
	If $F$ is a~formula of unary $\mu$\=/calculus then $F_X\from X\parto X$ is a~monotone partial function.
\qed\end{lemma}

\subsection{Computing limits}

For an~ordinal $\eta>0$ consider a~sequence $(x_\zeta)_{\zeta<\eta}$ of elements of $X$.
We say that such a~sequence is \emph{increasing} if $x_\zeta\leq x_{\zeta'}$ whenever $\zeta<\zeta'<\eta$.
We define the \emph{limit} of such a~sequence, denoted $\lim_{\zeta<\eta}x_\zeta$, as the supremum of the chain $\set{x_\zeta\mid\zeta<\eta}$.
Analogously we define a \emph{decreasing} sequence, whose limit is $\inf\set{x_\zeta\mid\zeta<\eta}$.
A sequence is \emph{monotonic} if it is either increasing or decreasing.

\begin{definition}
	We say that a~partial function $F\from X\parto X$ is \emph{ascending on $A\subseteq X$} if:
	$F$ is monotone, $A\subseteq \dom(F)$, $A$ is a~chain\=/complete subset of $X$, and for each $x\in A$ we have $F(x)\in A$ and $F(x)\geq x$.

	$F$ is \emph{descending on $A$} if the analogue conditions are met with $F(x)\leq x$.
\end{definition}

Assume that $F\from X \parto X$ is either ascending or descending on $A\subseteq X$ and $x\in A$.
For an~ordinal $\eta$ define inductively $F^\eta(x)$ to be
\begin{itemize}
\item	$x$ if $\eta=0$,
\item	$F(F^{\eta-1}(x))$ if $\eta$ is a~successor ordinal, and
\item	$\lim_{\zeta<\eta}F^\zeta(x)$ if $\eta$ is a~limit ordinal.
\end{itemize}
This sequence is well-defined thanks to the assumption that $F$ is monotone and ascending (descending) on $A$ and $A$ is chain complete;
in particular it is increasing (decreasing, respectively) and all its elements remain in $A$.
Then, for some ordinal $\eta_\infty$ (at latest for the first ordinal of cardinality greater than $|X|$) the values must stabilise: $F^{\eta_\infty+1}(x) = F^{\eta_\infty}(x)$.
Define $F^\infty(x)\eqdef F^{\eta_\infty}(x)$.

\begin{lemma}\label{lem:lim-as-iteration}
	Assume that $F\from X\parto X$ is ascending on $A\subseteq X$.
	Then $A\subseteq \dom(\limU{F})$ and for every $x\in A$ we have $(\limU{F})(x)=F^\infty(x)\in\Fix(F)\cap A$.

	Analogously, if $F$ is descending on $A\subseteq X$ then $A\subseteq \dom(\limD{F})$ and for every $x\in A$ we have $(\limD{F})(x)=F^\infty(x)\in\Fix(F)\cap A$.
\end{lemma}

\begin{proof}
	We conduct the proof for $\limU{F}$; the case of $\limD{F}$ is symmetric.
	Fix some $x\in A$ and recall the sequence $F^\eta(x)$ used to define $F^\infty(x)$.
	We have $F\big(F^{\eta_\infty}(x)\big)=F^{\eta_\infty+1}(x)=F^{\eta_\infty}(x)$, so $F^{\eta_\infty}(x)\in\Fix(F)$.
	Moreover $F^{\eta_\infty}(x)\geq x$ because $F^{\eta_\infty}(x)$ belongs to an~increasing sequence starting at $x$.

	Take now any other fixed point $y'\in\Fix(F)$ such that $y'\geq x$.
	We can see by induction that $y'\geq F^\eta(x)$ for all ordinals $\eta$.
	Indeed, for $\eta=0$ this simply says that $y'\geq x$.
	If $\eta$ is a~successor ordinal, from the induction hypothesis $y'\geq F^{\eta-1}(x)$ we obtain
	\[y'=F(y')\geq F\big(F^{\eta-1}(x)\big)=F^\eta(x)\]
	by monotonicity of $F$.
	Finally, if $\eta$ is a~limit ordinal and $y'\geq F^\zeta(x)$ for all $\zeta<\eta$, then $y'$ is an~upper bound of the chain $\set{F^\zeta(x)\mid\zeta<\nobreak\eta}$,
	so the supremum $F^\eta(x)$ of this chain satisfies $y'\geq F^\eta(x)$.
	It follows that $y'\geq F^{\eta_\infty}(x)$.
	Thus, $x\in\dom(\limU{F})$ and $(\limU{F})(x)=F^{\eta_\infty}(x)$, as required.
\end{proof}

\subsection{A type system}

The definition of the partial function $F_X\from X\parto X$ corresponding to a~formula $F$ of unary $\mu$\=/calculus does not provide any natural way of proving that $x\in\dom(F_X)$ for a~given $x\in X$.
Therefore, to prove statements of this form, we introduce a~type system.
To this end, let $A,B\subseteq X$, and let $F$ be a~formula of unary $\mu$\=/calculus.
We show how to derive a~clause $F\dcolon A\to B$.
The type system is developed in a~way that ensures the following \lcnamecref{claim:type-system}.

\begin{claim}\label{claim:type-system}
	If we can derive $F\dcolon A \to B$ then $A\subseteq\dom(F_X)$ and moreover for each $x\in A$ we have $F_X(x)\in B$.
\end{claim}

The rules of the typing system follow.
A general form is
\[
	\frac{\mathit{clause}_1, \ldots ,\mathit{clause}_k}{\mathit{clause}}
	\hspace{7pt} \mathit{condition},
\]
where \emph{condition} (possibly empty) can be any set\=/theoretic condition involving the formulae and sets appearing in the clauses.
The reading of a~rule is: if $\mathit{clause}_1$, \ldots, $\mathit{clause}_k$ are already derived and the \emph{condition} is satisfied, then \emph{clause} is derived as well.
\[\infer
[\forall x\in A.\,H_X(x)\in B]
{H\dcolon A \to B}
{\ }\qquad\qquad
\infer[]
{(F_1\comp F_2)\dcolon A\to B}
{F_1\dcolon A\to C & F_2\dcolon C\to B}\]

\[\infer[A\in\Pcc(X),\forall x\in A.\, F_X(x)\geq x, \Fix\big(F_X\big)\cap A\subseteq B]
{(\limU{F})\dcolon A\to B}
{F\dcolon A\to A}\]

\[\infer[A\in\Pcc(X),\forall x\in A.\, F_X(x)\leq x, \Fix\big(F_X\big)\cap A\subseteq B]
{(\limD{F})\dcolon A\to B}
{F\dcolon A\to A}\]

Note that the above conditions in the case of $(\limU{F})$ and $(\limD{F})$ boil down to $F_X$ being ascending (respectively descending) on $A$.
\cref{claim:type-system} above then follows by an easy induction on the structure of $F$, where in the last two cases we use \cref{lem:lim-as-iteration}.

\section{Basic functions}\label{sec:basic-functions}

In this section we define a~structure in which we want to evaluate our formulae of unary $\mu$\=/calculus.

Let $(\CV,{\leq})$ be some complete lattice, with the least and the greatest elements denoted $\bot$ and $\top$, respectively.
Let $d> 0$ be an~even natural number.
Finally, we fix a~monotone chain\=/continuous function $\delta\from \CV^d\to\CV$.

While the results of this section hold in general, we actually need them for the particular $\CV$ and $\delta$ defined in \cref{sec:spaces}.
We work with the domain $\CV^d$ ordered component\=/wise, which, as already noted, is a~complete lattice.
When working in $\CV^d$ we write $\bar\bot$ for the vector $(\bot,\ldots,\bot)\in\CV^d$.
For $\vtau\in\CV^d$ and $i\in\set{0,\ldots,d{-}1}$ we write $\tau_i\in\CV$ for the respective coordinate of $\vtau$ (i.e., $\vtau=(\tau_0,\ldots,\tau_{d-1})$).
For convenience, we additionally assume that $\tau_{d}=\top$ and $\tau_{d+1}=\bot$.
Moreover, for $i\in\set{0,\ldots,d}$ we define $\vtau\sufcut_i = (\tau_i,\tau_{i+1},\ldots,\tau_{d-1}) \in\CV^{d-i}$ as the restriction of $\vtau$ in which we cut off its first $i$ coordinates.
In particular, $\vtau\sufcut_d$ is $()$, the empty tuple, while $\vtau\sufcut_0 = \vtau$.

\subsection{Functions in \texorpdfstring{$\CV^d$}{V\^{ }d}}\label{ssec:funs_in_V}

Let $\mathfrak{F}=\{\Delta\}\cup\{\Bid_0,\ldots,\Bid_{d-1}\}\cup \{\Cut_1,\ldots,\Cut_{d-1}\}$.
Our goal is to define actual monotone functions with these names, which will turn $\CV^d$ into a~structure over $\mathfrak{F}$.
All our work here is done in $\CV^d$, and therefore we skip the subscript;
that is, the function $H_{\CV^d} \from \CV^d\to\CV^d$, for $H\in\mathfrak{F}$, is denoted just $H$, similarly for more complex formulae $F$.

Let $\Delta\from \CV^d\to \CV^d$ be defined for $\vtau=\big(\tau_0,\ldots,\tau_{d-1}\big)\in \CV^d$ and $i\in\{0,\ldots,d{-}1\}$ as
\[\Delta(\tau_0,\ldots,\tau_{d-1})_i\eqdef \delta\big(\tau_i,\tau_i,\ldots,\tau_i,\tau_{i+1},\ldots,\tau_{d-1}\big).\]

\begin{fact}\label{cl:delta-restr}
	For $0\leq i \leq d$ and $\vtau,\vtau'\in \CV^d$ if $\vtau\sufcut_i=\vtau'\sufcut_i$ then
	\[\Delta(\vtau)\sufcut_i = \Delta(\vtau')\sufcut_i.\tag*{\qed}\]
\end{fact}

Define, for $n\in\set{0,\ldots,d{-}1}$, a~function $\Bid_n\from \CV^d\to \CV^d$, which replaces the coordinates $0,1,\ldots,n$ by the coordinate $n{+}2$:
\begin{align*}
	\Bid_n(\vtau)_i&\eqdef\left\{\begin{array}{ll}
		\tau_i &\text{if $i > n$;}\\
		\tau_{n+2} &\text{if $i\leq n$.}
	\end{array}\right.
\end{align*}
For $n\geq d{-}2$ we use the convention that $\tau_{d}=\top$ and $\tau_{d+1}=\bot$.

Similarly, define for $n\in\set{1,\ldots,d{-}1}$ a~function $\Cut_n\from \CV^d\to \CV^d$, which replaces the coordinates $0,1,\ldots,n$ by the coordinate $n{-}1$:
\begin{align*}
	\Cut_n(\vtau)_i&\eqdef\left\{\begin{array}{ll}
		\tau_i &\text{if $i > n$;}\\
		\tau_{n-1}&\text{if $i\leq n$.}
	\end{array}\right.
\end{align*}

\begin{fact}
	The functions $\Delta$, $\Bid_n$, and $\Cut_n$ are all monotone.
\qed\end{fact}

\subsection{Invariants}\label{ssec:invariants}

The goal of this subsection is to define sets $\CS_n\subseteq\CV^d$ for $n=0,\ldots,d$.
These sets are used later in the paper to derive types of certain functions on $\CV^d$.

Define the \emph{parity order} on $\{0,1,\ldots\}$ as the order $\dots \prec 5\prec 3 \prec 1 \prec 0 \prec 2 \prec 4\prec \dots$,
that is, the unique linear order $\preceq$ satisfying $2n+3 \prec 2n+1 \prec 2n \prec 2n+2$ for every $n\in\N$.
We often use the basic fact that if $n$ is odd and $n'$ is even then $n\prec n'$.

Consider $n\in\set{-1,0,\ldots,d{-}1}$ and $\vtau=(\tau_0,\ldots,\tau_{d-1})\in\CV^d$.
We say that $\vtau$ is:
\begin{itemize}
\item \emph{ordered} if for $0\leq i,j\leq d{-}1$ with $i\preceq j$ we have $\tau_i\leq \tau_j$;
\item \emph{$n$\=/fixed} if for $i\in\{0, \ldots, n\}$ we have $\tau_i=\tau_n$;
\item \emph{$n$\=/saturated} if $\Delta(\vtau)\sufcut_{n+1} = \vtau\sufcut_{n+1}$.
\end{itemize}
In other words, $\vtau$ is \emph{$n$\=/fixed} if the first coordinates of $\vtau$, up to the $n$\=/th, are all equal to each other;
while $\vtau$ is \emph{$n$\=/saturated} if the application of $\Delta$ does not change the coordinates of $\vtau$ above $n$.

For $n\in\set{-1,0,\ldots,d{-}1}$ let $\CS_n$ be the subset of $\CV^d$ containing those $\vtau\in \CV^d$ that are ordered, $n$\=/fixed, $n$\=/saturated, and satisfy
\begin{itemize}
\item	$\Delta(\vtau)\geq\vtau$ for $n$ odd,
\item	$\Delta(\vtau)\leq\vtau$ for $n$ even.
\end{itemize}
Finally, we say that a~partial function $F\from\CV^d\parto\CV^d$ is \emph{$n$\=/stable on $A\subseteq \CV^d$}
for $n\in\{-1,0,\ldots,d{-}1\}$ if $A\subseteq\dom(F)$ and for all $\vtau\in A$ we have $F(\vtau)\sufcut_{n+1} = \vtau\sufcut_{n+1}$.
In particular, the condition of being $n$\=/saturated implies that $\Delta$ is $n$\=/stable on $\CS_n$.

We now prove a~number of properties of the sets $\CS_n$ and of $n$\=/stability.

\begin{fact}\label{ft:bot-in-cs}
	The least element $\bar\bot$ of $\CV^d$ belongs to $\CS_{d-1}$ (recall that $d$ is even).
\qed\end{fact}

\begin{fact}\label{ft:saturation_mono}
	If $\vtau$ is $n$\=/saturated then it is $i$\=/saturated for all $i\in\{n,\ldots,d{-}1\}$.
\qed\end{fact}

\begin{lemma}\label{fixed-after-delta}
	For $n\in\set{0,\ldots,d{-}1}$ if $\vtau=(\tau_0,\ldots,\tau_{d-1})\in\CV^d$ is $n$\=/fixed then $\Delta(\vtau)$ is $n$\=/fixed as well.
\end{lemma}

\begin{proof}
	Take $i\in\{0,\ldots,n\}$ and notice that
	\begin{alignat*}{8}
		\Delta(\vtau)_i
		&=\delta\big(\tau_i&&,\ldots&&,\tau_i&&,\tau_{i+1}&&,\ldots,\tau_{n-1}&&,\tau_n,\tau_{n+1},\ldots,\tau_{d-1}&&\big)\\
		&=\delta\big(\tau_n&&,\ldots&&,\tau_n&&,\tau_n    &&,\ldots,\tau_n    &&,\tau_n,\tau_{n+1},\ldots,\tau_{d-1}&&\big)=\Delta(\vtau)_n;
	\end{alignat*}
	where the middle equation follows from the fact that $\vtau$ is $n$\=/fixed.
\end{proof}

\begin{lemma}\label{ft:stability-in-limit}
	Assume that $F\from\CV^d\parto \CV^d$ is a~partial function that is $n$\=/stable on $A\subseteq \CV^d$.
	If $F$ is ascending on $A$, then $\limU{F}$ is $n$\=/stable on $A$.

	Dually, if $F$ is descending on $A$, then $\limD{F}$ is $n$\=/stable on $A$.
\end{lemma}

\begin{proof}
	Let $\vtau\in A$.
	We use \cref{lem:lim-as-iteration} to see that $(\limU{F})(\vtau)=F^\infty(\vtau)$ or $(\limD{F})(\vtau)=F^\infty(\vtau)$, respectively.
	The element $F^\infty(\vtau)$ is defined by applying $F$ and by taking limits of sequences.
	Applying $F$ preserves the coordinates above $n$ by $n$\=/stability of $F$ on $A$.
	On the other hand, if we have a~sequence of elements of $\CV^d$ which all agree on the coordinates above $n$,
	then the limit of this sequence also agrees with them on these coordinates.
\end{proof}

\begin{lemma}\label{S-chain-complete}
	For $n\in\set{-1,0,\ldots,d{-}1}$ we have $\CS_n\in\Pcc(\CV^d)$, that is, $\CS_n$ is a~chain\=/complete subset of~$\CV^d$.
\end{lemma}

\begin{proof}
	Consider a~chain $C\subseteq\CS_n$, and its supremum $\vsigma=(\sigma_0,\ldots,\sigma_{d-1})\in\CV^d$.
	We have to prove that $\vsigma\in\CS_n$; we also have to prove this for the infimum, but the proof is symmetric, so we continue only with the supremum.

	To see that $\vsigma$ is ordered, consider two coordinates $0\leq i,j\leq d{-}1$ such that $i\preceq j$.
	We have $\tau_i\leq\tau_j\leq\sigma_j$ for all $\vtau\in C$, because $\vtau$ is ordered and because $\sigma_j$ is the supremum of $\set{\tau_j\mid\vtau\in C}$.
	This means that $\sigma_j$ is an~upper bound of the chain $\set{\tau_i\mid\vtau\in C}$, hence the supremum $\sigma_i$ of this chain satisfies $\sigma_i\leq\sigma_j$, as required.

	The order in $\CV^d$ works on every coordinate $i\in\set{0,\ldots,d{-}1}$ separately, that is, $\sigma_i$ is the supremum of the chain $\set{\tau_i\mid\vtau\in C}$.
	These chains for $i\in\set{0,\ldots,n}$ are all the same, because all $\vtau\in C$ are $n$\=/fixed, which implies that the suprema of these chains are equal.
	This means that $\vsigma$ is $n$\=/fixed.

	While proving $n$\=/saturation we use the assumption that $\delta$ is chain\=/continuous.
	It implies that the element $\Delta(\vsigma)_i$, that is, $\delta(\sigma_i,\sigma_i,\ldots,\sigma_i,\sigma_{i+1},\ldots,\sigma_{d-1})$ is the supremum of the chain
	$\set{\delta(\tau_i,\tau_i,\ldots,\tau_i,\allowbreak\tau_{i+1},\ldots,\tau_{d-1})\mid\vtau\in C}$, that is, $\set{\Delta(\vtau)_i\mid\vtau\in C}$.
	But all $\vtau\in C$ are $n$\=/saturated, so for $i\in\set{n{+}1,\ldots,d{-}1}$ the set $\set{\Delta(\vtau)_i\mid\vtau\in C}$ equals $\set{\tau_i\mid\vtau\in C}$,
	whose supremum is $\sigma_i$;
	we obtain $\Delta(\vsigma)_i=\sigma_i$, as required.

	To see that $\Delta(\vsigma)\geq\vsigma$ for odd $n$ we once again notice that $\Delta(\vsigma)_i$ is the supremum of $\set{\Delta(\vtau)_i\mid\vtau\in C}$ for all $i\in\set{0,\ldots,d{-}1}$;
	in other words, $\Delta(\vsigma)$ is the supremum of $\set{\Delta(\vtau)\mid\vtau\in C}$.
	Then $\Delta(\vsigma)\geq\Delta(\vtau)\geq\vtau$ for all $\vtau\in C$, that is, $\Delta(\vsigma)$ is an~upper bound of $C$.
	This implies that $\Delta(\vsigma)\geq\vsigma$, because $\vsigma$ is the supremum of $C$.
	The inequality $\Delta(\vsigma)\leq\vsigma$ for even $n$ is similar: we have $\Delta(\vtau)\leq\vtau\leq\vsigma$ for all $\vtau\in C$,
	that is, $\vsigma$ is an~upper bound of the chain $\set{\Delta(\vtau)\mid\vtau\in C}$ whose supremum is $\Delta(\vsigma)$.
\end{proof}

\begin{lemma}\label{Fix-increases-S}
	For $n\in\set{0,\ldots,d{-}1}$ we have $\Fix(\Delta)\cap\CS_n\subseteq\CS_{n-1}$.
\end{lemma}

\begin{proof}
	Consider $\vtau\in\Fix(\Delta)\cap\CS_n$.
	Because $\vtau\in\CS_n$, we know that $\vtau$ is ordered and $n$\=/fixed, hence also $(n{-}1)$\=/fixed.
	On the other hand $\Delta(\vtau)=\vtau$, which implies that $\vtau$ is $n$\=/saturated and satisfies both $\Delta(\vtau)\geq\vtau$ and $\Delta(\vtau)\leq\vtau$.
\end{proof}

\subsection{Typing the basic functions}

We now prove a~number of typing properties of the functions $\Delta$, $\Bid_n$, and $\Cut_n$.

\begin{lemma}\label{ft:typing_delta}
	For $n\in\set{0,\ldots,d{-}1}$ we can derive $\Delta\dcolon \CS_n\to \CS_n$.
	Moreover, $\Delta$ is $n$\=/stable on $\CS_n$.
\end{lemma}

\begin{proof}
	We have to prove that if $\vtau=(\tau_0,\ldots,\tau_{d-1})\in \CS_n$ then $\Delta(\vtau)\in \CS_n$ (then we can simply apply the first typing rule).
	Consider such $\vtau$.
	\cref{fixed-after-delta} implies that $\Delta(\vtau)$ is $n$\=/fixed.

	We now observe that $\Delta(\vtau)$ is $n$\=/saturated, using $n$\=/saturation of $\vtau$ and \cref{cl:delta-restr}:
	\begin{align*}
		\Delta(\vtau)\sufcut_{n+1}=\vtau\sufcut_{n+1}\quad\Longrightarrow\quad\Delta(\Delta(\vtau))\sufcut_{n+1}=\Delta(\vtau)\sufcut_{n+1}.
	\end{align*}

	Next, we prove that $\Delta(\vtau)$ is ordered.
	By $n$\=/saturation of $\vtau$, for $i\in\{n{+}1,\ldots,d{-}1\}$ we know that $\Delta(\vtau)_i=\tau_i$, and $\vtau$ is ordered.
	By the fact that $\Delta(\vtau)$ is $n$\=/fixed, the first coordinates of $\Delta(\vtau)$, up to the $n$\=/th, are all equal.
	Therefore, it remains to check the order of $\Delta(\vtau)_n$ with respect to the next two coordinates,
	that is, $\Delta(\vtau)_{n+1}$ and $\Delta(\vtau)_{n+2}$, if they exists.
	Consider the case of odd $n$; the other case is symmetric.
	Then one needs to check that $\Delta(\vtau)_n\leq\Delta(\vtau)_{n+1}$ if $n< d{-}1$ and $\Delta(\vtau)_n\geq\Delta(\vtau)_{n-2}$ if $n< d{-}2$.
	We have
	\begin{alignat*}{8}
		\Delta(\vtau)_n
		&=   \delta\big(\tau_n    &&,\tau_n    &&,\ldots&&,\tau_{n}  &&,\tau_{n+1}&&,\tau_{n+2},\tau_{n+3},\ldots,\tau_{d-1}&&\big)\\
		&\leq\delta\big(\tau_{n+1}&&,\tau_{n+1}&&,\ldots&&,\tau_{n+1}&&,\tau_{n+1}&&,\tau_{n+2},\tau_{n+3},\ldots,\tau_{d-1}&&\big)=\Delta(\vtau)_{n+1}
	\end{alignat*}
	because $\vtau$ is ordered (i.e., $\tau_n\leq\tau_{n+1}$) and $\delta$ is monotone.
	Likewise $\tau_{n+1}\geq\tau_n\geq\tau_{n+2}$, which implies
	\begin{alignat*}{8}
		\Delta(\vtau)_n
		&=   \delta\big(\tau_n    &&,\tau_n    &&,\ldots&&,\tau_n    &&,\tau_{n+1}&&,\tau_{n+2},\tau_{n+3},\ldots,\tau_{d-1}&&\big)\\
		&\geq\delta\big(\tau_{n+2}&&,\tau_{n+2}&&,\ldots&&,\tau_{n+2}&&,\tau_{n+2}&&,\tau_{n+2},\tau_{n+3},\ldots,\tau_{d-1}&&\big)=\Delta(\vtau)_{n+2}.
	\end{alignat*}

	Finally, monotonicity of $\Delta$ allows us to deduce $\Delta(\Delta(\vtau))\geq\Delta(\vtau)$ from $\Delta(\vtau)\geq\vtau$ (for odd $n$)
	and $\Delta(\Delta(\vtau))\leq\Delta(\vtau)$ from $\Delta(\vtau)\leq\vtau$ (for even $n$).

	This finishes the proof of $\Delta(\vtau)\in\CS_n$, which allows us to derive $\Delta\dcolon \CS_n\to \CS_n$.
	The fact that $\Delta$ is $n$\=/stable on $\CS_n$ is immediate: it is just another way of saying that elements of $\CS_n$ are $n$\=/saturated.
\end{proof}

\begin{lemma}\label{ft:typing_delta_arr}
	For $n\in\set{0,\ldots,d{-}1}$ we can derive
	\begin{align*}
		\limU{\Delta} &\dcolon \CS_n\to \CS_{n-1} &&\text{for odd $n$;}\\
		\limD{\Delta} &\dcolon \CS_n\to \CS_{n-1} &&\text{for even $n$.}
	\end{align*}
	Moreover, $(\limU{\Delta})$ (and $(\limD{\Delta})$, respectively) is $n$\=/stable on $\CS_n$.
\end{lemma}

\begin{proof}
	To obtain such a~derivation it is enough to apply one of the last two typing rules with $A=\CS_n$ and $B=\CS_{n-1}$;
	we use \cref{ft:typing_delta} to derive the assumption $\Delta\dcolon\CS_n\to\CS_n$,
	\cref{S-chain-complete} to obtain that $\CS_n\in\Pcc(\CV)$, and
	\cref{Fix-increases-S} to obtain that $\Fix(\Delta)\cap\CS_n\subseteq\CS_{n-1}$;
	the assumption that $\Delta(\vtau)\geq\vtau$ (for odd $n$) or $\Delta(\vtau)\leq\vtau$ (for even $n$) for $\vtau\in\CS_n$ is by the definition of $\CS_n$.
	Furthermore, $n$\=/stability of $\limU{\Delta}$ (or $\limD{\Delta}$, respectively) on $\CS_n$ follows from $n$\=/stability of $\Delta$ on $\CS_n$ by \cref{ft:stability-in-limit}.
\end{proof}

\begin{lemma}\label{ft:typing_bid}
	For $n\in\set{0,\ldots,d{-}1}$ we can derive $\Bid_n\dcolon\CS_n\to \CS_n$.
	Moreover, $\Bid_n$ is $n$\=/stable on~$\CS_n$.
\end{lemma}

\begin{proof}
	We have to prove that if $\vtau=(\tau_0,\ldots,\tau_{d-1})\in \CS_n$ then $\Bid_n(\vtau)\in\CS_n$.
	Consider such $\vtau$ and let $\vtau'=(\tau'_0,\ldots,\tau'_{d-1})=\Bid_n(\vtau)$.

	The definition of $\Bid_n$ and the assumption that $\vtau$ is ordered guarantee that $\vtau'$ is ordered as well.

	The definition of $\Bid_n$ explicitly guarantees that $\vtau'$ is $n$\=/fixed.

	Recall that $\vtau'\sufcut_{n+1}=\vtau\sufcut_{n+1}$ (by the definition of $\Bid_n$) and $\Delta(\vtau)\sufcut_{n+1}=\vtau\sufcut_{n+1}$ (by $n$\=/saturation of~$\vtau$).
	We use these two facts and \cref{cl:delta-restr} to see that $\vtau'$ is $n$\=/saturated:
	\begin{align*}
		\Delta(\vtau')\sufcut_{n+1}=\Delta(\vtau)\sufcut_{n+1}=\vtau\sufcut_{n+1}=\vtau'\sufcut_{n+1}.
	\end{align*}

	Finally, we also have to check that for all $i\in\{0,\ldots,d{-}1\}$
	we have $\Delta(\vtau')_i\geq\tau'_i$ if $n$ is odd and $\Delta(\vtau')_i\leq\tau'_i$ if $n$ is even.
	For $i\in\set{n{+}1,\ldots,d{-}1}$ we actually have equality $\Delta(\vtau')_i=\tau'_i$ due to $n$\=/saturation of $\vtau'$.
	On the other hand, $\vtau'$ and $\Delta(\vtau')$ are both $n$\=/fixed (the latter by \cref{fixed-after-delta}); it is thus enough to check the $n$\=/th coordinate.
	If $n=d{-}1$ then $\tau'_n=\tau'_{d+1}=\bot$, so trivially $\Delta(\vtau')_n\geq\vtau'_n$.
	Likewise, if $n=d{-}2$ then $\tau'_n=\tau'_d=\top$, so $\Delta(\vtau')_n\leq\vtau'_n$.
	For $n\leq d{-}3$ observe that
	\begin{alignat*}{10}
		\Delta(\vtau')_n
		&=\delta\big(\tau'_n   &&,\tau'_n   &&,\ldots&&,\tau'_n   &&,\tau'_{n+1}&&,\tau'_{n+2}&&,\tau'_{n+3}&&,\ldots&&, \tau'_{d-1}&&\big)\\
	 	&=\delta\big(\tau_{n+2}&&,\tau_{n+2}&&,\ldots&&,\tau_{n+2}&&,\tau_{n+1} &&,\tau_{n+2} &&,\tau_{n+3} &&,\ldots&&, \tau_{d-1}&&\big);\\
		\tau'_n=\tau_{n+2}
		&=\delta\big(\tau_{n+2}&&,\tau_{n+2}&&,\ldots&&,\tau_{n+2}&&,\tau_{n+2} &&,\tau_{n+2} &&,\tau_{n+3} &&,\ldots&&, \tau_{d-1}&&\big).
	\end{alignat*}
	The first two equations here follow from the definitions of $\Delta$ and $\Bid_n$.
	The last two equations follow from the definition of $\Bid_n$, the fact that $\vtau$ is $n$\=/saturated (so that $\Delta(\vtau)_{n+2}=\tau_{n+2}$), and the definition of $\Delta$.
	Since $\vtau$ is ordered, for odd $n$ we have $\tau_{n+1}\geq\tau_{n+2}$, which by monotonicity of $\delta$ implies $\Delta(\vtau')_n\geq\tau'_n$;
	symmetrically for even $n$.

	This finishes the proof of $\vtau'\in\CS_n$, which allows us to derive $\Bid_n\dcolon \CS_n\to \CS_n$.
	It is immediate from the definition that $\Bid_n$ is $n$\=/stable on $\CS_n$.
\end{proof}

\begin{lemma}\label{ft:typing_cut}
	For $n\in\set{1,\ldots,d{-}1}$ we can derive $\Cut_n\dcolon \CS_{n-2}\to \CS_n$.
	Moreover, $\Cut_n$ is $n$\=/stable on~$\CS_{n-2}$.
\end{lemma}

\begin{proof}
	We have to prove that if $\vtau=(\tau_0,\ldots,\tau_{d-1})\in \CS_{n-2}$ then $\Cut_n(\vtau)\in \CS_n$.
	Consider such $\vtau$ and let $\vtau'=(\tau'_0,\ldots,\tau'_{d-1})=\Cut_n(\vtau)$.

	The definition of $\Cut_n$ and the assumption that $\vtau$ is ordered guarantee that $\vtau'$ is ordered as well.

	The definition of $\Cut_n$ explicitly guarantees that $\vtau'$ is $n$\=/fixed.

	By the definition of $\Cut_n$ we have $\vtau'\sufcut_{n+1}=\vtau\sufcut_{n+1}$, which allows us to apply \cref{cl:delta-restr}.
	Moreover, $(n{-}2)$\=/saturation of $\vtau$ together with \cref{ft:saturation_mono} imply that $\vtau$ is $n$\=/saturated.
	Thus, we obtain $n$\=/saturation of $\vtau'$:
	\begin{align*}
		\Delta(\vtau')\sufcut_{n+1}=\Delta(\vtau)\sufcut_{n+1}=\vtau\sufcut_{n+1}=\vtau'\sufcut_{n+1}.
	\end{align*}

	Finally, we also have to check that for all $i\in\{0,\ldots,d{-}1\}$
	we have $\Delta(\vtau')_i\geq\tau'_i$ if $n$ is odd and $\Delta(\vtau')_i\leq\tau'_i$ if $n$ is even.
	For $i\in\set{n{+}1,\ldots,d{-}1}$ we actually have equality $\Delta(\vtau')_i=\tau'_i$ due to $n$\=/saturation of $\vtau'$.
	On the other hand, $\vtau'$ and $\Delta(\vtau')$ are both $n$\=/fixed (the latter by \cref{fixed-after-delta}); it is thus enough to check the $n$\=/th coordinate.
	Observe that
	\begin{alignat*}{10}
		\Delta(\vtau')_n
		&=\delta\big(\tau'_n   &&,\tau'_n   &&,\ldots&&,\tau'_n   &&,\tau'_n   &&,\tau'_{n+1}&&,\tau'_{n+2}&&,\ldots&&,\tau'_{d-1}&&\big)\\
		&=\delta\big(\tau_{n-1}&&,\tau_{n-1}&&,\ldots&&,\tau_{n-1}&&,\tau_{n-1}&&,\tau_{n+1} &&,\tau_{n+2} &&,\ldots&&,\tau_{d-1}&&\big);\\
		\tau'_n=\tau_{n-1}
		&=\delta\big(\tau_{n-1}&&,\tau_{n-1}&&,\ldots&&,\tau_{n-1}&&,\tau_n    &&,\tau_{n+1} &&,\tau_{n+2} &&,\ldots&&,\tau_{d-1}&&\big).
	\end{alignat*}
	The first two equations here follow from the definitions of $\Delta$ and $\Cut_n$.
	The last two equations follow from the definition of $\Cut_n$, the fact that $\vtau$ is $(n{-}2)$\=/saturated (so that $\Delta(\vtau)_{n-1}=\tau_{n-1}$), and the definition of $\Delta$.
	Since $\vtau$ is ordered, for odd $n$ we have $\tau_{n-1}\geq\tau_n$, which by monotonicity of $\delta$ implies $\Delta(\vtau')_n\geq\tau'_n$;
	symmetrically for even $n$.

	This finishes the proof of $\vtau'\in\CS_n$, which allows us to derive $\Cut_n\dcolon \CS_{n-2}\to \CS_n$.
	It is immediate from the definition that $\Cut_n$ is $n$\=/stable on $\CS_{n-2}$.
\end{proof}

\section{The formula}\label{sec:formula}

Define by induction on $n\in\set{0,\ldots,d{-}1}$ the following formulae of unary $\mu$\=/calculus (recall our assumption that $d$ is even):
\begin{align*}
	\Phi_{0} &\eqdef \Bid_0\comp\limD{\Delta};\\
	\Phi_{n} &\eqdef \Bid_n\comp\limU{\big(\limU{\Delta}\comp\Phi_{n-1}\comp\Cut_n\big)} && \text{for odd $n > 0$;}\\
	\Phi_{n} &\eqdef \Bid_n\comp\limD{\big(\limD{\Delta}\comp\Phi_{n-1}\comp\Cut_n\big)} && \text{for even $n > 0$.}
\end{align*}

In \cref{sec:realisations} we prove that the last formula, $\Phi_{d-1}$, describes accepting runs of the considered automaton.

One may ask if this formula, which is the heart of our proof, could not be defined in a~more direct way, making the whole construction simpler.
However, such a shape of $\Phi_{d-1}$ is dictated by the need to be able to simulate this formula over probabilistic distributions, as explained in \cref{sec:why-phi}.

The formula $\Phi_3$ is depicted below.

\begin{center}
\begin{tikzpicture}
	\showTreeBas{0}
	\showTreeInd{1}
	\showTreeInd{2}
	\showTreeInd{3}
\end{tikzpicture}
\end{center}

\subsection{Graphical representation of \texorpdfstring{$\Phi_{d-1}$}{Phi\_\{d-1\}}}

To make it a bit easier to read the structure of the functions $\Phi_n$, we provide a graphical representation of them.
Fix $d=4$ for the sake of this subsection.

A single value $\vtau\in\CV^d$ has the form $\vtau=(\tau_0,\ldots,\tau_3)$.
When a tree $t$ is given, we obtain a vector $(\tau_0(t),\ldots,\tau_3(t))\in(\powerset(Q))^4$, that is, a quadruple of sets of states.
Thus, to depict $\CV^d$ we draw four horizontal intervals, each representing one copy of $Q$.
The $i$\=/th coordinate $\tau_i$ is depicted by the $i$\=/th interval.

\begin{center}
\begin{tikzpicture}[scale=1.1]

\evalInt{\dPar}{4}
\evalInt{\minD}{0}
\evalInt{\maxD}{\dPar-1}

\showT{0}{}
\showR{0.5}{1}
\end{tikzpicture}
\end{center}

The involved functions $F$ will be drawn bottom\=/up, that is, the argument $\vtau\in\CV^d$ is given in the line below and the result $F(\vtau)\in\CV^d$ is provided in the line above.
Symbols $\bot$ and $\top$ drawn below the upper line mean that the respective coordinate of $F(\vtau)$ is constantly equal to $\bot$ or $\top$ (i.e.,~its value on a given tree $t$ equals $\emptyset$ or $Q$).
Arrows going up between the two lines indicate that the respective coordinates of $F(\vtau)$ are just copied from the respective coordinates of $\vtau$.

The functions $\Bid_n$ for $n\in\set{0,1}$ copy the coordinates above $n$ and set the remaining ones to the coordinate $n{+}2$ (i.e., $2$ or $3$).
The function $\Bid_2$ copies the last coordinate and sets the other coordinates to $\top$.
The function $\Bid_3$ is constantly equal $\bot$.

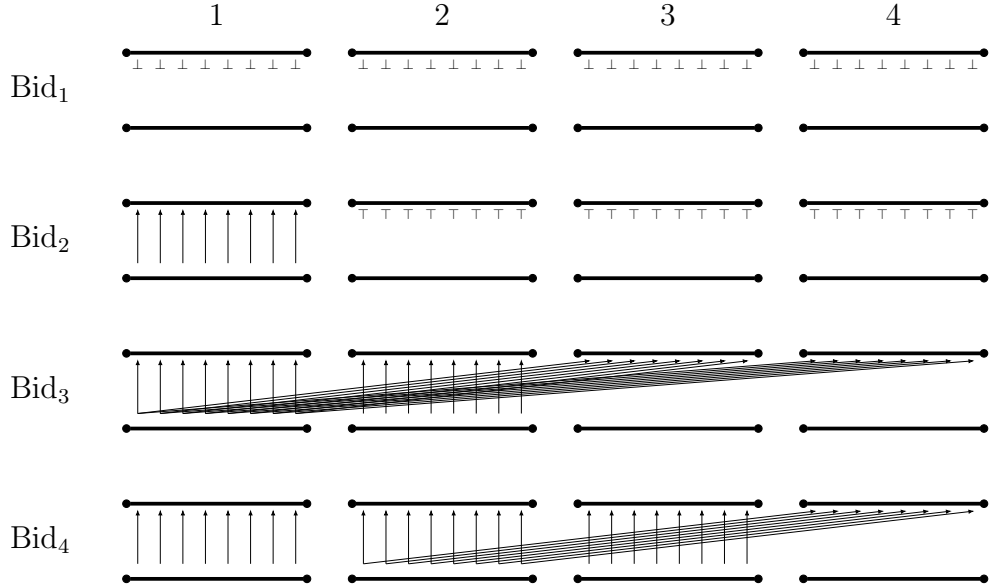
\begin{center}
\begin{tikzpicture}[scale=1.1]
\evalInt{\dPar}{4}
\evalInt{\minD}{0}
\evalInt{\maxD}{\dPar-1}

\showT{0}{}

\foreach \llev in {\minD,...,\maxD} {
	\showR{\llev * 2}{0}
	\showBid{\llev * 2 + 1}{\llev}
	\showR{\llev * 2 + 1}{0}
}
\end{tikzpicture}
\end{center}

We can now move to $\Cut_n$ for $n\in\set{1,2,3}$.
The behaviour of $\Cut_n$ does not depend on the parity of $n$.
$\Cut_n$ copies the $(n{-}1)$\=/th coordinate into the coordinates up to the $n$\=/th.
The remaining coordinates remain unchanged.

\begin{center}
\begin{tikzpicture}[scale=1.1]

\evalInt{\dPar}{4}
\evalInt{\minD}{0}
\evalInt{\maxD}{\dPar-1}

\evalInt{\dParn}{\dPar-1}

\showT{0}{}

\foreach \llev in {1,...,\dParn} {
	\showR{\llev * 2 - 2}{0}
	\showCut{\llev * 2 - 1}{\llev}
	\showR{\llev * 2 - 1}{0}
}

\end{tikzpicture}
\end{center}

Finally, we can draw the functions $\Delta\from \CV^d \to \CV^d$ responsible for applying transitions of the automaton.
We see that transitions from states of a priority $i$ that we use to compute a coordinate $j$ take those states from the coordinate $\max(i,j)$.

More concretely, consider the $1$\=/st coordinate, that is, $\Delta(\vtau)_1\in\CV$.
Assume that a tree $t$ is given and we want to know the set of states $\Delta(\vtau)_1(t)\subseteq Q$.
This set contains all the states $q$ such that there exists a~transition of the form $(q,a,q_\dL, q_\dR)$ with $a=t(\epsilon)$,
where we require that $q_\dL$ belongs to the $\max(1,\Omega(q_\dL))$\=/th coordinate of the profile $\vtau$ applied to the left subtree $t.\dL$, and similarly for $q_\dR$.
This is graphically represented by the shaded region, within which the transitions need to lay.

\begin{center}
\begin{tikzpicture}[scale=1.1]
\evalInt{\dPar}{4}
\evalInt{\minD}{0}
\evalInt{\maxD}{\dPar-1}

\evalInt{\lev}{1}

\showLabel{$\Delta$}

\showT{0}{}
\showR{0}{1}
\showR{1}{1}

\showTran{1}{3}
\end{tikzpicture}
\end{center}

Thus, when drawing all the coordinates of $\Delta$, we obtain:

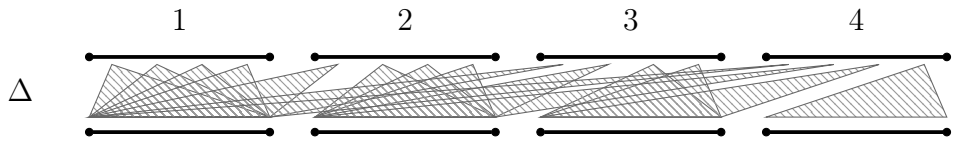
\begin{center}
\begin{tikzpicture}[scale=1.1]
\evalInt{\dPar}{4}
\evalInt{\minD}{0}
\evalInt{\maxD}{\dPar-1}

\showT{0}{}
\showR{0}{0}
\showR{1}{0}
\showDelta{1}
\end{tikzpicture}
\end{center}

Now we can see the formula $\Phi_1$ for $d=2$.
The arrows up and down on the right\=/hand side of the picture indicate the operators $\limU{F}$ and $\limD{F}$ of unary $\mu$\=/calculus.

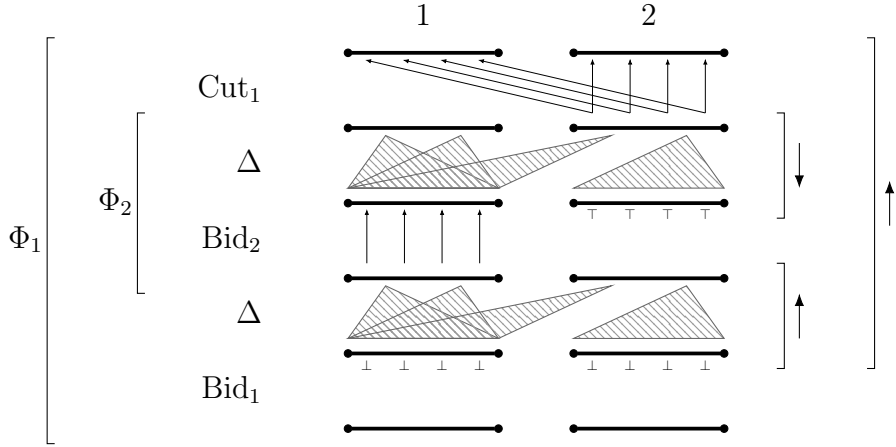
\begin{center}
\begin{tikzpicture}[scale=1.1]

\eval{\xScale}{0.8}
\eval{\yScale}{1.0}

\evalInt{\dPar}{2}
\evalInt{\minD}{0}
\evalInt{\maxD}{\dPar-1}

\showT{0}{}
\showR{0}{2}
\showR{1}{1}
\showR{2}{0}
\showR{3}{0}
\showR{4}{1}
\showR{5}{1}

\showCut{1}{1}
\showDelta{2}
\showBid{3}{0}
\showDelta{4}
\showBid{5}{1}

\showLimD{2}{1}{0}
\showLimU{4}{3}{0}
\showLimU{4}{0}{1}

\showForm{3}{1}{0}{$\Phi_0$}
\showForm{5}{0}{1}{$\Phi_1$}

\end{tikzpicture}
\end{center}

The structure of $\Phi_1$ in this case indicates certain important features of this formula.
First of all, $\Bid_1$ is constantly equal $\bot$, so the actual argument of $\Phi_1$ does not matter (however, it matters for the typing properties, see below).
Also, the $0$\=/th coordinate of the lower sub\=/formula $\limU{\Delta}$ is immediately ignored by the consecutive application of $\Bid_0$.
Thus, the internal sub\=/formula $\limU{\Delta};\Bid_0;\limD{\Delta}$ in fact depends only on the $1$\=/st coordinate of the input and only its $0$\=/th coordinate is used later by $\Cut_1$, by copying it into both coordinates of the result.

Now it's time to see $\Phi_3$ for $d=4$:

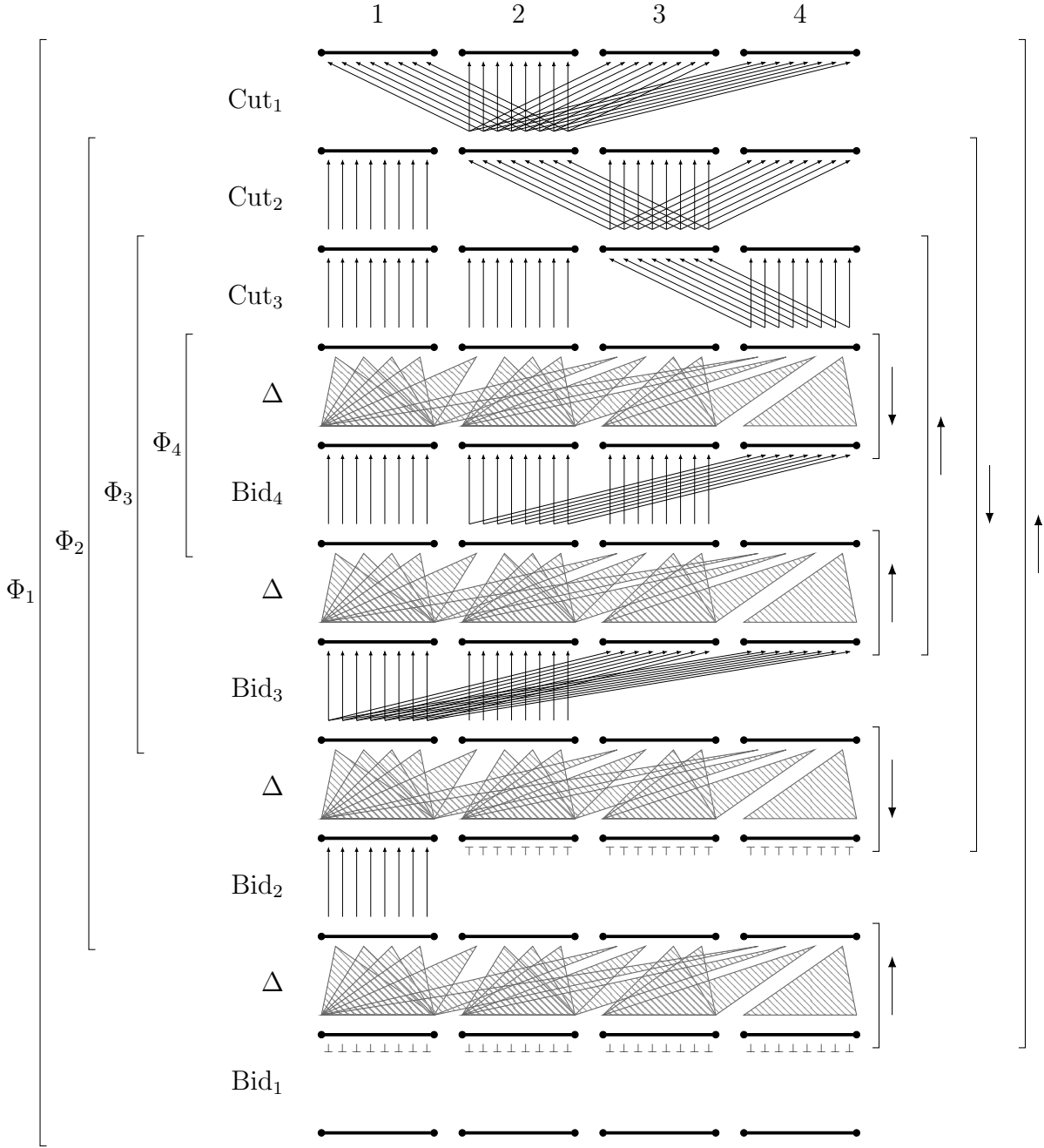
\begin{center}
\begin{tikzpicture}[scale=1.1]

\eval{\xScale}{0.33}
\eval{\yScale}{1.5}

\evalInt{\dPar}{4}
\evalInt{\minD}{0}
\evalInt{\maxD}{\dPar-1}

\showT{0}{}
\showR{0}{2}
\showR{1}{1}
\showR{2}{0}
\showR{3}{0}
\showR{4}{1}
\showR{5}{1}
\showR{6}{2}
\showR{7}{1}
\showR{8}{0}
\showR{9}{0}
\showR{10}{1}
\showR{11}{1}

\showCut{1}{3}
\showCut{2}{2}
\showCut{3}{1}
\showDelta{4}
\showBid{5}{0}
\showDelta{6}
\showBid{7}{1}
\showDelta{8}
\showBid{9}{2}
\showDelta{10}
\showBid{11}{3}

\showLimD{4}{3}{0}
\showLimU{6}{5}{0}
\showLimD{8}{7}{0}
\showLimU{10}{9}{0}

\showLimU{6}{2}{1}
\showLimD{8}{1}{2}
\showLimU{10}{0}{3}

\showForm{5}{3}{0}{$\Phi_0$}
\showForm{7}{2}{1}{$\Phi_1$}
\showForm{9}{1}{2}{$\Phi_2$}
\showForm{11}{0}{3}{$\Phi_3$}

\end{tikzpicture}
\end{center}

\subsection{Typing the functions \texorpdfstring{$\Phi_n$}{Phi\_n}}

Recall, that the syntax of unary $\mu$\=/calculus does not guarantee anything about the domain of the function defined by such a~formula (maybe the domain is empty?).
Before really using the above formulae, we need to provide appropriate type derivations.
Thus, our goal now is to inductively prove the following \lcnamecref{prop}.

\begin{proposition}\label{prop}
	For $n\in\set{0,\ldots, d{-}1}$ we can derive $\Phi_n\dcolon \CS_n \to \CS_{n-1}$.
	Moreover, $\Phi_n$ is $n$\=/stable on~$\CS_n$.
\end{proposition}

Before moving to the proof of the proposition, we mention an~obvious corollary.

\begin{corollary}
	We can derive $\Phi_{d-1}\dcolon \CS_{d-1} \to \CS_{d-2}$, so in particular $\Phi_{d-1}(\bar{\bot})$ is well defined
	(where $\bar{\bot}$ belongs to $\CS_{d-1}$ by \cref{ft:bot-in-cs}).
\qed\end{corollary}

The rest of this section is devoted to the proof of \cref{prop}.
For $n=0$ this is a~direct application of \cref{ft:typing_bid,ft:typing_delta_arr}.
For an~induction step assume that we can derive $\Phi_{n-1}\dcolon \CS_{n-1} \to \CS_{n-2}$
and that $\Phi_{n-1}$ is $(n{-}1)$\=/stable (hence also $n$\=/stable) on $\CS_{n-1}$, and consider $\Phi_n$.
To fix attention, assume that $n$ is odd; the case of even $n$ is symmetric.
Denote
\[\Psi_n=\limU{\Delta}\comp\Phi_{n-1}\comp\Cut_n,\]
so that $\Phi_n=\Bid_n\comp\limU{\Psi_n}$.
First of all, a~derivation of $\Psi_n\dcolon \CS_n \to \CS_n$ follows from sequential application of \cref{ft:typing_delta_arr,ft:typing_cut}:
\begin{align*}
	\limU{\Delta}&\dcolon\CS_n\to\CS_{n-1}	&&\text{(\cref{ft:typing_delta_arr}),}\\
	\Phi_{n+1}&\dcolon\CS_{n-1}\to\CS_{n-2}		&&\text{(induction hypothesis),}\\
	\Cut_n&\dcolon\CS_{n-2}\to\CS_n			&&\text{(\cref{ft:typing_cut}).}
\end{align*}
Moreover, $\Psi_n$, being the composition of these three $n$\=/stable functions (on the respective domains), is $n$\=/stable on $\CS_n$.

We now want to derive $\limU{\Psi_n}\dcolon\CS_n\to\CS_{n-1}$ using the appropriate typing rule.
From \cref{S-chain-complete} we know that $\CS_n\in\Pcc(\CV^d)$;
we need to prove that $\Psi_n(\vtau)\geq\vtau$ for all $\vtau\in\CS_n$ and that $\Fix\big(\Psi_n)\cap\CS_n\subseteq\CS_{n-1}$.

\begin{lemma}\label{lem:psi-vs-delta}
	For $n\in\{1,\ldots,d{-}1\}$ if $\vtau\in\CS_n$ then we have
	\begin{align*}
		\Psi_n(\vtau)&\geq\Delta(\vtau) &&\text{for odd $n$;}\\
		\Psi_n(\vtau)&\leq\Delta(\vtau) &&\text{for even $n$.}
	\end{align*}
\end{lemma}

\begin{proof}
	Consider $n\in\{1,\ldots,d{-}1\}$ and by symmetry assume that $n$ is odd.
	Take any $\vtau\in \CS_n$.
	Clearly it is enough to consider the $n$\=/th coordinate, as the higher coordinates of $\vtau$ are not changed by $\Psi_n$ nor by $\Delta$
	(by $n$\=/stability of these functions),
	and the lower coordinates are equal to the $n$\=/th coordinate (both $\Psi_n(\vtau)$ and $\Delta(\vtau)$, as elements of $\CS_n$, are $n$\=/fixed).
	First, $\vtau\leq(\limU{\Delta})(\vtau)$ by the definition of $\limU{\Delta}$, which by monotonicity of $\Delta$ implies
	$\Delta(\vtau)_n\leq\Delta\big((\limU{\Delta})(\vtau)\big)_n=(\limU{\Delta})(\vtau)_n$.
	Next, the operation $\Phi_{n-1}$ preserves the $n$\=/th coordinate, because it is $(n{-}1)$\=/stable on $\CS_{n-1}$.
	We thus have $\Delta(\vtau)_n\leq\big(\limU{\Delta}\comp\Phi_{n-1}\big)(\vtau)_n$.
	Because the result $\big(\limU{\Delta}\comp\Phi_{n-1}\big)(\vtau)$, as an~element of $\CS_{n-1}$, is ordered,
	we know that also $\Delta(\vtau)_n\leq\big(\limU{\Delta}\comp\Phi_{n-1}\big)(\vtau)_{n-1}$.
	This $(n{-}1)$\=/th coordinate is moved to the $n$\=/th coordinate by $\Cut_n$; we thus have
	\begin{align*}
		\Delta(\vtau)_n&\leq\big(\limU{\Delta}\comp\Phi_{n-1}\comp\Cut_n\big)(\vtau)_n
			=\Psi_n(\vtau)_n,
	\end{align*}
	which finishes the proof of \cref{lem:psi-vs-delta}.
\end{proof}

For $\vtau\in\CS_n$ we have $\Delta(\vtau)\geq\vtau$ (again assuming that $n$ is odd),
thus a~first consequence of \cref{lem:psi-vs-delta} is that $\Psi_n(\vtau)\geq\vtau$ for all $\vtau\in\CS_n$.
Moreover, for $\vtau\in\Fix(\Psi_n)\cap\CS_n$ we have
$\vtau=\Psi_n(\vtau)\geq\Delta(\vtau)$ (by \cref{lem:psi-vs-delta} again) and $\Delta(\vtau)\geq\vtau$ (by the definition of $\CS_n$).
This means that $\Delta(\vtau)=\vtau$, so $\vtau\in\Fix(\Delta)$.

This proves that $\Fix(\Psi_n)\cap\CS_n\subseteq\Fix(\Delta)\cap\CS_n\subseteq\CS_{n-1}$ (cf.~\cref{Fix-increases-S}).
The above allows us to apply the typing rule for $\limU{F}$ and derive $\limU{\Psi_n}\dcolon\CS_n\to\CS_{n-1}$.
By \cref{ft:stability-in-limit} we also have that $\limU{\Psi_n}$ is $n$\=/stable on $\CS_n$.

We finish the proof of \cref{prop} by recalling that $\Phi_n=\Bid_n\comp\limU{\Psi_n}$,
that we can derive $\Bid_n\dcolon\CS_n\to\CS_n$, and that $\Bid_n$ is $n$\=/stable on $\CS_n$ (cf.~\cref{ft:typing_bid}).

\section{\texorpdfstring{$\mu$\=/calculus}{mu-calculus} toolbox}

In this section we recall some facts concerning the standard $\mu$\=/calculus and develop some equalities, which are used in the next section, for the sake of proving \cref{prop:realisation}, that is, that $\Phi_{d-1}$ is equivalent to the formula from \cref{tau-automaton}.

Let $(\CV,{\leq})$ be any complete lattice---the tools developed here are generic $\mu$\=/calculus equalities.

First, by the Knaster\=/Tarski theorem (cf.\@ \cref{sec:basic-concepts}), for an~element $\tau\in\CV$, to prove $\mu x.\,f(x)\leq \tau$, it is enough to show $f(\tau)\leq \tau$.
Similarly, to prove $\tau\leq\nu x.\,f(x)$, it is enough to show $\tau\leq f(\tau)$.
Second, we have the following fact.

\begin{fact}[{\cite[Corollary~1.3.8]{niwinski_rudiments}}]\label{mu-rebase}
	Let $f\from\CV^n\to\CV$ be monotone, let $\theta_0,\ldots,\theta_{n-1}\in\set{\mu,\nu}$, and let $\sigma_\bot,\sigma_\top\in\CV$.
	If $\sigma_\bot\leq\theta_{n-1} x_{n-1}\ldots\theta_0 x_0.\,f(x_0,\ldots,x_{n-1})\leq \sigma_\top$ then
	\begin{alignat*}{4}
		&\theta_{n-1} x_{n-1}\ldots\theta_0 x_0.\,f(x_0&&,x_1,\ldots,x_{n-1})\\
		=\,&\theta_{n-1} x_{n-1}\ldots\theta_0 x_0.\,f((x_0{\wedge}\sigma_\top){\vee}\sigma_\bot&&,x_1,\ldots,x_{n-1}).
	\tag*{\qed}
	\end{alignat*}
\end{fact}

We now have a~series of lemmata, which allow us to simplify some expressions of $\mu$\=/calculus.
The first lemma can be shown using directly the Knaster\=/Tarski theorem.

\begin{lemma}\label{mu-piccolo3}
	For every monotone function $f\from\CV^2\to\CV$, we have
	\begin{align*}
		\nu y.\,\mu x.\,f(x{\land}y,y)=\nu y.\,\mu x.\,f(x,y),\\
		\mu y.\,\nu x.\,f(x{\lor}y,y)=\mu y.\,\nu x.\,f(x,y).
	\end{align*}
\end{lemma}

\begin{proof}
	We prove only the first equality; the second one is symmetric.
	Let $L=\nu y.\,\mu x.\,f(x{\wedge}y,y)$ and $R=\nu y.\,\mu x.\,f(x,y)$.
	The inequality $L\leq R$ follows from monotonicity.
	To show $R\leq L$, by the Knaster\=/Tarski theorem it is enough to show
	\begin{align*}
		R\leq\mu x.\,f(x{\wedge}R,R).
	\end{align*}
	Let $R'=\mu x.\,f(x{\wedge}R,R)$.
	Since $R$ itself is a~fixed point, $R=\mu x.\,f(x, R)$, therefore, again by the Knaster\=/Tarski theorem, to show $R\leq R'$,
	it is enough to show $f(R',R)\leq R'$.
	But we know that, by monotonicity, $R'\leq R$, hence $R'{\wedge}R=R'$ and then $R'=f(R'{\wedge}R, R)=f(R',R)$.
	Thus the claim is proved.
\end{proof}

In a~proof of the next lemma we use the Knaster\=/Tarski theorem and \cref{mu-rebase}.

\begin{lemma}\label{mu-piccolo2}
	For every monotone function $f\from\CV^3\to\CV$, taking
	\begin{alignat*}{4}
		L &=\mu z.\,\nu y.\,\mu x.\,f(x{\lor}z&&,y,z),\\
		R &=\mu z.\,\nu y.\,\mu x.\,f(x&&,y,z),
	\end{alignat*}
	we have $L=R$.
\end{lemma}

\begin{proof}
	Since $f(x{\lor}z,y,z)\geq f(x,y,z)$, by monotonicity of fixed-point operations we have $L=\mu z.\,\nu y.\,\mu x.\,f(x{\lor}z,y,z)\geq R$.
	To show the other inequality, by the Knaster\=/Tarski theorem, it is enough to show
	\begin{align}\label[inequality]{enough}
		\nu y.\,\mu x.\,f(x{\lor}R,y,R)\leq R.
	\end{align}
	By \cref{mu-rebase} (used with $\sigma_\bot=R$ and $\sigma_\top=\top$), we have
	\begin{align*}
		R = \mu z.\,\nu y.\,\mu x.\,f(x{\lor}R,y,z) = \nu y.\,\mu x.\,f(x{\lor}R,y,R)
	\end{align*}
	(the last equality by definition of a~fixed point).
	This implies \cref{enough} and completes the proof.
\end{proof}

The next lemma is again just a (somewhat tricky) consequence of the Knaster\=/Tarski theorem.

\begin{lemma}\label{mu-piccolo}
	For every monotone function $g\from\CV^2\to\CV$, taking
	\begin{align*}
		L&=\mu y.\,\nu x.\,g \big(x,\mu z.\,g(z{\lor}y, z{\lor}y)\big),\\
		R&=\mu y.\,\nu x.\,g(x, y),
	\end{align*}
	we have $L=R$.
\end{lemma}

\begin{proof}
	We first show $L\leq R$.
	By the Knaster\=/Tarski theorem, it is enough to show
	\begin{align*}
		\nu x.\,g\big(x,\mu z.\,g(z{\lor}R, z{\lor}R)\big)\leq R.
	\end{align*}
	We have, by definition, $R=\nu x.\,g(x,R)$, so, by monotonicity, it is enough to show $\mu z.\,g(z{\lor}R,z{\lor}R)\leq R$.
	But this follows immediately (again by the Knaster\=/Tarski theorem) from the fact that	$R=g(R,R)=g(R{\lor}R,R{\lor}R)$.

	To show $R\leq L$, we shall in turn show
	\begin{align*}
		\nu x.\,g(x,L)\leq L.
	\end{align*}
	We have, by definition, $L=\nu x.\,g(x,M)$, where $M=\mu z.\,g(z{\lor}L,z{\lor}L)$.
	It is then enough to show $L\leq M$.
	But we have $L=g(L,M)$ and $M=g(M{\lor}L,M{\lor}L)$, thus the claim follows by monotonicity.
\end{proof}

Our next lemma is a~variant of \cref{mu-piccolo}, where we limit ourselves to the sublattice between $\sigma_\bot$ and $\sigma_\top$.
The statement may look somewhat complicated, but this is exactly what we need in the next section.

\begin{lemma}\label{real-reorganize}
	Let $f\from\CV^2\to\CV$ be monotone, let $\sigma_\bot,\sigma_\top\in\CV$, and let
	\begin{align*}
		L&=\mu y.\,\nu x.\, f\big((x{\lor}h(y)){\land}\sigma_\top,h(y)\big)&&\text{where}\\
		h(y)&=\mu z.\,f(z{\lor}y{\lor}\sigma_\bot,z{\lor}y{\lor}\sigma_\bot)&&\text{and}\\
		R&=\mu y.\, \nu x.\, f\big((x{\lor}((y{\land}\sigma_\top){\lor}\sigma_\bot)){\land}\sigma_\top,(y{\land}\sigma_\top){\lor}\sigma_\bot\big).
	\end{align*}
	If $\sigma_\bot\leq L\leq \sigma_\top$, $f(\sigma_\top,\sigma_\top)=\sigma_\top$, and $f(\sigma_\bot,\sigma_\bot)\geq \sigma_\bot$,
	then $L=R$.
\end{lemma}

\begin{proof}
	First, because of $\sigma_\bot\leq L\leq \sigma_\top$, by \cref{mu-rebase} we can replace $x$ with $(x{\land}\sigma_\top){\lor}\sigma_\bot$ and $y$ with $(y{\land}\sigma_\top){\lor}\sigma_\bot$:
	\begin{align*}
		L=\mu y.\,\nu x.\, f\big(((x{\land}\sigma_\top){\lor}\sigma_\bot{\lor}h'(y)){\land}\sigma_\top,h'(y))\big),
	\end{align*}
	where
	\begin{align*}
		h'(y)=\mu z.\,f\big(z{\lor}(y{\land}\sigma_\top){\lor}\sigma_\bot,z{\lor}(y{\land}\sigma_\top){\lor}\sigma_\bot\big).
	\end{align*}

	Next, for every value of $y$ we have $\sigma_\bot\leq(y{\land}\sigma_\top){\lor}\sigma_\bot\leq\sigma_\top$, so by monotonicity
	\begin{align*}
		\mu z.\,f(z{\lor}\sigma_\bot,z{\lor}\sigma_\bot)\leq h'(y)\leq\mu z.\,f(z{\lor}\sigma_\top,z{\lor}\sigma_\top).
	\end{align*}
	On the one hand, we have
	\begin{align*}
		\sigma_\bot\leq f(\sigma_\bot,\sigma_\bot)=\mu z.\,f(\sigma_\bot,\sigma_\bot)\leq \mu z.\,f(z{\lor}\sigma_\bot,z{\lor}\sigma_\bot).
	\end{align*}
	On the other hand, the assumption $f(\sigma_\top,\sigma_\top)=\sigma_\top$ says that $\sigma_\top$ is a~fixed point of the function $z\mapsto f(z{\lor}\sigma_\top,z{\lor}\sigma_\top)$,
	necessarily greater or equal than the least fixed point $\mu z.\,f(z{\lor}\sigma_\top,z{\lor}\sigma_\top)$.
	We thus have $\sigma_\bot\leq h'(y)\leq \sigma_\top$, so by \cref{mu-rebase} we can write
	\begin{align*}
		h'(x)=\mu z.\,f\big(((z{\lor}y){\land}\sigma_\top){\lor}\sigma_\bot,((z{\lor}y){\land}\sigma_\top){\lor}\sigma_\bot\big).
	\end{align*}
	Due to $\sigma_\bot\leq h'(y)\leq \sigma_\top$ we can also change $h'(y)$ into $(h'(y){\land}\sigma_\top){\lor}\sigma_\bot$ in the previous formula for~$L$, which gives us
	\begin{align*}
		L=\mu y.\,\nu x.\, f\big(((x{\lor}h'(y)){\land}\sigma_\top){\lor}\sigma_\bot,(h'(y){\land}\sigma_\top){\lor}\sigma_\bot\big).
	\end{align*}

	Denote now
	\begin{align*}
		g(x,y)=f\big(((x{\lor}y){\land}\sigma_\top){\lor}\sigma_\bot,(y{\land}\sigma_\top){\lor}\sigma_\bot\big);
	\end{align*}
	then
	\begin{align*}
		L=\mu y.\,\nu x.\, g\big(x,\mu z.\,g(z{\lor}y,z{\lor}y)\big).
	\end{align*}
	Using \cref{mu-piccolo} and expanding the definition of $g$ we obtain the thesis:
	\begin{alignat*}{4}
		L&=\mu y.\,\nu x.\,g(x,y)\\
		 &=\mu y.\,\nu x.\,f\big(((x{\lor}y){\land}\sigma_\top){\lor}\sigma_\bot&&,(y{\land}\sigma_\top){\lor}\sigma_\bot\big)\\
		 &=\mu y.\,\nu x.\,f\big((x{\lor}((y{\land}\sigma_\top){\lor}\sigma_\bot)){\land}\sigma_\top&&,(y{\land}\sigma_\top){\lor}\sigma_\bot\big)=R.
	\tag*{\qedhere}
	\end{alignat*}
\end{proof}

\section{What does the formula \texorpdfstring{$\Phi_{d-1}$}{Phi\_\{d-1\}} do?}\label{sec:realisations}

We now come back to the setting of \cref{sec:formula,sec:basic-functions}.
Recall that we have considered a~complete lattice $(\CV,{\leq})$ and a~monotone and chain-continuous function $\delta\from\CV^d\to\CV$.
Based on $\delta$, we have defined a~function $\Delta\from\CV^d\to\CV^d$,
which then together with functions $\Bid_n$ and $\Cut_n$ was used in a~formula $\Phi_{d-1}$ of unary $\mu$\=/calculus defined in \cref{sec:formula}.

The goal of this section is to prove the following \lcnamecref{prop:realisation}, connecting the formula $\Phi_{d-1}$ with the expression from \cref{tau-automaton},
which has been used there to describe accepting runs of the considered
automaton.

\begin{proposition}\label{prop:realisation}
	We have
	\begin{align*}
		\Phi_{d-1}(\bot)_{d-1}=\mu x_{d-1}.\,\nu x_{d-2}.\,\mu x_{d-3}.\,\nu x_{d-4}\ldots\mu x_{1}.\,\nu x_{0}.\,\delta(x_0,\ldots,x_{d-1}).
	\end{align*}
\end{proposition}

The proof of \cref{prop:realisation} consists of several steps.
At first, from the formulae $\Phi_n$ of unary $\mu$-calculus we obtain expressions that additionally use lattice operations $\wedge$ and $\vee$ (\cref{real-phi} below),
which are later eliminated in \cref{real-no-limits}.

We first give a~general lemma about $\limU{F}$ and $\limD{F}$ for any partial function $F$.

\begin{lemma}\label{real-lim}
	Let $F\from\CV^d\parto\CV^d$ be a~partial function that is $n$\=/stable on $A\subseteq \CV^d$, for some $n\in\set{0,\ldots,d{-}1}$.
	Let also $\phi\from\CV^{d-n}\to\CV$ be such that
	\begin{align}\label{eq:real-lim}
		F(\vtau)_n=\phi(\tau_n,\dots,\tau_{d-1})
	\end{align}
	for all $\vtau=(\tau_0,\ldots,\tau_{d-1})\in A$.
	If $F$ is ascending on $A$ then, for all $\vtau\in A$,
	\[(\limU{F})(\vtau)_n=\mu y_n.\,\phi(y_n{\lor}\tau_n,\tau_{n+1},\ldots,\tau_{d-1}).\]

	Dually, if $F$ is descending on $A$ then, for all $\vtau\in A$,
	\[(\limD{F})(\vtau)_n=\nu y_n.\,\phi(y_n{\land}\tau_n,\tau_{n+1},\ldots,\tau_{d-1}).\]
\end{lemma}

A proof of this lemma follows from the fact that $F$ is ascending (resp.~descending) on~$A$,
and thus the operation $\lor$ (resp.~$\land$) on the first coordinate of the argument of $\phi$ boils down to taking the least (resp.~greatest) fixed point which is ${\geq}$ (resp.~${\leq}$) $\tau_n$.

\begin{proof}
	Let us concentrate on the case when $F$ is ascending on $A$; the case of descending $F$ is symmetric.
	Fix some $\vtau\in A$ and denote $\sigma=\mu y_n.\,\phi(y_n{\lor}\tau_n,\tau_{n+1},\ldots,\tau_{d-1})$.
	First, by \cref{lem:lim-as-iteration} we have $(\limU{F})(\vtau)=F^\infty(\vtau)$.
	Recall the sequence $F^\eta(\vtau)$ used to define $F^\infty(\vtau)$ as $F^{\eta_\infty}(\vtau)$.
	By construction $F^\eta(\vtau)\in A$ and $F^\eta(\vtau)\geq\vtau$ for all ordinals $\eta$.

	First, as in the proof of \cref{ft:stability-in-limit}, we can see that $F^\eta(\vtau)_i=\tau_i$ for $i\in\set{n{+}1,\dots,d{-}1}$, for all ordinals $\eta$.

	We now prove by induction on $\eta$ that $F^\eta(\vtau)_n\leq\sigma$.
	For $\eta=0$ we have
	\begin{alignat*}{4}
		F^0(\vtau)_n=\tau_n\leq F(\vtau)_n&=\phi(&\tau_n&,\tau_{n+1},\ldots,\tau_{d-1})\\
		&\leq\phi(&\sigma{\lor}\tau_n&,\tau_{n+1},\ldots,\tau_{d-1})=\sigma.
	\end{alignat*}

	For a~successor ordinal $\eta$ suppose that $F^{\eta-1}(\vtau)_n\leq\sigma$.
	Recalling that $F^{\eta-1}(\vtau)_n\geq\tau_n$, we have
	\begin{alignat*}{9}
		F^\eta(\vtau)_n=F\big(F^{\eta-1}(\vtau)\big)_n
			&&=   \phi\big(&F^{\eta-1}(\vtau)_n&&,F^{\eta-1}(\vtau)_{n+1}&&,\ldots&&,F^{\eta-1}(\vtau)_{d-1}&&\big)\\
			&&=   \phi\big(&F^{\eta-1}(\vtau)_n{\lor}\tau_n&&,\tau_{n+1}&&,\ldots&&,\tau_{d-1}&&\big)\\
			&&\leq\phi\big(&\sigma{\lor}\tau_n&&,\tau_{n+1}&&,\ldots&&,\tau_{d-1}&&\big)=\sigma,
	\end{alignat*}
	as required.

	Next, consider a~limit ordinal $\eta$, and suppose that for all $\zeta<\eta$ we have $F^\zeta(\vtau)_n\leq\sigma$.
	Recall that $F^\eta(\vtau)$ is the supremum of $\set{F^\zeta(\vtau)\mid\zeta<\eta}$.
	This supremum is taken coordinate\=/wise: $F^\eta(\vtau)_n$ is the supremum of $\set{F^\zeta(\vtau)_n\mid\zeta<\eta}$.
	Because $\sigma$ is an~upper bound of this chain, we have $F^\eta(\vtau)_n\leq\sigma$.

	The above in particular shows that $F^{\eta_\infty}(\vtau)_n\leq\sigma$.
	Simultaneously
	\begin{alignat*}{7}
		F^{\eta_\infty}(\vtau)_n=F^{\eta_\infty+1}(\vtau)_n=F\big(F^{\eta_\infty}(\vtau)\big)_n
			&=\phi\big(F^{\eta_\infty}(\vtau)_n            &&,F^{\eta_\infty}(\vtau)_{n+1}&&,\ldots&&,F^{\eta_\infty}(\vtau)_{d-1}&&\big)\\
			&=\phi\big(F^{\eta_\infty}(\vtau)_n{\lor}\tau_n&&,\tau_{n+1}&&,\ldots&&,\tau_{d-1}&&\big).
	\end{alignat*}
	This means that $F^{\eta_\infty}(\vtau)_n$ is a~fixed point of the function \[y_n\mapsto\phi(y_n{\lor}\tau_n,\tau_{n+1},\ldots,\allowbreak\tau_{d-1}),\]
	whose least fixed point is $\sigma$.
	We thus have $F^{\eta_\infty}(\vtau)_n=\sigma$, which finishes the proof.
\end{proof}

\begin{lemma}\label{real-delta-lim}
	For $n\in\set{0,\ldots,d{-}1}$ and $\vtau=(\tau_0,\dots,\tau_{d-1})\in\CS_n$ we have
	\[(\limU{\Delta})(\vtau)_n=\mu y_n.\,\delta(y_n{\lor}\tau_n,\ldots,y_n{\lor}\tau_n,\tau_{n+1},\ldots,\tau_{d-1})\]
	if $n$ is odd and
	\[(\limD{\Delta})(\vtau)_n=\nu y_n.\,\delta(y_n{\land}\tau_n,\ldots, y_n{\land}\tau_n,\tau_{n+1},\ldots,\tau_{d-1})\]
	if $n$ is even.
\end{lemma}

\begin{proof}
	Suppose that $n$ is odd (the other case is symmetric).
	Because we can derive $\limU{\Delta}\dcolon\CS_n\to\CS_{n+1}$ (cf.\@ \cref{ft:typing_delta_arr}), we know that $\Delta$ is ascending on $\CS_n$.
	It is also $n$\=/stable on $\CS_n$, by \cref{ft:typing_delta}.
	Finally, by the definition of $\Delta$ we have
	$\Delta(\vtau)_n=\delta(\tau_n,\ldots,\tau_n,\tau_{n+1},\ldots,\tau_{d-1})$,
	which by \cref{real-lim} implies the thesis.
\end{proof}

We now define inductively functions $\phi_n\from\CV^{d-n-1}\to\CV$ for $n=\set{-1,0,\ldots,d{-}1}$ taking
\begin{align*}
	\phi_{-1}(x_0,\ldots, x_{d-1})&\eqdef\delta(x_0,\ldots, x_{d-1})
\end{align*}
and, for $n\in\set{0,\ldots,d{-}1}$,
\begin{align*}
	\phi_n(x_{n+1},\ldots, x_{d-1})\eqdef\left\{\begin{array}{ll}
		\hspace{-0.3em}\mu y_n.\,\phi_{n-1}\big((y_n{\land}x_{n+1}){\lor}x_{n+2}, x_{n+1}, x_{n+2}, \ldots, x_{d-1}\big)
			&\mbox{if $n$ is odd;}\\
		\hspace{-0.3em}\nu y_n.\,\phi_{n-1}\big((y_n{\lor}x_{n+1}){\land}x_{n+2}, x_{n+1}, x_{n+2}, \ldots, x_{d-1}\big)
			&\mbox{if $n$ is even,}
	\end{array}\right.
\end{align*}
with the convention that $x_d=\top$ and $x_{d+1}=\bot$.

The functions $\phi_n$ are primarily designed so that they correspond to the formulae $\Phi_n$ (cf.~\cref{real-phi} below),
but they also give an~alternative description of $\delta$ (hence also $\Delta$), as stated by the following equality.

\begin{lemma}\label{phixx=delta}
	For $n\in\set{1,\ldots,d{-}1}$ and $\tau_n,\ldots,\tau_{d-1}\in\CV$ we have
	\begin{align*}
		\phi_{n-2}(\tau_n,\tau_n,\tau_{n+1},\ldots,\tau_{d-1})=\delta(\tau_n,\tau_n,\ldots,\tau_n,\tau_{n+1},\ldots,\tau_{d-1}).
	\end{align*}
\end{lemma}

\begin{proof}
	Induction on $n\in\set{1,\ldots,d{-}1}$.
	For $n=1$ the equality holds by the definition of $\phi_{-1}$.
	Let $n\geq 2$.
	To fix attention, assume that $n$ is odd; the case of even $n$ is symmetric.
	Then, by the definition of $\phi_{n-2}$ and by the induction hypothesis, we have
	\begin{align*}
		\phi_{n-2}(\tau_n,\tau_n,\tau_{n+1},\ldots,\tau_{d-1})
			&=\mu y_{n-1}.\,\phi_{n-3}\big((y_{n-1}{\land}\tau_{n}){\lor}\tau_n,\tau_n,\tau_n,\tau_{n+1},\ldots,\tau_{d-1}\big)\\
			&=\mu y_{n-1}.\,\phi_{n-3}\big(\tau_n,\tau_n,\tau_n,\tau_{n+1},\ldots,\tau_{d-1}\big)\\
			&=\phi_{n-3}\big(\tau_n,\tau_n,\tau_n,\tau_{n+1},\ldots,\tau_{d-1}\big)\\
			&=\delta(\tau_n,\tau_n,\ldots,\tau_n,\tau_{n+1},\ldots,\tau_{d-1}).
	\tag*{\qedhere}\end{align*}
\end{proof}

\begin{lemma}\label{real-phi}
	For $n\in\set{0,\ldots,d{-}1}$ and $\vtau=(\tau_0,\dots,\tau_{d-1})\in\CS_n$ we have
	\[\Phi_n(\vtau)_n=\phi_n(\tau_{n+1},\ldots,\allowbreak \tau_{d-1}).\]
\end{lemma}

\begin{proof}
	Induction on $n$.
	Consider first $n=0$.
	Recall that $\Phi_0=\Bid_0\comp\limD\Delta$.
	By definition $\Bid_0(\vtau)_0=\tau_{2}$ and $\Bid_0(\vtau)_i=\tau_i$ for $i>0$,
	so by \cref{real-delta-lim} for $\vtau\in\CS_n$ we have
	\begin{align*}
		\Phi_0(\vtau)_0=\nu y_0.\,\delta(y_0{\land}\tau_2,\tau_1,\tau_2,\ldots,\tau_{d-1}).
	\end{align*}
	Recall from \cref{prop} that $\Phi_0$ is $0$\=/stable on $\CS_0$ and that $\Phi_0(\vtau)\in\CS_{-1}$;
	in particular $\Phi_0(\vtau)$ is ordered, so
	$\Phi_0(\vtau)_0\geq\Phi_0(\vtau)_{1}=\tau_{1}$.
	In such a~situation \cref{mu-rebase} (used with $\sigma_\bot=\tau_1$ and $\sigma_\top=\top$)
	allows us to replace $y_0$ with $y_0{\lor}\tau_1$ in the above expression for $\Phi_0(\vtau)_0$:
	\begin{align*}
		\Phi_0(\vtau)_0&=\nu y_0.\,\delta((y_0{\lor}\tau_1){\land}\tau_2, \tau_1, \tau_2,\ldots,\tau_{d-1})\\
			&=\phi_0(\tau_1,\dots,\tau_{d-1}).
	\end{align*}

	Let now $n>0$ be odd; the case of even $n$ is symmetric.
	Recall that $\Phi_n=\Bid_n\comp\limU{\Psi_n}$ for $\Psi_n=\limU\Delta\comp\Phi_{n-1}\comp\Cut_n$.
	First, from \cref{real-delta-lim}, where we use \cref{phixx=delta} to replace $\delta$ with $\phi_{n-2}$,
	we obtain for $\vtau'=(\tau'_0,\dots,\tau'_{d-1})\in\CS_n$ that
	\begin{align*}
		(\limU{\Delta})(\vtau')_n=\mu z_n.\,\phi_{n-2}(z_n{\lor}\tau'_n,z_n{\lor}\tau'_n,\tau'_{n+1},\ldots,\tau'_{d-1}).
	\end{align*}
	Moreover, $(\limU{\Delta})(\vtau')_i=\tau_i'$ for $i>n$, because $\limU\Delta$ is $n$\=/stable on $\CS_n$ (cf.\@ \cref{ft:typing_delta_arr}).
	Together with the induction hypothesis this gives us
	\begin{align*}
		\Psi_n(\vtau')_n&=(\limU\Delta\comp\Phi_{n+1})(\vtau')_{n-1}\\
		&=\phi_{n-1}\big(\mu z_n.\,\phi_{n-2}(z_n{\lor}\tau'_n,z_n{\lor}\tau'_n,\tau'_{n+1},\ldots,\tau'_{d-1}), \tau'_{n+1},\ldots,\tau'_{d-1}\big),
	\end{align*}
	where the first equality holds because $\Cut_n$ puts the $(n{-}1)$\=/th coordinate on the $n$\=/th coordinate.
	In the proof of \cref{prop} we have shown that $\Psi_n$ is $n$\=/stable on $\CS_n$ and that we can derive $\limU{\Psi_n}\dcolon\CS_n\to\CS_{n-1}$;
	the latter implies that $\Psi_n$ is ascending on $\CS_n$.
	It follows that \cref{real-lim} can be applied.
	Recalling that $\Bid_n(\vtau)_n=\tau_{n+2}$ and $\Bid_n(\vtau)_i=\tau_i$ for $i>n$, it gives us for $\vtau\in\CS_n$
	\begin{align*}
		\Phi_n(\vtau)_n&=\mu y_n.\,\phi_{n-1}\big(\mu z_n.\,\phi_{n-2}(z_n{\lor}y_n{\lor}\tau_{n+2}, z_n{\lor}y_n{\lor}\tau_{n+2}, \tau_{n+1},\ldots,\tau_{d-1}), \tau_{n+1}, \ldots, \tau_{d-1})
	\end{align*}
	(comparing to the previous expression, we have replaced $\tau'_n$ with $y_n{\lor}\tau_{n+2}$ and every other $\tau'_i$ with~$\tau_i$).

	It remains to prove that the above value equals $\phi_n(\tau_{n+1},\ldots,\allowbreak \tau_{d-1})$.
	Denote $\sigma_\bot=\tau_{n+2}$, $\sigma_\top=\tau_{n+1}$,
	\begin{align*}
		&f(x,y)=\phi_{n-2}(x,y,\tau_{n+1},\ldots,\tau_{d-1}),\qquad\mbox{and}\\
		&h(y)=\mu z.\,f(z{\lor}y{\lor}\sigma_\bot,z{\lor}y{\lor}\sigma_\bot).
	\end{align*}
	Then
	\begin{align*}
		\Phi_n(\vtau)_n=\mu y_n.\,\phi_{n-1}\big(h(y_n),\tau_{n+1},\ldots,\tau_{d-1}\big).
	\end{align*}
	Expanding the definition of $\phi_{n-1}$ in the formula above, we obtain
	\begin{align*}
		\Phi_n(\vtau)_n&=\mu y_n.\,\nu y_{n-1}.\,\phi_{n-2}\big((y_{n-1}{\lor}h(y_n)){\land}\tau_{n+1},h(y_n),\tau_{n+1},\ldots,\tau_{d-1})\\
		&=\mu y.\,\nu x.\,f\big((x{\lor}h(y)){\land}\sigma_\top,h(y)\big).
	\end{align*}
	We know that $\Phi_n$ is $n$\=/stable on $\CS_n$ and that $\Phi_n(\vtau)$, as an~element of $\CS_{n+1}$, is ordered;
	we thus have $\tau_{n+2}\leq \Phi_n(\vtau)_n\leq \tau_{n+1}$, that is, $\sigma_\bot\leq \Phi_n(\vtau)_n\leq \sigma_\top$.
	Moreover, because $\vtau$ is $n$\=/saturated (and because $\tau_{n+2}\leq \tau_{n+1}$), we have
	\begin{alignat*}{6}
		f(\sigma_\bot,\sigma_\bot)
		&=   \delta(\tau_{n+2},\ldots,\tau_{n+2}&&,\tau_{n+1}&&,\tau_{n+2},\tau_{n+3},\ldots,\tau_{d-1})\\
		&\geq\delta(\tau_{n+2},\ldots,\tau_{n+2}&&,\tau_{n+2}&&,\tau_{n+2},\tau_{n+3},\ldots,\tau_{d-1})=\tau_{n+2}=\sigma_\bot
	\end{alignat*}
	and
	\begin{align*}
		f(\sigma_\top,\sigma_\top)
		&=   \delta(\tau_{n+1},\ldots,\tau_{n+1},\tau_{n+2},\ldots,\tau_{d-1})=\tau_{n+1}=\sigma_\top.
	\end{align*}
	This means that all the assumptions of \cref{real-reorganize} are satisfied; we can conclude that
	\begin{align*}
		\Phi_n(\vtau)_n&=\mu y.\, \nu x.\, f\big((x{\lor}((y{\land}\sigma_\top){\lor}\sigma_\bot)){\land}\sigma_\top,(y{\land}\sigma_\top){\lor}\sigma_\bot\big)\\
		&=\mu y_n.\,\nu y_{n-1}.\,\phi_{n-2}\big((y_{n-1}{\lor}((y_n{\land}\sigma_\top){\lor}\sigma_\bot)){\land}\sigma_\top,
			(y_n{\land}\sigma_\top){\lor}\sigma_\bot,\sigma_\top,\sigma_\bot,\tau_{n+3},\ldots,\tau_{d-1}\big)\\
		&=\mu y_n.\,\phi_{n-1}\big((y_n{\land}\sigma_\top){\lor}\sigma_\bot,\sigma_\top,\sigma_\bot,\tau_{n+3},\ldots,\tau_{d-1}\big)\\
		&=\mu y_n.\,\phi_{n-1}\big((y_n{\land}\tau_{n+1}){\lor}\tau_{n+2},\tau_{n+1},\tau_{n+2},\tau_{n+3},\ldots,\tau_{d-1}\big)\\
		&=\phi_n(\tau_{n+1},\ldots,\tau_{d-1}).
	\tag*{\qedhere}\end{align*}
\end{proof}

Next, we define functions $\psi_n\from\CV^{d-n-1}\to\CV$ for $n\in\set{-1,0,\ldots,d{-}1}$ taking
\begin{align*}
	\psi_{-1}(x_0,\ldots, x_{d-1})&=\delta(x_0,\ldots, x_{d-1})
\end{align*}
and, for $n\in\set{0,\ldots,d{-}1}$,
\begin{align*}
	\psi_n(x_{n+1},\ldots, x_{d-1})=\left\{\begin{array}{ll}
		\hspace{-0.3em}\mu y_n.\,\psi_{n-1}(y_n, x_{n+1}, \ldots, x_{d-1})
			&\mbox{if $n$ is odd;}\\
		\hspace{-0.3em}\nu y_n.\,\psi_{n-1}(y_n, x_{n+1}, \ldots, x_{d-1})
			&\mbox{if $n$ is even.}
	\end{array}\right.
\end{align*}
Note that the formula for $\psi_n$ is exactly the suffix of our target expression
\[\mu x_{n}.\,\nu x_{n-1}\ldots\mu x_{1}.\,\nu x_{0}.\,\delta(x_0,\ldots,x_{d-1})\]
without the first $d{-}n{-}1$ fixed\=/point operators (these variables are given as arguments to $\psi_n$).

\begin{lemma}\label{real-no-limits}
	For $n\in\set{1,\ldots,d{-}1}$ and $\tau_{n+1},\ldots,\tau_{d-1}\in\CV$ we have
	\begin{align*}
		\mu y_{n}.\,\nu y_{n-1}.\,\phi_{n-2}(y_{n-1}, y_n, \tau_{n+1}, \tau_{n+2},\ldots,\tau_{d-1})=\psi_n(\tau_{n+1},\ldots,\tau_{d-1})
	\end{align*}
	if $n$ is odd and
	\begin{align*}
		\nu y_{n}.\,\mu y_{n-1}.\,\phi_{n-2}(y_{n-1}, y_n, \tau_{n+1}, \tau_{n+2},\ldots,\tau_{d-1})=\psi_n(\tau_{n+1},\ldots,\tau_{d-1})
	\end{align*}
	if $n$ is even.
\end{lemma}

\begin{proof}
	Induction on $n$.
	For $n=1$ the equality is by the definitions of $\phi_{-1}$ and $\psi_{1}$.
	Let $n>1$.
	To fix attention, assume that $n$ is odd; the case of even $n$ is symmetric.
	Then, by the definition of $\phi_{n-2}$, by \cref{mu-piccolo3}, by \cref{mu-piccolo2}, and by the induction hypothesis, we have
	\begin{align*}
		\mu y_n.\,\nu y_{n-1}.\,\phi_{n-2}(y_{n-1}, y_n, \tau_{n+1},\ldots,\tau_{d-1})\hspace{-7.5em}&\\
		&=\mu y_n.\,\nu y_{n-1}.\,\mu y_{n-2}.\,\phi_{n-3}\big((y_{n-2}{\land}y_{n-1}){\lor}y_n,y_{n-1},y_n,\tau_{n+1},\ldots,\tau_{d-1}\big)\\
		&=\mu y_n.\,\nu y_{n-1}.\,\mu y_{n-2}.\,\phi_{n-3}\big(y_{n-2}{\lor}y_n,y_{n-1},y_n,\tau_{n+1},\ldots,\tau_{d-1}\big)\\
		&=\mu y_n.\,\nu y_{n-1}.\,\mu y_{n-2}.\,\phi_{n-3}\big(y_{n-2},y_{n-1},y_n,\tau_{n+1},\ldots,\tau_{d-1}\big)\\
		&=\mu y_n.\,\psi_{n-1}(y_n,\tau_{n+1},\ldots,\tau_{d-1})\\
		&=\psi_n(\tau_{n+1},\ldots,\tau_{d-1}).
	\tag*{\qedhere}\end{align*}
\end{proof}

\begin{lemma}\label{phi-psi}
	We have $\phi_{d-1}()=\psi_{d-1}()$.
\end{lemma}

\begin{proof}
	Recall our convention that $x_{d}=\top$ and $x_{d+1}=\bot$.
	Thus the thesis follows directly from \cref{mu-piccolo3,real-no-limits}:
	\begin{align*}
		\phi_{d-1}()&=\mu y_{d-1}.\,\nu y_{d-2}.\,\phi_{d-3}(y_{d-2}{\lor}y_{d-1}, y_{d-1})\\
		&=\mu y_{d-1}.\,\nu y_{d-2}.\,\phi_{d-3}(y_{d-2},y_{d-1})=\psi_{d-1}().
	\tag*{\qedhere}\end{align*}
\end{proof}

\cref{real-phi,phi-psi} imply that $\Phi_{d-1}(\bot)_{d-1}=\psi_{d-1}()$
(recall that $\bot \in\CS_{d-1}$ by \cref{ft:bot-in-cs});
this finishes the proof of \cref{prop:realisation}.

\section{Measurability}\label{sec:measurability}

In our final argument (\cref{sec:tarski}) we want to describe profiles $\tau\in \CV$ by measures of sets like $\{t\in\trees_\Sigma\mid \tau(t)=R\}$ for some fixed $R$.
For this to make sense, we need to prove that such sets are measurable for all profiles $\tau$ in consideration;
this is the goal of this section.

Recall that the set of profiles is $\CV=\big(\trees_\Sigma\to\powerset(Q)\big)$.
Define $\CR$ as the set $\powerset\big(Q\times\set{0,\ldots,\allowbreak d{-}1}\big)$.
Instead of $\CV^d$ we prefer to consider $\CW\eqdef(\trees_\Sigma\to\CR)$.

There is a~canonical bijection between $\CV^d$ and $\CW$:
a~tuple $\vtau=(\tau_0,\ldots,\tau_{d-1})\in\CV^d$ is in bijection with a~function in $\CW$ that maps a~tree $t$ to $\big\{(q,i)\mid q\in\tau_i(t)\big\}$.
The coordinate\=/wise order $\leq$ on $\CV^d$ is then mapped to the order $\leq$ on $\CW$ given for $\sigma,\sigma'\in\CW$
by $\sigma\leq\sigma'$ if $\sigma(t)\subseteq\sigma'(t)$ for all trees $t\in\trees_\Sigma$.
We can thus identify these two spaces, $\CV^d$ and $\CW$.
We denote the elements of $\CW$ by normal letters $\tau$, $\sigma$, as they are no longer treated as vectors.

The measurability argument given in this section follows that from Gogacz et al.\@ \cite[page~117]{michalewski_measure_final}---%
the proof that $E^\alpha_v$ are all measurable and their measure converges to the measure of their union.
This is a~variant of the technique used by Lusin and Sierpiński~\cite{lusin_sierpinski} to prove that the set\=/theoretic \emph{Souslin operation} $\CA$ preserves measurability.

Our proof (except of \cref{our-local}, talking about particular basic functions) works for a~general space of the form $\CW=(\trees_\Sigma\to\CR)$,
with $(\CR,\subseteq)$ being a finite complete lattice; in our case $\CR=\powerset\big(Q\times\{0,\ldots,d{-}1\}\big)$, and $\subseteq$ is inclusion of sets.

\begin{definition}
	We say that $\tau\in\CW$ is a \emph{random variable} if for every $R\in\CR$ the set of trees $\tau^{-1}(\{R\})$ is measurable.

	We say that a~partial function $F\from\CW\parto\CW$ is \emph{local} on a~set $A\subseteq\CW$ if $A\subseteq\dom(F)$ and for all $\tau, \tau'\in A$ and all trees $t\in\trees_\Sigma$,
	if $\tau$ and $\tau'$ agree on all the subtrees of $t$ then also $F(\tau)$ and $F(\tau')$ agree on all the subtrees of $t$.

	We say that $F\from\CW\parto\CW$ \emph{preserves measurability} on $A\subseteq\CW$ if $A\subseteq\dom(F)$ and whenever $\tau\in A$ is a~random variable then so is $F(\tau)$.
\end{definition}

Because the measure on $\trees_\Sigma$ is the product measure, where the probability in every node is defined in the same way, we have the following \lcnamecref{subtree-measurable}.

\begin{fact}\label{subtree-measurable}
	Let $u\in\set{\dL,\dR}^\ast$.
	If a~set $S\subseteq\trees_\Sigma$ is measurable, then $\set{t\in\trees_\Sigma\mid t.u\in S}$ is measurable as well.
	Moreover, the measure of these two sets is the same.
\qed\end{fact}

\begin{lemma}\label{our-local}
	All the basic functions $\Delta$, $\Bid_n$, $\Cut_n$ (when treated as functions on $\CW$) are local and preserve measurability on the whole $\CW$.
\end{lemma}

\begin{proof}
	We say that a~function $F\from\CW\to\CW$ is \emph{superficial}
	if there exists a~function $F_\CR\from\CR\to\CR$ such that $F(\tau)(t)=F_\CR(\tau(t))$ for all $\tau\in\CW$ and $t\in\trees_\Sigma$.
	Note that the functions $\Bid_n$ and $\Cut_n$ are superficial.
	Clearly every superficial function $F\from\CW\to\CW$ is local on $\CW$;
	it also preserves measurability on $\CW$ because for any $\tau\in\CW$ and $R\in\CR$ the set
	$$\big(F(\tau)\big)^{-1}\big(\set{R}\big)=\big\{t\in\trees_\Sigma\mid F(\tau)(t)=R\big\}=\big\{t\in\trees_\Sigma\mid F_\CR(\tau(t))=R\big\}$$
	is a~finite union of sets of the form $\tau^{-1}\big(\{R'\}\big)$ ranging over certain $R'\in\CR$.

	Now consider $\Delta$.
	Clearly, the value of $\Delta(\tau)(t)$ depends only on $t(\epsilon)\in \Sigma$ and the values of $\tau(t.\dL)$ and $\tau(t.\dR)$, so $\Delta$ is also local on $\CW$.
	Next, observe that for any $\tau\in\CW$ and $R\in \CR$ the set
	\[\big(\Delta(\tau)\big)^{-1}\big(\{R\}\big)=\big\{t\in\trees_\Sigma\mid \Delta(\tau)(t)=R\big\},\]
	is a~finite union of the sets of the form
	\[\big\{t\in\trees_\Sigma\mid t(\epsilon)=a\land\tau(t.\dL)=R^\dL\land \tau(t.\dR)=R^\dR\big\},\]
	ranging over certain $a$, $R^\dL$, $R^\dR$.
	These sets are measurable by \cref{subtree-measurable}.
\end{proof}

Let $\zeta$ be an~ordinal and let $(x_\zeta)_{\zeta<\eta}$ be a~sequence of real numbers.
We say that this sequence \emph{converges} to $y\in\R$ (or that $y$ is the \emph{limit} of $(x_\zeta)_{\zeta<\eta}$)
if for every $\varepsilon>0$ there exists $\zeta_0<\eta$ such that for every $\zeta$ satisfying $\zeta_0\leq \zeta < \eta$ we have $|x_\zeta-y|<\epsilon$.

Consider an~ordinal $\eta$ and take a~sequence of random variables $(\tau_\zeta)_{\zeta<\eta}$ from $\CW$.
We say that this sequence \emph{converges in measure} to a~random variable $\sigma\in\CW$
if for every set $R\in\CR$ the sequence of numbers $(x_\zeta)_{\zeta<\eta}$ defined as $x_\zeta\eqdef\prob\big(\tau_\zeta^{-1}(\{R\})\big)$
converges to the number $\prob\big(\sigma^{-1}(\{R\})\big)$.

\subsection{Taking limits}

The crucial difficulty and potential source of non\=/measurability may come from the operations $\limU{F}$ and $\limD{F}$.
This is taken care of by the following proposition.

\begin{proposition}\label{pro:all-measurable-lim}
	Let $F\from\CW\parto\CW$ be a~partial function that is ether ascending or descending on $A\subseteq\CW$,
	local on $A$, and preserves measurability on $A$.
	Assume that $\tau\in A$ is a~random variable.
	Then all $F^\eta(\tau)$ from the definition of $F^\infty(\tau)$ are random variables.
	Moreover, for every limit ordinal $\eta$ the sequence $(F^\zeta(\tau))_{\zeta<\eta}$ converges in measure to $F^\eta(\tau)$.
\end{proposition}

The rest of the subsection is devoted to a~proof of this \lcnamecref{pro:all-measurable-lim}.

Before starting, let us recall two known facts.
The first of them holds because the supremum for ordinals is defined as the union.

\begin{fact}\label{sup-countable}
	The supremum of a~countable set of countable ordinals is countable.
\qed\end{fact}

The next lemma is a~classical result in measure theory.
It is usually stated for sequences indexed by naturals numbers (see, e.g., Oxtoby~\cite[Theorem~3.17]{oxtoby}), but we need it in a~slightly more general form.

\begin{lemma}\label{countable-union}
	Let $\eta>0$ be a~countable ordinal and let $(V_\zeta)_{\zeta<\eta}$ be an~increasing sequence of measurable sets.
	Then the (countable) union $V\eqdef\bigcup_{\zeta<\eta} V_\zeta$ is measurable and the measures of $V_\zeta$ converge to the measure of $V$.

	Similarly for a~decreasing sequence and its intersection.
\end{lemma}

\begin{proof}
	Consider the case of an~increasing sequence $(V_\zeta)_{\zeta<\eta}$; the part about a~descending sequence can be deduced by switching to complements.
	If $\eta$ is a~successor ordinal then $V=V_{\eta-1}$ and the \lcnamecref{countable-union} follows.

	Assume that $\eta$ is a~limit ordinal.
	Because $\eta$ is countable, there are only countably many ordinals $\zeta<\eta$, which means that we can organise them into a~sequence indexed by natural numbers.
	Out of this sequence let us select only those elements that are greater than all elements before them.
	This way we obtain an~increasing sequence of ordinals $(\zeta_i)_{i\in\N}$ such that $\eta=\sup_{i\in\N}\zeta_i$.

	We then have $V=\bigcup_{i\in\N} V_{\zeta_i}$ with $(V_{\zeta_i})_{i\in\N}$ being an~increasing sequence of sets,
	indexed by natural numbers.
	Thus, continuity of measure in its usual formulation~\cite[Theorem~3.17]{oxtoby} implies that the measure of $V$ is the limit of the measures of $V_{\zeta_i}$.
	But this implies that the measures of $(V_{\zeta})_{\zeta<\eta}$ also converge to the measure of $V$.
\end{proof}

\cref{countable-union} implies the following standard extension to random variables.

\begin{lemma}\label{lem:countable-limit}
	Let $\eta>0$ be a~countable ordinal and let $(\tau_\zeta)_{\zeta<\eta}$ be a~sequence of random variables from $\CW$.
	If the sequence is monotonic then its limit $\sigma$ is a~random variable and the sequence $(\tau_\zeta)_{\zeta<\eta}$ converges in measure to $\sigma$.
\end{lemma}

\begin{proof}
	Consider the case of an~increasing sequence.
	Let $\sigma$ be the limit from the statement and take $R\in\CR$.
	We need to prove that $\sigma^{-1}(\{R\})$ is a~measurable subset of $\trees_\Sigma$ and that it is the limit of~$\tau_\zeta^{-1}(\{R\})$.

	For $\zeta<\eta$ let $U_\zeta\eqdef \{t\in\trees_\Sigma\mid \tau_\zeta(t)\supseteq R\}$ and $V_\zeta\eqdef\{t\in\trees_\Sigma\mid \tau_\zeta(t)\supsetneq R\}$.
	Clearly all the sets $U_\zeta$, $V_\zeta$ are measurable, as Boolean combinations of sets of the form $\tau_\zeta^{-1}(\{R'\})$.
	Moreover, $\tau_\zeta^{-1}(\{R\})=U_\zeta \setminus V_\zeta$.

	Because the sequence $(\tau_\zeta)_{\zeta<\eta}$ is increasing, we know that also the two sequences of sets, $(U_\zeta)_{\zeta<\eta}$ and $(V_\zeta)_{\zeta<\eta}$, are increasing.
	Take $U\eqdef \bigcup_{\zeta<\eta} U_\zeta$ and $V\eqdef \bigcup_{\zeta<\eta} V_\zeta$.
	Apply \cref{countable-union} to both these sets to notice that $U$ and $V$ are both measurable and
	their measures are limits of the measures of $(U_\zeta)_{\zeta<\eta}$ and $(V_\zeta)_{\zeta<\eta}$.
	It follows that the measures of $\big(\tau_\zeta^{-1}(\set{R})\big)_{\zeta<\eta}$ converge to the measure of~$U\setminus V$.

	It remains to prove that
	\[\sigma^{-1}(\{R\})=U \setminus V.\]
	To this end, consider a~tree $t\in\trees_\Sigma$.
	Since the order in $\CW$ works separately on every tree, the sequence $\big(\tau_\zeta(t)\big)_{\zeta<\eta}$ is increasing, and its limit equals $\sigma(t)$.
	But this is an~increasing sequence in a~finite lattice $(\CR,\subseteq)$, so it must stabilise from some moment on:
	for some $\zeta_0<\eta$ we have $\sigma(t)=\tau_\zeta(t)$ for all $\zeta$ such that $\zeta_0\leq\zeta<\eta$.
	It follows that $\sigma(t)\supseteq R$ if and only if $t\in U$.
	Similarly, $\sigma(t)\supsetneq R$ if and only if $t\in V$.
	Therefore, $t\in U\setminus V$ if and only if $\sigma(t)=R$.
\end{proof}

We are now ready to start a~proof of \cref{pro:all-measurable-lim}.
We consider the case when $F$ is ascending on $A$; the dual case when $F$ is descending on $A$ is entirely symmetric.
We fix $\tau\in A$.

\begin{claim}\label{cl:locality}
	Let $t_0\in\trees_\Sigma$, and let $\zeta$ be an~ordinal.
	If $F^{\zeta+1}(\tau)$ and $F^\zeta(\tau)$ agree on all the subtrees of $t_0$,
	then $F^\eta(\tau)$ and $F^\zeta(\tau)$ agree on all the subtrees of $t_0$ (in particular $F^\eta(\tau)(t_0)=F^\zeta(\tau)(t_0)$), for all $\eta\geq\zeta$.
\end{claim}

\begin{proof}
	Induction on $\eta\geq\zeta$.
	For $\eta=\zeta$ the statement is trivial.
	For successor ordinals $\eta\geq\zeta+1$ we have
	\begin{align*}
		F^\eta(\tau)(t)&=F(F^{\eta-1}(\tau))(t)=F(F^\zeta(\tau))(t)
		=F^{\zeta+1}(\tau)(t)=F^\zeta(\tau)(t)
	\end{align*}
	for all the subtrees $t$ of $t_0$, where the second equality is by the induction hypothesis and by locality of $F$ on $A$.
	If $\eta>\zeta$ is a~limit ordinal, then for every subtree $t$ of $t_0$ we have that $F^\eta(\tau)(t)$
	is the limit of the increasing sequence $\big(F^{\eta'}(\tau)(t)\big)_{\eta'<\eta}$ that stabilises on $F^\zeta(\tau)(t)$ by the induction hypothesis;
	thus this limit equals $F^\zeta(\tau)(t)$.
\end{proof}

Below, we write $\omega_1$ for the least uncountable ordinal.

\begin{claim}\label{omega1-is-enough}
	We have $F^\eta(\tau)=F^{\omega_1}(\tau)$ for all $\eta\geq\omega_1$.
\end{claim}

\begin{proof}
	Consider a~single tree $t\in\trees_\Sigma$.
	We prove that the value $F^\eta(\tau)(t)$ is constant for all $\eta\geq \zeta_t$ for some countable ordinal $\zeta_t$,
	hence even more for all $\eta\geq\omega_1$.

	Consider
	\[Z_t\eqdef \big\{\eta<\omega_1\mid \exists v\in\{\dL,\dR\}^\ast.\ F^{\eta+1}(\tau)(t.v)\neq F^\eta(\tau)(t.v)\big\}\]
	and take $\zeta_t\eqdef\sup Z_t+1$.
	Recall that $\big(F^\eta(\tau)\big)_{\eta<\omega_1}$ is an~increasing sequence.
	Since $\CR$ is a~finite lattice, the value of $F^\eta(\tau)(t.v)\in\CR$ on a~single tree $t.v\in\trees_\Sigma$ can grow only finitely many times.
	Because moreover there are only countably many nodes $v\in\{\dL,\dR\}^\ast$, the set $Z_t$ is countable;
	in consequence, $\zeta_t<\omega_1$ by \cref{sup-countable}.
	Simultaneously $\eta<\zeta_t$ for all $\eta\in Z_t$, which means that $F^{\zeta_t+1}(\tau)$ and $F^{\zeta_t}(\tau)$ agree on all the subtrees of $t$.
	By \cref{cl:locality} we then have $F^\eta(\tau)(t)=F^{\zeta_t}(\tau)(t)$ for all $\eta\geq\zeta_t$,
	which finishes the proof.
\end{proof}

Because $F$ preserves measurability on $A$ and so does taking countable limits (see \cref{lem:countable-limit}), we have the following \lcnamecref{cl:countable-measurable}.

\begin{fact}\label{cl:countable-measurable}
	All the elements $F^\eta(\tau)$ for $\eta<\omega_1$ are random variables.
	Moreover, for every limit ordinal $\eta<\omega_1$ the sequence $(F^\zeta(\tau))_{\zeta<\eta}$ converges in measure to $F^\eta(\tau)$.
\qed\end{fact}

Note that \cref{cl:countable-measurable} gives us the thesis of \cref{pro:all-measurable-lim} for $\eta<\omega_1$.
Because $F^\eta(\tau)=F^{\omega_1}(\tau)$ for all $\eta\geq\omega_1$ (cf.\@ \cref{omega1-is-enough}),
to finish the proof of \cref{pro:all-measurable-lim} it remains to prove that that $F^{\omega_1}(\tau)$ is a~random variable
and that the sequence $(F^\zeta(\tau))_{\zeta<\omega_1}$ converges in measure to $F^{\omega_1}(\tau)$.

For all $\eta<\omega_1$, $R\in\CR$, and $u\in\{\dL,\dR\}^\ast$ we define
\[E^\eta_{R,u}\eqdef \big\{t\in\trees_\Sigma\mid R\subseteq F^\eta(\tau)(t.u)\big\}.\]

\begin{claim}
	For each $R\in \CR$ and $u \in\{\dL,\dR\}^\ast$ the sequence of sets $\big(E^\eta_{R,u}\big)_{\eta<\omega_1}$ is increasing in the subset order.
	Moreover, all these sets are measurable.
\end{claim}

\begin{proof}
	The fact that the sequence is increasing is obvious.
	To see that each of the sets $E^\eta_{R,u}$ is measurable it is enough to recall that $F^\eta(\tau)$ is a~random variable,
	to use \cref{subtree-measurable} for shifting from $t$ to $t.u$,
	and to observe that $R\subseteq\tau(t)$ if and only if $t$ belongs to the (finite) union of the sets $\tau^{-1}(\{R'\})$ for $R'\supseteq R$.
\end{proof}

Fix $R\in \CR$ and $u \in\{\dL,\dR\}^\ast$ and consider the measures $m_\eta\eqdef \prob\big(E^\eta_{R,u}\big)$ for $\eta<\omega_1$.
It is a~nondecreasing sequence of real numbers in $[0,1]$ indexed by ordinals.
Whenever $m_{\eta+1}>m_{\eta}$ then there exists a~rational number in the interval $(m_\eta,m_{\eta+1})$.
Thus, the set
\begin{align*}
	Z_{R,u}\eqdef\set{\eta<\omega_1\mid m_{\eta+1}\neq m_\eta}
\end{align*}
is countable.
Let $\zeta_{R,u}\eqdef\sup Z_{R,u}+1$.
We have $\zeta_{R,u}<\omega_1$ by \cref{sup-countable} and
$\prob\big(E^\eta_{R,u}\big)=\prob\big(E^{\zeta_{R,u}}_{R,u}\big)$ for all $\eta\geq \zeta_{R,u}$ by definition.
Let $\zeta$ be the supremum of all the (countably many) numbers $\zeta_{R,u}$ for $R\in\CR$ and $u\in\{\dL,\dR\}^\ast$.
Using again \cref{sup-countable} we see that $\zeta<\omega_1$.
Let
\[U\eqdef \bigcup_{R\in\CR, u\in\{\dL,\dR\}^\ast} \left(E^{\zeta+1}_{R,u}\setminus E^\zeta_{R,u}\right).\]

Notice that for every $R\in\CR$ and $u \in\{\dL,\dR\}^\ast$ we have $\prob\big(E^{\zeta+1}_{R,u}\big)=\prob\big(E^{\zeta}_{R,u}\big)$ and, therefore, the set $E^{\zeta+1}_{R,u}\setminus E^\zeta_{R,u}$ has measure $0$.
This means that also $U$ is of measure $0$, as a~countable union of such sets.

\begin{claim}
	If $t_0\in\trees_\Sigma$ is such that $F^{\omega_1}(\tau)(t_0)\neq F^\zeta(\tau)(t_0)$ then $t_0\in U$.
\end{claim}

\begin{proof}
	Consider such a~tree and assume for the sake of contradiction that $t_0\notin U$.
	By the definition of $U$ it means that for every $R\in\CR$ and $u\in\{\dL,\dR\}^\ast$ we have $t_0\in E^{\zeta}_{R,u}\Leftrightarrow t_0\in E^{\zeta+1}_{R,u}$.
	This implies that whenever $t$ is a~subtree of $t_0$ (i.e., whenever $t=t_0.u$ for some $u\in\set{\dL,\dR}^\ast$) then $F^\zeta(\tau)(t)=F^{\zeta+1}(\tau)(t)$.
	Thus, by \cref{cl:locality} we know that for every $\eta\geq\zeta$ we have $F^\eta(\tau)(t_0)=F^\zeta(\tau)(t_0)$.
	Due to $\zeta<\omega_1$ this implies that $F^{\zeta}(\tau)(t_0)=F^{\omega_1}(\tau)(t_0)$, a~contradiction.
\end{proof}

For every set $R\in\CR$, denoting $V_R=\big(F^\zeta(\tau)\big)^{-1}\big(\{R\}\big)$, we thus have
\begin{align*}
	V_R\subseteq\big(F^{\omega_1}(\tau)\big)^{-1}\big(\{R\}\big)\subseteq V_R\cup U.
\end{align*}
Because $V_R$ is measurable (by \cref{cl:countable-measurable}) and $U$ has measure $0$,
the set $\big(F^{\eta_\infty}(\tau)\big)^{-1}\big(\{R\}\big)$ lies between two sets of the same measure; it is thus measurable.

Actually, the above argument gives us more: the measure of the sets $\big(F^\eta(\tau)\big)^{-1}\big(\{R\}\big)$ for all $\eta\geq\zeta$ is the same.
Due to $\zeta<\omega_1$ this means that sequence of measures of the sets $\big(F^\eta(\tau)\big)^{-1}\big(\{R\}\big)$ stabilises at (hence its limit is equal to)
the measure of $\big(F^\zeta(\tau)\big)^{-1}\big(\{R\}\big)$, which equals the measure of $\big(F^{\omega_1}(\tau)\big)^{-1}\big(\{R\}\big)$.
The above holds for every set $R\in\CR$ which by definition means that $\big(F^\eta(\tau)\big)_{\eta<\omega_1}$ converges in measure to $F^{\omega_1}(\tau)$.
This finishes the proof of \cref{pro:all-measurable-lim}.

\subsection{Unary \texorpdfstring{$\mu$}{mu}-calculus preserves measurability}

Let $\mathfrak{F}$ be a~signature and assume that $(\CW,{\leq})$ is a~structure over $\mathfrak{F}$ in such a~way
that for every function name $H\in\mathfrak{F}$ the function $H_\CW\from \CW\to\CW$ is local and preserves measurability on the whole $\CW$.
Our goal is to prove that each formula of unary $\mu$\=/calculus is also local and preserves measurability on a~certain domain,
as stated in \cref{all-local,pro:all-measurable}.

As shown in \cref{our-local}, the above assumptions are met by our structure with functions $\Delta$, $\Bid_n$, and $\Cut_n$,
but results of this subsection are not specific to this structure.

\begin{lemma}\label{all-local}
	If we can derive $F\dcolon A\to B$ for a~formula $F$ of unary $\mu$\=/calculus then the partial function $F\from\CW\parto\CW$ is local on $A$.
\end{lemma}

\begin{proof}
	Induction on the structure of the formula.
	Base functions $H\in\mathfrak{F}$ are local on the whole~$\CW$ by the assumption.
	For $F\comp G$ we directly use the induction hypothesis.
	Consider a~formula of the form $\limU{F}$.
	We use \cref{lem:lim-as-iteration} to see that $\limU{F}(\tau)=F^\infty(\tau)$ for $\tau\in A$.
	We then recall that $F^\infty(\tau)$ is defined by applying $F$ and by taking limits of monotonic sequences.
	Every application of $F$ is local: if $F^{\eta-1}(\tau)$ and $F^{\eta-1}(\tau')$ agree on all the subtrees of some $t_0$,
	then also $F^\eta(\tau)=F(F^{\eta-1}(\tau))$ and $F^\eta(\tau')=F(F^{\eta-1}(\tau'))$ agree on all the subtrees of $t_0$.
	For limits we have an~even stronger property:
	if $F^\zeta(\tau)$ and $F^\zeta(\tau')$ agree on a~tree $t$ (a~subtree of $t_0$) for every $\zeta<\eta$, where $\eta$ is a~limit ordinal,
	then also the limits $F^\eta(\tau)$ and $F^\eta(\tau')$ agree on $t$,
	because the order considers every tree $t$ separately.
	It follows that if $\tau,\tau'\in A$ agree on all the subtrees of $t_0$, then $\limU{F}(\tau)$ and $\limU{F}(\tau')$ agree on all the subtrees of $t_0$, that is, $\limU{F}$ is local on $A$.
	Finally, the case of $\limD{F}$ is symmetric.
\end{proof}

\begin{lemma}\label{pro:all-measurable}
	If we can derive $F\dcolon A\to B$ for a~formula $F$ of unary $\mu$\=/calculus then the partial function $F\from\CW\parto\CW$ preserves measurability on $A$.
\end{lemma}

\begin{proof}
	We proceed inductively over the shape of $F$.
	For $H\in\mathfrak{F}$ the thesis follows from the assumption on $H$, and for $F\comp G$ it is a~direct consequence of the induction hypothesis.
	For a~formula of the form $\limU{F}$ we know that $F$ is ascending on $A$ because we can derive $\limU{F}\dcolon A\to B$;
	that $F$ is local on $A$ by \cref{all-local}; and that $F$ preserves measurability on $A$ by the induction hypothesis.
	Then \cref{pro:all-measurable-lim} says in particular that
	if $\tau$ is a~random variable then $F^\infty(\tau)$ (which equals $(\limU{F})(\tau)$ by \cref{lem:lim-as-iteration}) is a~random variable;
	this is exactly what we need.
	The case of $\limD{F}$ is analogous.
\end{proof}

\section{Distributions}\label{sec:distributions}

Each random variable $\tau\in\CW$ has its \emph{distribution}---a~function assigning measures to possible values of $\tau$.
Recall that $\CW=(\trees_\Sigma\to\CR)$, with $\CR$ being a~finite complete lattice (in our case $\CR=\powerset(Q\times\{0,\ldots,d{-}1\})$).
Let $\CD=\mathbb{D}(\CR)$ be the set of probability distributions over $\CR$, that is, $\CD\eqdef \big\{\alpha\in [0,1]^\CR\mid \sum_{R\in\CR} \alpha(R) = 1\big\}$.
The \emph{distribution} of a~random variable $\tau\in\CW$, denoted $\widehat{\tau}\in\CD$, is defined for $R\in\CR$ as
\[\widehat{\tau}(R)\eqdef \prob\big(\set{t\in\trees_\Sigma\mid \tau(t)=R}\big).\]

\begin{definition}\label{def:dist-order}
	Consider two distributions $\alpha,\beta\in\CD$.
	We say that $\alpha\preceq \beta$ if for each upward\=/closed subset $\CU\subseteq\CR$ we have
	\[\sum_{R\in\CU} \alpha(R)\leq \sum_{R\in\CU} \beta(R).\]
\end{definition}

The above order was first introduced by Saheb{-}Djahromi~\cite{saheb-powerdomain} and is often called a \emph{probabilistic powerdomain}~\cite{probabilisticPowerdomains}.
It has been used in the preceding papers where measures of
the weakly MSO\=/definable sets of infinite trees were computed~\cite{przybylko_simple_sets,niwinski_measures_wmso}.

We first remark that $(\CD,{\preceq})$ may not be a~lattice, as expressed by the following fact.

\begin{lemma}
	There exists a finite complete lattice $\CR$ such that $\CD=\mathbb{D}(\CR)$ is not a lattice.
\end{lemma}

\begin{proof}
	Take as $\CR$ the powerset $\powerset(\{p,q\})$ for some two\=/element set $Q=\{p,q\}$.
	Consider two distributions, written as formal convex combinations of the elements of $\CR$:
	\begin{alignat*}{4}
		\alpha &= \frac{1}{2}\cdot \{p\} &&+ \frac{1}{2}\cdot \{q\},\\
		\beta &= \frac{1}{2}\cdot \emptyset &&+ \frac{1}{2}\cdot Q.
	\end{alignat*}

	Clearly $\alpha\npreceq\beta$ nor $\beta\npreceq \alpha$.
	To simplify notation, we say that $\gamma$ is an~\emph{upper bound} if both $\gamma\succeq \alpha$ and $\gamma\succeq \beta$.
	For $\CD$ to be a lattice, there should be a unique smallest upper bound $\gamma$.
	Consider
	\begin{alignat*}{4}
		\gamma_p &= \frac{1}{2}\cdot \{p\} &&+ \frac{1}{2}\cdot Q,\\
		\gamma_q &= \frac{1}{2}\cdot \{q\} &&+ \frac{1}{2}\cdot Q.
	\end{alignat*}
	It is easy to check that both $\gamma_p$ and $\gamma_q$ are upper bounds.
	Moreover, $\gamma_p\npreceq\gamma_q$ nor $\gamma_q\npreceq\gamma_p$.
	Consider any upper bound $\gamma$ and assume that $\gamma_p\succeq \gamma$.
	We will prove that $\gamma=\gamma_p$.
	Since $\gamma_p\succeq\gamma\succeq \beta$, by taking $\CU=\{Q\}$ we know that $\gamma(Q)=\frac{1}{2}$.
	Since $\gamma\succeq \alpha$, by taking $\CU=\{Q,\{p\}\}$ we know that $\gamma(Q)+\gamma(\{p\})=1$ and therefore $\gamma(\{p\})=\frac{1}{2}$.
	Thus, $\gamma=\gamma_p$.
	Similarly, $\gamma=\gamma_q$ for any upper bound $\gamma$ such that $\gamma_q\succeq\gamma$.

	This means that $\gamma_p$ and $\gamma_q$ are two distinct minimal upper bounds.
	Thus, there is no unique smallest upper bound.
\end{proof}

On the other hand, every chain (but not every set) in $(\CD, {\preceq})$ has a~supremum and an~infimum, as stated below
(this observation, however, is not used in our arguments).

\begin{proposition}
	Every chain $C$ in the structure $(\CD, {\preceq})$ has a~supremum and an~infimum.
\end{proposition}

\begin{proof}
	Consider a~chain $C\subseteq\CD$.
	By symmetry we assume that we need to find its supremum.
	Let $\emptyset=\CU_0,\ldots,\CU_{n-1}=\CR$ be the list of all upward\=/closed subsets $\CU_i\subseteq \CR$.
	We can represent each distribution $\alpha\in \CD$ as $\vec{\alpha}\in\R^n$, where $\vec{\alpha}_i=\sum_{R\in\CU_i} \alpha(R)$
	(the notation $\vec{\alpha}$ is used only for the sake of this proof).
	Clearly the value $\vec{\alpha}_i$ belongs to the interval $[0,1]$, as the probability of the event $R\in\CU_i$.

	By definition,
	\begin{equation}\label{eq:order-in-vectors}
		\alpha\preceq\beta\quad \Leftrightarrow \quad\vec{\alpha}\leq \vec{\beta},
	\end{equation}
	where we order the vectors $\vec{\alpha}$ and $\vec{\beta}$ by the coordinate\=/wise order in $\R$.
	Let $(a_0,\ldots,a_{n-1})$ be the coordinate\=/wise supremum of $\{\vec{\alpha}\mid \alpha\in C\}$.
	Due to this definition, the vector $(a_0,\ldots,a_{n-1})$ satisfies $a_0=0$, $a_{n-1}=1$, and if $\CU_i\subseteq \CU_j$ then $a_i\leq a_j$.

	We claim that $(a_0,\ldots,a_{n-1})=\vec{\beta}$ for some distribution $\beta\in \CD$.
	Note that in this case $\beta$ must be the supremum of $C$, because of \cref{eq:order-in-vectors}.

	Take $R\in\CR$ and let $U_R=\{R'\in\CR\mid R'\geq R\}$, $V_R=\{R'\in\CR\mid R' > R\}$.
	Both these sets are upward closed.
	Assume that $U_R=\CU_i$ and $V_R=\CU_j$ and put $\beta(R)\eqdef a_i - a_j$.
	This value clearly belongs to $[0,1]$ (we know that $\CU_j\subseteq \CU_i$).
	Moreover, by induction over the size of $\CU_i$ we can show that $\sum_{R\in \CU_i} \beta(R)= a_i$.
	Hence, $\sum_{R\in\CR} \beta(R)=a_{n-1}=1$, which means that $\beta\in \CD$.
	This implies that $\vec{\beta}$ is well defined and equal to $(a_0,\ldots,a_{n-1})$, which finishes the proof.
\end{proof}

The following two lemmata are rather standard.

\begin{lemma}\label{widehat-monotone}
	The mapping $\tau\mapsto\widehat{\tau}$ defined on the random variables $\tau\in\CW$ is monotone.
\end{lemma}

\begin{proof}
	Take $\tau,\tau'\in\CW$ such that $\tau\leq\tau'$, and take any upward\=/closed subset $\CU\subseteq\CR$.
	Then
	\begin{align*}
		\sum_{R\in\CU}\widehat{\tau}(R)
		&=\sum_{R\in\CU}\prob\big(\{t\in\trees_\Sigma\mid\tau(t)=R\}\big)\\
		&=\prob\big(\{t\in\trees_\Sigma\mid\tau(t)\in\CU\}\big)\\
		&\leq\prob\big(\{t\in\trees_\Sigma\mid\tau'(t)\in\CU\}\big)\\
		&=\sum_{R\in\CU}\prob\big(\{t\in\trees_\Sigma\mid\tau'(t)=R\}\big)\\
		&=\sum_{R\in\CU}\widehat{\tau'}(R),
	\end{align*}
	where the middle inequality follows from the assumption that $\tau(t)\subseteq\tau'(t)$ for all $t\in\trees_\Sigma$ and the fact that $\CU$ is upward closed.
	Having this for all $\CU$ implies that $\widehat\tau\preceq\widehat\tau'$
\end{proof}

\begin{lemma}\label{lem:conv-in-measure-as-dist-limits}
	If $(\tau_\zeta)_{\zeta<\eta}$ is a~monotonic sequence of random variables that converges in measure to a~random variable $\tau$ then
	the sequence of distributions $\widehat{\tau_\zeta}$ is also monotonic and its limit is $\widehat{\tau}$.
\end{lemma}

\begin{proof}
	To fix attention, suppose that the sequence $(\tau_\zeta)_{\zeta<\eta}$ is increasing; argumentation for a~decreasing sequence is symmetric.
	\cref{widehat-monotone} immediately implies that the sequence of distributions $(\widehat\tau_\zeta)_{\zeta<\eta}$ is increasing as well.
	
	Convergence in measure by definition implies that, for every $R\in\CR$, the sequence of numbers $\big(\widehat\tau_\zeta(R)\big)_{\zeta<\eta}$ converges to $\widehat\tau(R)$.
	Take any upward-closed subset $\CU\subseteq\CR$.
	Because $\CU$ is finite, the sequence of sums $(s_{\zeta,\CU})_{\zeta<\eta}$ defined as $s_{\zeta,\CU}\eqdef\sum_{R\in\CU}\widehat\tau_\zeta(R)$
	converges to $s_\CU\eqdef\sum_{R\in\CU}\widehat\tau(R)$.
	The sequence $(\widehat\tau_\zeta)_{\zeta<\eta}$ is increasing, which by definition means that $(s_{\zeta,\CU})_{\zeta<\eta}$ is increasing as well.
	In consequence $s_{\zeta,\CU}\leq s_\CU$ for all $\zeta<\eta$ (the limit of an~increasing sequence of numbers is above its elements).
	Having this for all $\CU$ implies $\widehat\tau_\zeta\preceq\widehat\tau$ for all $\zeta<\eta$.
	
	Consider now any other distribution $\beta$ such that $\widehat\tau_\zeta\preceq\beta$ for all $\zeta<\eta$,
	and take again any upward-closed subset $\CU\subseteq\CR$.
	The inequality $\widehat\tau_\zeta\preceq\beta$ implies $s_{\zeta,\CU}\leq s'_\CU$ for all $\zeta<\eta$, where $s'_\CU\eqdef\sum_{R\in\CU}\beta(R)$.
	But then also the limit of $(s_{\zeta,\CU})_{\zeta<\eta}$, namely $s_\CU$, satisfies $s_\CU\leq s'_\CU$.
	Having this for all $\CU$ implies $\widehat\tau\preceq\beta$.
	
	Thus $\widehat\tau$ is the least upper bound of the set $\{\widehat\tau_\zeta\mid\zeta<\eta\}$, that is, the limit of $(\widehat\tau_\zeta)_{\zeta<\eta}$.
\end{proof}

We remark that the mapping $\tau\mapsto\widehat{\tau}$ does not
preserve limits of all chains, or at least this would contradict the continuum hypothesis.
This is the reason why we need results of the previous section, in particular \cref{pro:all-measurable-lim}, saying that the limits we take are convergent in measure.
Indeed, assuming the continuum hypothesis, we can organize all trees in a~sequence of type $\omega_1$, that is, denote them $t_\eta$ for $\eta<\omega_1$.
Then all the sets $U_\eta=\set{t_\zeta\mid\zeta\leq\eta}$ for $\eta<\omega_1$ have measure $0$ (these are countable sets).
But the limit of the increasing sequence $(U_\eta)_{\eta<\omega_1}$ (i.e., union of these sets) is $\trees_\Sigma$, a~set of measure~$1$.

\subsection{Computing in \texorpdfstring{$\CD$}{D}---basic functions}\label{ssec:compute-D-basic}

Because now we consider both functions on $\CW$ and on $\CD$, in order to distinguish them we use names with a~subscript $\CW$ or $\CD$.

Recall the definition given in the proof of \cref{our-local}:
a~function $F_\CW\from\CW\to\CW$ is \emph{superficial} if there exists a~function $F_\CR\from\CR\to\CR$ such that $F_\CW(\tau)(t)=F_\CR(\tau(t))$ for all $\tau\in\CW$ and $t\in\trees_\Sigma$.
Having such a~function $F_\CW$ we define its \emph{lift to $\CD$} as the function $F_\CD\from \CD\to \CD$
such that $F_\CD(\alpha)(R)=\sum_{R'\in F_\CR^{-1}(\{R\})} \alpha(R')$ for all $\alpha\in \CD$ and $R \in \CR$.

\begin{lemma}
	If $F_\CW\from \CW\to \CW$ is superficial then for each random variable $\tau\in\CV$ we have
	\[\widehat{F_\CW(\tau)}=F_\CD(\widehat{\tau}).\]
\end{lemma}

\begin{proof}
	Fix a~set $R\in\CR$.
	For every tree $t\in\trees_\Sigma$ we have
	\begin{align*}
		F_\CW(\tau)(t)=R
		\ &\Leftrightarrow\
		F_\CR(\tau(t))=R
		\ \Leftrightarrow\
		\tau(t)\in F_\CR^{-1}(\{R\})
		\Leftrightarrow\
		\bigvee\nolimits_{R'\in F_\CR^{-1}(\{R\})}\tau(t)=R'.
	\end{align*}
	Moreover, at most one component of the above alternative is true (there is a~unique set $R'$ such that $\tau(t)=R'$).
	It follows that
	\begin{align*}
		\widehat{F_\CW(\tau)}(R)
		&=\prob\big(\{t\in\trees_\Sigma\mid F_\CW(\tau)(t)=R\}\big)\\
		&=\sum\nolimits_{R'\in F_\CR^{-1}(\{R\})}\prob\big(\{t\in\trees_\Sigma\mid\tau(t)=R'\}\big)\\
		&=\sum\nolimits_{R'\in F_\CR^{-1}(\{R\})}\widehat\tau(R')
		=F_\CD(\widehat\tau)(R).
	\tag*{\qedhere}\end{align*}
\end{proof}

We now define probabilistic variants of functions $\Delta_\CW$, $(\Bid_n)_\CW$ and $(\Cut_n)_\CW$.
First, $(\Bid_n)_\CD$ and $(\Cut_n)_\CD$ are just lifts of $(\Bid_n)_\CW\from\CW\to\CW$ and $(\Cut_n)_\CW\from\CW\to\CW$, which are superficial.

Now we need to define a~variant of $\Delta_\CW$ working on distributions, which we denote $\Delta_\CD\from \CD\to\CD$.
We use the function $\delta_a\from\powerset(Q)\times\powerset(Q)\to\powerset(Q)$ defined in \cref{sec:spaces}.
First, for $a\in\Sigma$ and $R^\dL,R^\dR\in\CR$ let us define
\begin{align*}
	\Delta_a(R^\dL, R^\dR)&\eqdef\bigcup\nolimits_{i=0}^{d-1} \delta_a\big(K^\dL_i, K^\dR_i\big){\times}\set{i},&&\text{where}\\
	&\hspace{-4em}K^\dL_i=\big\{q\in Q\mid(q,\max(i,\Omega(q)))\in R^\dL\big\}&&\text{and}\\
	&\hspace{-4em}K^\dR_i=\big\{q\in Q\mid(q,\max(i,\Omega(q)))\in R^\dR\big\}.
\end{align*}
Now take $\alpha\in \CD$ and $R\in\CR$ and let
\begin{equation}\label{eq:delta-D-wzor}
	\Delta_\CD(\alpha)(R)\eqdef\sum_{a\in\Sigma} \frac{1}{|\Sigma|}\cdot \sum_{(R^\dL,R^\dR)\in\Delta_a^{-1}(\{R\})} \alpha(R^\dL)\cdot \alpha(R^\dR).
\end{equation}

The following lemma can be proved by a~simple computation.

\begin{lemma}\label{lem:delta-commutes}
	For every random variable $\tau\in\CW$ we have
	\[\widehat{\Delta_\CW(\tau)}=\Delta_\CD(\widehat{\tau}).\]
\end{lemma}

\begin{proof}
	Fix a~random variable $\tau\in\CW$ and a~set $R\in\CR$.
	First, recalling the definition of $\Delta_{\CV^d}\from\CV^d\to\CV^d$, of $\delta\from\CV^d\to\CV$, and of the bijection between $\CV^d$ and $\CW$,
	we notice that $\Delta_\CW\from\CW\to\CW$ can be defined by
	\begin{align*}
		\Delta_\CW(\tau)(t)=\Delta_{t(\epsilon)}\big(\tau(t.\dL),\tau(t.\dR)\big).
	\end{align*}
	From the above it follows that
	\begin{align*}
		\big\{t\in\trees_\Sigma\mid\Delta_\CW(\tau)(t)=R\big\}
			=\bigcup\nolimits_{a\in\Sigma}\bigcup\nolimits_{(R^\dL,R^\dR)\in\Delta_a^{-1}(\{R\})}(A_a\cap B_{R^\dL}\cap C_{R^\dR}),
	\end{align*}
	where
	\begin{align*}
		A_a&=\big\{t\in\trees_\Sigma\mid t(\epsilon)=a\big\},\\
		B_{R^\dL}&=\big\{t\in\trees_\Sigma\mid \tau(t.\dL)=R^\dL\big\},&&\mbox{and}\\
		C_{R^\dR}&=\big\{t\in\trees_\Sigma\mid \tau(t.\dR)=R^\dR\big\}.
	\end{align*}
	Both unions in the above formula are disjoint unions: for every tree we have a~unique letter $t(\epsilon)$ and unique sets $\tau(t.\dL)$ and $\tau(t.\dR)$.
	Moreover, the three random events, $A_a$, $B_{R^\dL}$, and $C_{R^\dR}$ are independent:
	a~label in the root, labels in the left subtree, and labels in the right subtree are chosen independently.
	We have $\prob(A_a)=\frac{1}{|\Sigma|}$ by definition, and $\prob(B_{R^\dL})=\widehat\tau(R^\dL)$ and $\prob(C_{R^\dR})=\widehat\tau(R^\dR)$ by \cref{subtree-measurable}.
	Thus
	\begin{align*}
		\widehat{\Delta_\CW(\tau)}(R)&=\prob\big(\{t\in\trees_\Sigma\mid\Delta_\CW(\tau)(t)=R\}\big)\\
		&=\sum_{a\in\Sigma}\sum_{(R^\dL,R^\dR)\in\Delta_a^{-1}(\{R\})}\!\!\!\prob(A_a)\cdot\prob(B_{R^\dL})\cdot\prob(C_{R^\dR})\\
		&=\sum_{a\in\Sigma}\frac{1}{|\Sigma|}\cdot\sum_{(R^\dL,R^\dR)\in\Delta_a^{-1}(\{R\})}\!\!\!\widehat\tau(R^\dL)\cdot\widehat\tau(R^\dR)\\
		&=\Delta_\CD(\widehat\tau)(R).
	\tag*{\qedhere}\end{align*}
\end{proof}

\subsection{Moving from \texorpdfstring{$\CW$}{W} to \texorpdfstring{$\CD$}{D}}

Assume that $\CW$ and $\CD$ are structures over $\mathfrak{F}$ such that for every basic function $H\in\mathfrak{F}$,
\begin{enumerate}
\item the function $H_\CW$ is local and preserves measurability on the whole $\CW$, and
\item $H$ \emph{commutes with computing distributions}, that is, if $\tau\in\CW$ is a~random variable then $\widehat{H_\CW(\tau)}= H_\CD(\widehat\tau)$.
\end{enumerate}

\begin{proposition}\label{pro:commutes-with}
	Assume the above conditions on the structures $\CW$ and $\CD$.
	Let $F$ be a~formula of unary $\mu$\=/calculus such that we can derive $F\dcolon A\to B$ for some $A,B\subseteq\CW$.
	Let also $\tau\in A$ be a~random variable.
	Then $\widehat\tau\in\dom(F_\CD)$ and
	\[\widehat{F_{\CW}(\tau)} = F_\CD(\widehat\tau).\]
\end{proposition}

Recall \cref{pro:all-measurable}, which, in our situation, ensures that $F_\CW$ preserves measurability on $A$;
in particular $F_{\CW}(\tau)$ is a~random variable and it makes sense to consider $\widehat{F_{\CW}(\tau)}$.

The rest of this subsection is devoted to a proof of this proposition.
We proceed by induction over the structure of $F$.

For the basic functions $H\in\mathfrak{F}$ the statement follows from the assumptions that they commute with computing distributions.
For $F; G$ the fact follows directly from the induction hypothesis.
Consider a~formula of the form $\limU{F}$ (the case of $\limD{F}$ is symmetric) and take a~random variable $\tau\in A$.
By the assumption that we can derive $\limU{F}\dcolon A\to B$, we know that $F$ is ascending on $A$.
Therefore, \cref{lem:lim-as-iteration} implies that $(\limU{F})_{\CW}(\tau)=(F_\CW)^\infty(\tau)$.
Denote this random variable by $\sigma$.

The rest of the proof is realised by the following claim.

\begin{claim}
	We have $\widehat{\tau}\in\dom\big((\limU{F})_\CD\big)$ and $\widehat{\sigma}=(\limU{F})_\CD(\widehat{\tau})$.
\end{claim}

\begin{proof}
	Recall the definition of $(F_\CW)^\eta(\tau)$ used to define $(F_\CW)^\infty(\tau)$.
	By \cref{pro:all-measurable-lim} all the elements $(F_\CW)^\eta(\tau)$ are random variables.
	Note that $\sigma$ is a~fixed point of $F_\CW$, so the induction hypothesis implies that
	\[F_\CD(\widehat{\sigma})=\widehat{F_\CW(\sigma)}=\widehat{\sigma}.\]
	Thus, $\widehat{\sigma}\in\Fix(F_\CD)$.
	By definition $\tau\leq\sigma$, so also $\widehat\tau\preceq\widehat\sigma$ by \cref{widehat-monotone}.
	
	Consider now another distribution $\beta\in\CD$ with $\widehat{\tau}\preceq \beta$ such that $F_\CD(\beta)=\beta$.
	We need to prove that $\widehat{\sigma}\preceq \beta$.
	To this end, we prove by induction on $\eta$ that $\widehat{(F_\CW)^\eta(\tau)}\preceq \beta$.
	For $\eta=0$ this follows from the assumption.
	For a~successor ordinal $\eta$ we have
	\begin{align*}
		\widehat{F_\CW^\eta(\tau)}
		=\widehat{F_\CW\big(F_\CW^{\eta-1}(\tau)\big)}
		=F_\CD\big(\widehat{F_\CW^{\eta-1}(\tau)}\big)
		\preceq F_\CD(\beta)=\beta,
	\end{align*}
	where the second equality is the hypothesis of the external induction (about $F$)
	and the next inequality follows from monotonicity of $F_\CD$ and from the hypothesis of the internal induction (about $\beta$).

	Consider now a~limit ordinal $\eta$.
	\cref{pro:all-measurable-lim} implies that the monotonic sequence $\big(F_\CW^{\zeta}(\tau)\big)_{\zeta<\eta}$ converges in measure to $F_\CW^{\eta}(\tau)$.
	\cref{lem:conv-in-measure-as-dist-limits} thus implies that $\big(\widehat{F_\CW^{\zeta}(\tau)}\big)_{\zeta<\eta}$ has limit $\widehat{F_\CW^{\eta}(\tau)}$ in $\CD$.
	Therefore, using the induction hypothesis that $\widehat{F_\CW^{\zeta}(\tau)}\preceq \beta$ for all $\zeta<\eta$ we also have $\widehat{F_\CW^{\eta}(\tau)}\preceq \beta$.

	Thus, for the particular case of $\eta=\eta_\infty$ we have $\widehat{\sigma}=\widehat{(F_\CW)^{\eta_\infty}(\tau)}\preceq\beta$.

	We have just shown that $\widehat{\sigma}$ is the least fixed\=/point of $F_\CD$ above $\widehat{\tau}$ in $\CD$ and therefore $\widehat{\tau}\in\dom\big((\limU{F})_\CD\big)$
	and $(\limU{F})_\CD(\widehat{\tau})=\widehat{\sigma}$.
\end{proof}

\cref{ssec:compute-D-basic} defines basic functions $\Delta_\CD$, $(\Bid_n)_\CD$ and $(\Cut_n)_\CD$
in such a~way that the assumptions of \cref{pro:commutes-with} are met by our structures $\CW$ and $\CD$.
This gives us the following corollary for the specific case of $F=\Phi_{d-1}$.

\begin{corollary}\label{cor:commutes-with-Phi}
	For every random variable $\tau\in \CS_{d-1}$ we have $\widehat{\tau}\in\dom\big((\Phi_{d-1})_\CD\big)$ and
	\[\widehat{(\Phi_{d-1})_\CW(\tau)}=(\Phi_{d-1})_\CD(\widehat{\tau}).\]
\end{corollary}

\subsection{An a posteriori explanation of the formula \texorpdfstring{$\Phi_{d-1}$}{Phi\_\{d-1\}}}\label{sec:why-phi}

In this section we explain the reasons for the shape of the formula $\Phi_{d-1}$.
The source of difficulty of our construction can be explained already for $d=2$ when $\Omega\from Q\to\{0,1\}$ (i.e.,~the \emph{co\=/B\"uchi} case).
In this case the formula defining $\tau_\CA$ from \cref{tau-automaton} has the form
\begin{equation}\label{eq:form-aut-for-2}
	\mu x_1.\,\nu x_0.\,\delta(x_0,x_1).
\end{equation}
When we move to distributions of random variables, it is clearly not enough to consider separately a distribution of $x_0$ and a distribution of $x_1$,
because the distribution of $\delta(x_0,x_1)$ depends on the joint distribution of its arguments.

Thus, one could try to work with pairs of profiles and consider functions $\CV^2\to\CV^2$.
This leads to a~function
\[U_2\from (x_0,x_1)\mapsto \big(\delta(x_0,x_1),x_1\big),\]
which is a~monotone and chain\=/continuous function on $\CV^2$.
Moreover, the fixed point $\nu x_0.\,\delta(x_0,x_1)$ can be computed in a~descending way, that is, as the first coordinate of $\limD{U_2}(\top,x_1)$.

The problem here is that again the operation $U_2$ cannot be simulated on distributions of random variables:
the events used to compute the value of $\delta(\widehat{\tau}_0,\widehat{\tau}_1)$ (i.e.,~the root, the left subtree, and the right subtree of a~random tree)
and the event used to compute the value of $\widehat{\tau}_1$ (i.e.,~the whole tree) are not independent, unlike in the proof of \cref{lem:delta-commutes}.
Roughly speaking, to compute $U_2$ one has to look simultaneously at the whole randomly generated tree $t$ as well as its root label $t(\epsilon)$ and its two subtrees $t.\dL$ and $t.\dR$.
The variables $\big(t(\epsilon), t.\dL, t.\dR\big)$ are independent as a~vector, but stop being independent when we also consider the whole tree $t$.

The above problem would not arise, if the involved operator had the form
\[U'_2\from (x_0,x_1)\mapsto \big(\delta(x_0,x_1),\delta(x_1,x_1)\big).\]
This operation indeed can be translated to distributions in such a~way that it commutes with computing distributions.
Moreover, if $\tau_1$ is any fixed point of $x_1\mapsto\delta(x_1,x_1)$ then obviously $U'_2(\tau_0,\tau_1)=U_2(\tau_0, \tau_1)$.

This explains the motivation standing behind our formula $\Phi_1$: on some coordinates, instead of applying the identity $\tau_1\mapsto\tau_1$,
we can apply a~proper variant of $\delta$ if we make sure in advance that the current value of the given coordinate is already a~fixed point of it
(see the definition of a~$2$\=/saturated tuple $(\tau_0,\tau_1)$ from \cref{ssec:invariants}).
To achieve this, we need to add additional nested fixed\=/point operations on variants of $\delta$, as in the inductive definition of $\Phi_n$ from $\Phi_{n-1}$.
The operations $\Bid_n$ and $\Cut_n$, moving values between coordinates, give us a technical tool for switching between computations of particular fixed points from the original formula.

\section{Tarski's first order theory of reals}\label{sec:tarski}

Recall that $\CD$ contains probability distributions over a~finite lattice $\CR$.
Let $n$ be the cardinality of $\CR$.
We now treat $\CD$ as a~set of vectors $\vec{a}=(a_0,\ldots,a_{n-1})\in\R^n$.

In this section we think of $\R$ as a~logical structure with $0$, $1$, the operations of addition and multiplication, and the usual linear order, that is, $\langle \R, 0, 1, {+}, {*},{\leq}\rangle$.
Take the standard definition of first\=/order logic.
Our goal is to show that unary $\mu$\=/calculus over distributions $\CD$ can be expressed in first\=/order logic over $\R$.
The technique closely follows an~analogous argument given by Przybyłko and Skrzypczak~\cite{przybylko_simple_sets} for the case of safety automata.

Tarski's quantifier elimination~\cite{tarskiDecision1951} implies that if $\varphi(x_0,\allowbreak\ldots,\allowbreak x_{n-1})$ is a~formula of first\=/order logic over the reals $\R$
then the set of vectors $\vec{x}\in\R^n$ satisfying it is a~semialgebraic set~\cite[Chapter~2]{bochnak_real_algebraic},
that is, a~set represented by a~Boolean combination of equations and inequalities involving polynomials in the variables $x_0,\ldots,x_{n-1}$.
If this semialgebraic set is a~singleton $\{\vec{x}\}$, then we can treat its representation as a~representation of the vector $\vec{x}\in\R^n$.

\begin{theorem}[\cite{collins_algebraic_decomposition}]
	A~semialgebraic representation of the set of vectors satisfying a~given formula $\varphi(x_0,\ldots,x_{n-1})$ can be computed in time double\=/exponential in the size of $\varphi(x_0,\ldots,x_{n-1})$.
\qed\end{theorem}

Our goal is to use this method to represent the distribution $\widehat{\tau_{\CA}}$.
To improve the complexity of the decision problems (e.g.,~whether $\prob\big(\lang(\CA)\big)>0$) we rely on the following result.

\begin{theorem}[\cite{benOrFOonRealsComplexity}]\label{thm:fo-R-decidable}
	The theory of real\=/closed fields can be decided in deterministic exponential space.
	In particular, given a~formula $\varphi$ with no free variables, it is decidable in deterministic exponential space whether $\varphi$ holds.
\qed\end{theorem}

A set $A\subseteq\CD$ is \emph{definable} if there exists a~formula $\psi_A(x_0,\ldots,x_{n-1})$ such that $\vec{x}\in A$ if and only if $\psi_A(\vec{x})$ holds.
A distribution $\alpha\in \CD$ is \emph{definable} if the singleton set $\{\alpha\}$ is definable.
A partial function $F\from\CD\parto\CD$ is \emph{definable}
if there exists a~formula $\psi_F(x_0,\ldots,x_{n-1}, y_0,\ldots, y_{n-1})$ such that $\psi_F(\vec{x},\vec{y})$ holds if and only if $\vec{x}\in\dom(F)$ and $\vec{y}=F(\vec{x})$.
Similarly, we introduce a~notion of a~relation on $\CD$ being definable.

\begin{lemma}
	The basic functions $(\Bid_n)_\CD$, $(\Cut_n)_\CD$, and $\Delta_\CD$ from $\CD$ to $\CD$ are definable.
\end{lemma}

\begin{proof}
	The case of $(\Bid_n)_\CD$ and $(\Cut_n)_\CD$ follows directly from the definition.
	For $\Delta_\CD$, it is enough to write the arithmetic expression from \cref{eq:delta-D-wzor}.
\end{proof}

\begin{lemma}\label{lem:delta-and-preceq-def}
	The set of distributions $\alpha$ is definable.
	Moreover, the order $\preceq$ on distributions is also definable.
\end{lemma}

\begin{proof}
	A vector encodes a~distribution if its values are between $0$ and $1$ and they sum up to $1$.

	It remains to write a~formula $\varphi_{\preceq}(\alpha, \beta)$ which holds if and only if $\alpha\preceq \beta$.
	One solution is just to write a~large conjunction over all upward\=/closed subsets of $\CR$, as in \cref{def:dist-order}.
	However, one can simulate the quantification over upward\=/closed sets within first\=/order logic,
	see the proof of Theorem~6.1 in Przybyłko and Skrzypczak~\cite{przybylko_simple_sets}.
	The formula given there is only polynomial in the cardinality of $\CR$.
\end{proof}

\begin{proposition}\label{prop:all-definable}
	If all the basic functions $H\in\mathfrak{F}$ are definable then for every formula $F$ of unary $\mu$\=/calculus we know that the partial function $F_\CD\from\CD\parto\CD$ is definable.
\end{proposition}

\begin{proof}
	The equations used to define $\dom(F_\CD)$ and $F_\CD(x)$ in \cref{ssec:intensional} are clearly expressible in first\=/order logic over $\R$,
	using the formulae for the basic functions and for the order ${\preceq}$ on $\CD$.
\end{proof}

The following fact follows directly from definitions.

\begin{fact}\label{ft:result-definable}
	If a~partial function $F\from \CD\parto \CD$ is definable and its argument $\vec{x}\in\dom(F)$ is also definable then the value $F(\vec{x})$ is also definable.
\qed\end{fact}

We are now in place to prove our main theorem.

\begin{proof}[Proof of \cref{thm:main-theorem}]
	Input a~nondeterministic (or alternating, see \cref{rem:alternating}) parity tree automaton $\CA$.
	Define $\CI\subseteq\CR$ as the set $\{R\in\CR\mid (q_I,d{-}1)\in R\}$.
	Observe that
	\begin{align}
		\prob\big(\lang(\CA)\big)&=\prob\Big(\big\{t\in\trees_\Sigma\mid \tau_\CA(t)\ni q_I\big\}\Big)\nonumber\\
		&=\prob\Big(\big\{t\in\trees_\Sigma\mid (\Phi_{d-1})_{\CV^d}(\bot)_{d-1}(t)\ni q_I\big\}\Big)\nonumber\\
		&=\prob\Big(\big\{t\in\trees_\Sigma\mid (\Phi_{d-1})_\CW(\bot)(t)\ni (q_I,d{-}1)\big\}\Big)\nonumber\\
		&=\sum_{R\in \CI}\prob\Big(\big\{t\in\trees_\Sigma\mid (\Phi_{d-1})_\CW(\bot)(t)=R\big\}\Big)\nonumber\\
		&=\sum_{R\in \CI}\widehat{(\Phi_{d-1})_\CW(\bot)}(R)\nonumber\\
		&=\sum_{R\in \CI}(\Phi_{d-1})_\CD(\widehat{\bot})(R),\label{eq:final-sum}
	\end{align}
	where:
	\begin{itemize}
	\item the first equality is just the definition of $\lang(\CA)$;
	\item the second is the composition of equations from \cref{tau-automaton,prop:realisation};
	\item the third relies on the identification between $\CV^d=\big(\trees_\Sigma\to \powerset(Q)\big)^d$ and $\CW=(\trees_\Sigma\to \CR)$ where $\CR=\powerset\big(Q\times\{0,\ldots,d{-}1\}\big)$;
	\item the fourth is just additivity of measure (for distinct sets $R$ the considered sets of trees are disjoint);
	\item the fifth is just the definition of $\widehat{\tau}$;
	\item the sixth equality follows from \cref{cor:commutes-with-Phi}.
	\end{itemize}

	Now, recall that $\bot\in \CW$ is the function constantly equal $\emptyset\in\CR$.
	Therefore, the distribution $\widehat{\bot}$ assigns probability $1$ to $\emptyset$ and $0$ to all other $R\in\CR$.
	Thus, this distribution is definable and therefore \cref{ft:result-definable} implies that the distribution $(\Phi_{d-1})_\CD\big(\widehat{\bot}\big)$ (treated as a~vector) is definable.
	In consequence, the sum of some of its coordinates as in~\cref{eq:final-sum} must also be definable.

	Therefore, first of all, the number $\prob\big(\lang(\CA)\big)$ is an~algebraic number.
	Moreover, Tarski's quantifier elimination allows us to effectively compute a~representation of this number.
	Finally, \cref{thm:fo-R-decidable} says that the decision problems, like whether $\prob\big(\lang(\CA)\big)>q$ for a~given rational number $q$, are also decidable.
\end{proof}

\subsection{Complexity}\label{ssec:complexity}

We provide here upper bounds on the complexity of the decision method provided in the proof of \cref{thm:main-theorem}.
The situation is analogous to that from the cited proof of Theorem~6.1 in Przybyłko and Skrzypczak~\cite{przybylko_simple_sets}.
Our parameter is $N=|Q|$, that is, the number of states of the given automaton~$\CA$.
The representation of the whole automaton $\CA$ is polynomial in $N$, and so is the value of $d\leq\max_{q\in Q}\Omega(q)+1$
(we can assume that the values of $\Omega$ form an~interval and the least value equals either $0$ or $1$).
Thus, the syntactic size of the formula $\Phi_{d-1}$ defined at the beginning of \cref{sec:formula} is also polynomial in $N$ (in fact linear).

Now, the cardinality of the set $\CR=\powerset\big(Q\times\{0,\ldots,d{-}1\}\big)$ is exponential in $N$.
Thus, so is $n$---the length of the vectors $\vec{x}\in\R^n$ which encode distributions $\alpha\in \CD$.
Similarly, exponential in $N$ are the formulae of first\=/order logic that define the function $\Delta_\CD$ and the order ${\preceq}$
(see the comment in the proof of \cref{lem:delta-and-preceq-def}).

\begin{lemma}
	The translation of $\Phi_{d-1}$ into first\=/order logic over reals given by \cref{prop:all-definable} is exponential in the size of $\CA$.
\end{lemma}

\begin{proof}[Sketch of the proof]
	The proof is inductive on the sub\=/formula $\Phi_i$ for $i\in\set{0,\ldots,d{-}1}$.
	The translation of the formula $\Phi_0$ is exponential in $N = |Q|$, because so are the vectors of quantified variables and the basic predicates for expressing $\Delta_\CD$ and ${\preceq}$.
	Now, the translation of the formula $\Phi_{i+1}$ is obtained from the translation of $\Phi_{i}$ by copying it a~constant number of times
	(because of the definition of $\limU{F}(x)$ and $\limD{F}(x)$) and combining with translations of $\Delta_\CD$.
	Thus, the translation of the final formula $\Phi_{d-1}$ has size of the form $C^N\cdot E(N)$,
	where $C$ is a~constant and $E$ is an exponential function bounding the size of the basic predicates.
\end{proof}

\begin{corollary}
	Given a (nondeterministic or alternating) automaton $\CA$, a representation of the algebraic number $\prob\!\big(\lang(\CA)\big)$ can be computed in three\=/fold exponential time in the size of $\CA$.
	For every fixed rational number $q$, the decision problem whether $\prob\!\big(\lang(\CA)\big)$ is equal, smaller, or greater than~$q$
	can be solved in two\=/fold exponential space in the size of $\CA$.
\qed\end{corollary}

Note that the latter decision problem is known to be EXPTIME\=/hard~\cite[Proposition~7.6]{przybylko_simple_sets}.

\section{Stochastic branching processes}\label{sec:branching}

In this section we consider measures on the set of trees induced by stochastic branching processes, showing how \cref{cor:branching} follows from \cref{thm:main-theorem}.
Essentially, such a reduction (in a slightly different setting) was already given by Niwiński, Przybyłko, and Skrzypczak~\cite{niwinski_measures_wmso} (with unpublished proofs in~\cite{niwinski_measures_wmso_arxiv}).
To make this paper self\=/contained, we provide here a similar reduction, with a~minor improvement with respect to complexity:
we construct a product automaton whose size is polynomial in the size of the branching process (including the binary representations of the rational probabilities involved) and of the original automaton.

In this section, we slightly generalize the notion of a tree, as follows.
A \emph{ranked alphabet} $(\Sigma,\ar)$ consists of a~finite set $\Sigma$ and of an~arity function $\ar\colon\Sigma\to\N$.
Given a number $s\in\N$ we write $\Sigma_s$ for $\set{a\in\Sigma\mid\ar(a)=s}$.
Then, a \emph{tree} over $(\Sigma,\ar)$ is a partial function $t\colon\N^*\parto\Sigma$ whose domain $\dom(t)$ is non\=/empty and closed under taking prefixes,
and such that for each \emph{node} $v\in\dom(t)$ and \emph{direction} $d\in\N$ we have $vd\in\dom(t)$ if and only if $1\leq d\leq\ar(t(v))$
(i.e., $\ar(a)$ specifies the number of children of every node labelled by $a$).
A binary tree as defined in \cref{sec:basic-concepts} can be then seen as a~special case where the arity of every letter is $2$, if we identify the directions $\dL$, $\dR$ with $1$, $2$.
The set of all trees over $(\Sigma,\ar)$ is denoted $\trees_{\Sigma,\ar}$.

An automaton over such a tree is defined as previously, except that a transition relation is of the form $\gamma\subseteq Q\times\bigcup_{s=0}^\infty\Sigma_s\times Q^s$;
then, defining a run $\rho$ we require that $\big(\rho(v),t(v),\rho(v1),\dots,\allowbreak\rho(vs)\big)\in\gamma$, where $s=\ar(t(v))$.

A \emph{branching process} is a tuple $\CP=\langle\Sigma,\ar,R,r_I,\Theta,\theta\rangle$, where $(\Sigma,\ar)$ is a finite ranked alphabet, $R$ is a finite set of states,
$r_I\in R$ is an initial state, $\Theta\subseteq\bigcup_{s=0}^\infty\Sigma_s\times R^s$ is a set of transitions, and $\theta\colon R\to\mathbb{D}(\Theta)$ is a function assigning distributions of transitions to states.
We assume that all probabilities occurring in $\theta$ are rational, with numerator and denominator encoded in binary.

Intuitively, in order to generate a tree, we start with the state $r_I$ in the root;
then, while having a~state~$r$ in some node $v$, the function $\theta(r)$ specifies the probability of each tuple $(a,r_1,\dots,r_{\ar(a)})\in\Theta$;
we choose at random one of such tuples, we put the label $a$ in $v$, and we send the states $r_1,\dots,r_{\ar(a)}$ to the children of $v$.
Formally, a \emph{tree prefix} is a partial function $f\colon\N^*\to\Sigma$ such that for each node in $\dom(f)$ of the form $vd$, where $n\in\N^*$ and $d\in\N$, we have $v\in\dom(f)$ and $1\leq d\leq\ar(f(v))$;
it \emph{agrees} with a tree $t$ if $t\restr_{\dom(f)}=f$.
A \emph{run} of $\CP$ over a tree prefix $f$ is a function $\rho\colon\dom(f)\to\Theta$ such that for each node $v\in\dom(\rho)$, the first coordinate of $\rho(v)$ equals $f(v)$.
Suppose that $\dom(\rho)$ is finite.
For every state $r\in R$ we define by induction on $|\dom(\rho)|$ probabilities
$$p_\rho^r=\left\{\begin{array}{ll}
	1&\mbox{if }\dom(\rho)=\emptyset,\\
	\theta(r)(\rho(\epsilon))\cdot p_{\rho\restr_1}^{r_1}\cdots p_{\rho\restr_s}^{r_s}&\mbox{if }\rho(\epsilon)=(a,r_1,\dots,r_s).
\end{array}\right.$$
Then, we define $U_f$ as the set of trees $t$ that agree with a tree prefix $f$,
and we equip $\trees_{\Sigma,\ar}$ with a topology generated by a basis consisting of sets $U_f$ for all tree prefixes $f$ having finite domain.
Finally, we consider the Lebesgue measure $\prob_\CP$ defined for the basis by $\prob_\CP(U_f)=\sum_\rho p_\rho^{r_I}$, where the sum is over all runs $\rho$ of $\CP$ over $f$.

Our goal is to prove the following theorem, which allows us to deduce \cref{cor:branching} from \cref{thm:main-theorem}.

\begin{theorem}\label{thm:reduction}
	Given a branching process $\CP$ and a nondeterministic parity tree automaton $\CA$ over the same ranked alphabet $(\Sigma,\ar)$,
	one can construct a nondeterministic parity tree automaton $\CB$ over an unranked (i.e., with all letters of arity $2$) alphabet $\set{0,1}$ such that $\prob(\lang(\CB))=\prob_\CP(\lang(\CA))$.
	Moreover, the size of $\CB$ can be polynomial in the size of $\CP$ and $\CA$.
\end{theorem}

As a first step, we show how to realize a random choice from a finite set according to an arbitrary distribution as a choice of a random tree.

\begin{lemma}\label{lem:reduction-numbers}
	Let $\alpha\in\mathbb{D}(\Theta)$ be a probability distribution over a finite set $\Theta$, with all probabilities rational.
	Then, there exists a function $e_\alpha\colon\trees_{\set{0,1}}\to\Theta$ such that for each $x\in\Theta$
	\begin{itemize}
	\item	we can construct an automaton recognizing the set $e^{-1}_\alpha(x)$, of size polynomial in the size of the representation of~$\alpha$, and
	\item	the measure of $e^{-1}_\alpha(x)$ equals $\alpha(x)$.
	\end{itemize}
\end{lemma}

\begin{proof}
	Let $x_1,\dots,x_k$ be all the elements of $\Theta$,
	let $N$ be the product of denominators of all probabilities in $\alpha$, and let $n$ be the length of the binary representation of $N$.
	Consider also the numbers $p_i=N\cdot\sum_{j\leq i}\alpha(x_i)$ for $i\in\set{0,\dots,k}$; these are integers, where $p_0=0$ and $p_k=N$.

	Given a tree $t\in\trees_{\set{0,1}}$, we split its rightmost branch into (infinitely many) fragments of length~$n$, and we read labels on these fragments as binary numbers between $0$ and $2^n-1$.
	If none of these numbers is smaller than $N$, we put $e_\alpha(t)=x_1$.
	Otherwise, we choose the first number that is smaller that $N$, denote it $\ell_t$, and we define $e_\alpha(t)$ to be the unique element $x_i$ for which $p_{i-1}\leq\ell_t<p_i$.
	
	It is easy to construct automata recognizing the sets $e^{-1}_\alpha(x_i)$.
	Moreover, note that $n$ is of polynomial size (the length of a~product is bounded by the sum of lengths of factors),
	and that the automata do not need to remember the whole number written in a fragment of the rightmost branch, but rather they can compare it with $p_{i-1}$ and $p_i$ bit-by-bit.
	This means that they can be of size polynomial in the size of the representation of $\alpha$.

	For the second item of the thesis, note that the probability of obtaining a particular number in a given fragment of the rightmost branch is $2^{-n}$, and that results in different fragments are independent.
	Thus, given $i\in\set{1,\dots,k}$ and $j\geq 0$,
	the probability that the first $j$ fragments contain numbers above $N-1$ equals $(1-N\cdot2^{-n})^j$,
	and the probability that the $j$-th fragment contains a number $\ell$ satisfying $p_{n-1}\leq\ell<p_n$ equals $N\cdot\alpha(x_i)\cdot 2^{-n}$,
	hence the probability that both these events hold equals $(1-N\cdot2^{-n})^j\cdot N\cdot\alpha(x_i)\cdot 2^{-n}$.
	Moreover, all the fragments contain numbers above $N{-}1$ with probability $\lim_{j\to\infty}(1-N\cdot2^{-n})^j=0$; this way of obtaining $e_\alpha(t)=x_1$ is negligible.
	Summing up over all $j$, we obtain
	$$\prob(e^{-1}_\alpha(x_i))=\sum_{j=0}^\infty(1-N\cdot2^{-n})^j\cdot N\cdot\alpha(x_i)\cdot2^{-n}=\frac{N\cdot\alpha(x_i)\cdot2^{-n}}{N\cdot 2^{-n}}=\alpha(x_i),$$
	as requested.
\end{proof}

Heading towards a proof of \cref{thm:reduction}, we now specify how binary trees over $\set{0,1}$ encode ranked trees over $(\Sigma,\ar)$.
To this end, let $\CP=\langle\Sigma,\ar,R,r_I,\Theta,\theta\rangle$.
We define a tuple of functions $F_r\colon\trees_{\set{0,1}}\to\trees_{\Sigma,\ar}$, one for each state $r\in R$.
Consider a tree $t\in\trees_{\set{0,1}}$, and recall the function $e_{\theta(r)}\colon\trees_{\set{0,1}}\to\Theta$ from \cref{lem:reduction-numbers}.
Denote $e_{\theta(r)}(t\restr_\dL)=(a,r_1,\dots,r_s)$; we have $s=\ar(a)$.
Then, we define $F_r(t)$ to be the tree having label $a$ in its root, and the tree $F_{r_i}(t\restr_{\dR^i\dL})$ in the $i$-th child of the root, for $i\in\set{1,\dots,s}$
(see \cref{fig:reduction} for an illustration).
The above definition is self-recursive, but clearly there is a unique tuple of functions satisfying it.

\begin{figure}
\begin{center}
	\begin{tikzpicture}[scale=0.6]
	
	\newcommand{\subtr}[3]{
	\coordinate (#1) at (#2) {};
	
	\draw (#1) -- ++(-0.5,-1);
	\draw (#1) -- ++(+0.5,-1);
	\node[anchor=center] at ($(#1)+(0,-0.8)$) {#3};
	}
	
	\newcommand{\subTR}[3]{
	\coordinate (#1) at (#2) {};
	
	\draw (#1) -- ++(-1.2,-2.4);
	\draw (#1) -- ++(+1.2,-2.4);
	\node[anchor=center] at ($(#1)+(0,-2.0)$) {#3};
	}
	
	\node (t) at (0,+1.5) {$t$};
	\node (T) at (10,+1.5) {$F_r(t)$};
	\draw (t) edge[|-Latex, bend left=10] (T);
	
	\tikzstyle{treeN} = [draw,circle]
	\tikzstyle{treeE} = [draw]
	\tikzstyle{treeD} = [line cap=round,thick,dash pattern=on \pgflinewidth off 12pt]
	
	\node[treeN] (r) at (0,0) {};
	\subtr{tl}{$(r)+(-1.5,-2)$}{$t'$}
	\draw[treeE] (r) -- (tl);
	
	\node[treeN] (t1) at ($(r)+(+1,-2)$) {};
	\draw[treeE] (r) -- (t1);
	\subtr{tr1}{$(t1)+(-1,-2)$}{$t_1$}
	\draw[treeE] (t1) -- (tr1);
	\draw[treeE] (t1) -- ++(0.6,-1.2);
	\coordinate (dds) at ($(t1)+(+1.2,-2.4)$);
	
	\draw[treeD] ($(dds)+(-0.3,+0.6)$) -- ($(dds)+(+0.3,-0.6)$);
	
	\coordinate (dds) at ($(t1)+(+1.2,-2.4)+(-1,-2)$);
	
	\draw[treeD] ($(dds)+(-0.3,+0.6)$) -- ($(dds)+(+0.3,-0.6)$);
	
	\node[treeN] (t2) at ($(t1)+(+2,-4)$) {};
	
	\subtr{tr2}{$(t2)+(-1,-2)$}{$t_s$}
	\subtr{te}{$(t2)+(+1,-2)$}{}
	
	\draw[treeE] (t2) -- (tr2);
	\draw[treeE] (t2) -- (te);
	
	\node[treeN] (r) at (10,0) {$a$};
	\subTR{tr1}{$(r)+(-2.5,-2)$}{$T_{r_1}(t_1)$}
	\subTR{tr2}{$(r)+(+2.5,-2)$}{$T_{r_s}(t_s)$}
	\draw[treeE] (r) -- (tr1);
	\draw[treeE] (r) -- (tr2);
	
	\draw[treeD] ($(r)+(-1,-2.5)$) -- ($(r)+(+1,-2.5)$);
	\end{tikzpicture}
\end{center}
	\caption{An illustration of the function $F_r$, for $r\in R$.
		The first subtree, denoted $t'$, is used to determine $(a,r_1,\dots,r_s)$ as $e_{\theta(r)}(t')$.
		This $a$ goes to the label of the root,	while states $s_i$ together with subtrees denoted $t_i$ are used to recursively define subtrees attached below the root.
		Labels in the remaining part of the input tree $t$ are irrelevant for the construction of $F_r(t)$.
	}\label{fig:reduction}
\end{figure}

\begin{claim}\label{cl:reduction}
	For every measurable set $L\subseteq\trees_{\Sigma,\ar}$ we have $\prob(F_{r_I}^{-1}(L))=\prob_\CP(L)$.
\end{claim}

\begin{proof}
	We define $G_r$ analogously to $F_r$, except that it labels the root of the created tree (and likewise, recursively, all its nodes) by the whole tuple $(a,r_1,\dots,r_s)$, not only by the letter $a$.
	Then $G_r(t)$ is a run of $\CP$ over $F_r(t)$.
	Moreover, for a run $\rho$ over a tree prefix we define $U_\rho$ as the set of runs $\rho'$ of $\CP$ over a whole tree such that $\rho'\restr_{\dom(\rho)}=\rho$.
	
	By additivity of measures (more specifically Dynkin's $\pi$\=/$\lambda$ Theorem), it is enough to prove the claim for sets $L$ in the basis, that is, when $L=U_f$ for a tree prefix $f$ with finite domain.
	By definition $\prob_\CP(U_f)=\sum_\rho p_\rho^{r_I}$, and we can see that $F^{-1}_{r_I}(U_f)=\biguplus_\rho G^{-1}_{r_I}(U_\rho)$;
	in both cases $\rho$ ranges over runs of $\CP$ over $f$.
	It is thus enough to prove that $\prob(G^{-1}_r(U_\rho))=p_\rho^r$ for every run $\rho$ of $\CP$ with a~finite domain, and every $r\in R$.
	We prove it by induction on $|\dom(\rho)|$.
	If $\dom(\rho)=\emptyset$, then every run $\rho'$ of $\CP$ satisfies $\rho'\restr_\emptyset=\rho$, hence belongs to $U_\rho$;
	we thus have $\prob(G^{-1}_r(U_\rho))=1=p_\rho^r$.
	Let now $\dom(\rho)\neq\emptyset$ with $\rho(\epsilon)=(a,r_1,\dots,r_s)$.
	By definition $G^{-1}_r(U_\rho)$ contains trees $t$ such that $t\restr_\dL\in e^{-1}_{\theta(r)}(\rho(\epsilon))$ and $t\restr_{\dR^i\dL}\in G^{-1}_{r_i}(U_{\rho\restr_i})$ for all $i\in\set{1,\dots,s}$.
	These events are independent, because they concern disjoint parts of the tree $t$.
	The measure of $e^{-1}_{\theta(r)}(\rho(\epsilon))$ equals $\theta(r)(\rho(\epsilon))$ by \cref{lem:reduction-numbers},
	and the measure of $G^{-1}_{r_i}(U_{\rho\restr_i})$ equals $p_{\rho\restr_i}^{r_i}$ by the induction hypothesis.
	We thus have $\prob(G_{r_I}^{-1}(U_\rho))=\theta(r)(\rho(\epsilon))\cdot p_{\rho\restr_1}^{r_1}\cdots p_{\rho\restr_s}^{r_s}=p^r_\rho$, as needed.
\end{proof}

Finally, we construct an automaton $\CB$, which reads a tree $t\in\trees_{\set{0,1}}$, and simulates a run of $\CA=\langle \Sigma,Q,q_I,\gamma,\Omega\rangle$ on $F_{r_I}(t)$.
This automaton keeps track of a state $q$ of $\CA$ and a state $r$ of~$\CP$.
It guesses some transition $x=(a,r_1,\dots,r_s)$, and sends to the left subtree $t'$ a subautomaton (constructed in \cref{lem:reduction-numbers}) checking that $e_{\theta(r)}(t')$ is indeed $x$;
then it simulates a transition of~$\CA$ reading the letter $a$, and sends appropriate states of $\CA$ and $\CP$ to the next subtrees to the left of rightmost branch
(we omit tedious details of this construction, as they are easy to reproduce).
The size of $\CB$ is polynomial: its ``main states'' are pairs from $Q\times R$;
it also has auxiliary states that remember transitions $(a,r_1,\dots,r_s)\in\Theta$ and $(q,a,q_1,\dots,q_s)\in\gamma$ together with a number up to $s$,
and states of the $|R|\cdot|\Theta|$ automata from \cref{lem:reduction-numbers}.
Then, we have $\lang(\CB)=F_{r_I}^{-1}(\lang(\CA))$, hence $\prob(\lang(\CB))=\prob_\CP(\lang(\CA))$ by \cref{cl:reduction}.
This finishes a proof of \cref{thm:reduction}.

\section{Conclusions}

We have shown that the following question can be effectively solved.
\begin{quote}
	What is the probability that a~randomly chosen tree satisfies a~given formula of MSO?
\end{quote}
This can be viewed as a~completion of the famous Rabin tree theorem by its probabilistic aspect.
If a~formula is replaced by an equivalent nondeterministic or alternating parity automaton $\CA$, our algorithm works in elementary time.

We believe that our results or techniques may find some further applications although---as mentioned in Introduction---%
a~number of related problems in probabilistic verification of finite-state system is known to be undecidable.
One of the research directions points to logics whose original version reduces to MSO logic over trees, but the decidability status of the probabilistic version is unknown as, for instance~PCTL*.

On the technical level, our work confirms the usefulness of the construction of the probabilistic powerdomain $\CD$ that comprises all the necessary information and is still manageable.
We have also discovered that the framework of the $\mu$\=/calculus---originally based on complete lattices---can be transposed to more general partial orders.
The unary $\mu$\=/calculus that we have introduced for this purpose has a~potential of incorporating other formalisms, going beyond the MSO theory of the tree.
As a matter of fact, our construction provides a way to translate a~generic formula of unary $\mu$\=/calculus into a~formula of Tarski’s first\=/order theory of reals,
which yields decidability of the formalism in question.

\paragraph*{Acknowledgements} The authors would like to express their gratitude to Grzegorz Cichosz, Marcin Przybyłko, and Igor Walukiewicz, for their valuable comments on the paper.

\bibliographystyle{alpha}
\bibliography{mskrzypczak}

\begin{thebibliography}{GMMS17}

\bibitem[AN01]{niwinski_rudiments}
Andr{\'{e}} Arnold and Damian Niwiński.
\newblock {\em Rudiments of {$\mu$}-Calculus}.
\newblock Studies in Logic and the Foundations of Mathematics. Elsevier, 2001.

\bibitem[BBG08]{BaierBG08}
Christel Baier, Nathalie Bertrand, and Marcus Gr{\"{o}}{\ss}er.
\newblock On decision problems for probabilistic {B}{\"{u}}chi automata.
\newblock In {\em Foundations of Software Science and Computational Structures,
  11th International Conference, {FOSSACS} 2008, Held as Part of the Joint
  European Conferences on Theory and Practice of Software, {ETAPS} 2008,
  Budapest, Hungary, March 29 - April 6, 2008. Proceedings}, volume 4962 of
  {\em Lecture Notes in Computer Science}, pages 287--301. Springer, 2008.

\bibitem[BCR98]{bochnak_real_algebraic}
Jacek Bochnak, Michel Coste, and Marie-Fran{\c{c}}oise Roy.
\newblock {\em Real Algebraic Geometry}, volume~36 of {\em A Series of Modern
  Surveys in Mathematics}.
\newblock Springer-Verlag Berlin Heidelberg, 1998.

\bibitem[BOKR86]{benOrFOonRealsComplexity}
Michael Ben-Or, Dexter Kozen, and John Reif.
\newblock The complexity of elementary algebra and geometry.
\newblock {\em J. Comput. Syst. Sci.}, 32(2):251--264, 1986.

\bibitem[B{\"{u}}c62]{buchi_decision}
Julius~Richard B{\"{u}}chi.
\newblock On a decision method in restricted second-order arithmetic.
\newblock In {\em Proc. 1960 Int. Congr. for Logic, Methodology and Philosophy
  of Science}, pages 1--11, 1962.

\bibitem[BW18]{Julian-Igor-hand}
Julian~C. Bradfield and Igor Walukiewicz.
\newblock The mu-calculus and model checking.
\newblock In {\em Handbook of Model Checking}, pages 871--919. Springer, 2018.

\bibitem[CDK12]{chen_model_checking}
Taolue Chen, Klaus Dr{\"{a}}ger, and Stefan Kiefer.
\newblock Model checking stochastic branching processes.
\newblock In {\em Mathematical Foundations of Computer Science 2012 - 37th
  International Symposium, {MFCS} 2012, Bratislava, Slovakia, August 27-31,
  2012. Proceedings}, volume 7464 of {\em Lecture Notes in Computer Science},
  pages 271--282. Springer, 2012.

\bibitem[CHS14]{CarayolHS14}
Arnaud Carayol, Axel Haddad, and Olivier Serre.
\newblock Randomization in automata on infinite trees.
\newblock {\em {ACM} Trans. Comput. Log.}, 15(3):24:1--24:33, 2014.

\bibitem[CJH04]{chatterjee_stochastic_parity}
Krishnendu Chatterjee, Marcin Jurdziński, and Thomas~A. Henzinger.
\newblock Quantitative stochastic parity games.
\newblock In {\em Proceedings of the Fifteenth Annual {ACM-SIAM} Symposium on
  Discrete Algorithms, {SODA} 2004, New Orleans, Louisiana, USA, January 11-14,
  2004}, pages 121--130. {SIAM}, 2004.

\bibitem[Col75]{collins_algebraic_decomposition}
George~E. Collins.
\newblock Hauptvortrag: Quantifier elimination for real closed fields by
  cylindrical algebraic decomposition.
\newblock In {\em Automata Theory and Formal Languages, 2nd {GI} Conference,
  Kaiserslautern, May 20-23, 1975}, volume~33 of {\em Lecture Notes in Computer
  Science}, pages 134--183. Springer, 1975.

\bibitem[CY95]{courcou-yanna-acm}
Costas Courcoubetis and Mihalis Yannakakis.
\newblock The complexity of probabilistic verification.
\newblock {\em J. {ACM}}, 42(4):857--907, 1995.

\bibitem[DGL16]{Demri2016}
St{\'{e}}phane Demri, Valentin Goranko, and Martin Lange.
\newblock {\em Temporal Logics in Computer Science: Finite-State Systems}.
\newblock Cambridge Tracts in Theoretical Computer Science. Cambridge
  University Press, 2016.

\bibitem[FGW08]{2008thomas}
J{\"{o}}rg Flum, Erich Gr{\"{a}}del, and Thomas Wilke, editors.
\newblock {\em Logic and Automata: History and Perspectives [in Honor of
  Wolfgang Thomas]}, volume~2 of {\em Texts in Logic and Games}. Amsterdam
  University Press, 2008.

\bibitem[GMMS17]{michalewski_measure_final}
Tomasz Gogacz, Henryk Michalewski, Matteo Mio, and Micha{\l} Skrzypczak.
\newblock Measure properties of regular sets of trees.
\newblock {\em Inf. Comput.}, 256:108--130, 2017.

\bibitem[GO10]{gimbert_probabilistic}
Hugo Gimbert and Youssouf Oualhadj.
\newblock Probabilistic automata on finite words: Decidable and undecidable
  problems.
\newblock In {\em Automata, Languages and Programming, 37th International
  Colloquium, {ICALP} 2010, Bordeaux, France, July 6-10, 2010, Proceedings,
  Part {II}}, volume 6199 of {\em Lecture Notes in Computer Science}, pages
  527--538. Springer, 2010.

\bibitem[Izu01]{takeuti-rational}
Takeuti Izumi.
\newblock The measure of an omega regular language is rational.
\newblock Technical report, Kyoto University, 2001.

\bibitem[JP89]{probabilisticPowerdomains}
C.~Jones and Gordon~D. Plotkin.
\newblock A probabilistic powerdomain of evaluations.
\newblock In {\em Proceedings of the Fourth Annual Symposium on Logic in
  Computer Science {(LICS} '89), Pacific Grove, California, USA, June 5-8,
  1989}, pages 186--195. {IEEE} Computer Society, 1989.

\bibitem[JW95]{JaninW95}
David Janin and Igor Walukiewicz.
\newblock Automata for the modal mu-calculus and related results.
\newblock In {\em Mathematical Foundations of Computer Science 1995, 20th
  International Symposium, MFCS'95, Prague, Czech Republic, August 28 -
  September 1, 1995, Proceedings}, volume 969 of {\em Lecture Notes in Computer
  Science}, pages 552--562. Springer, 1995.

\bibitem[Kir01]{Kirsten-alter-auto}
Daniel Kirsten.
\newblock Alternating tree automata and parity games.
\newblock In {\em Automata, Logics, and Infinite Games: {A} Guide to Current
  Research [outcome of a Dagstuhl seminar, February 2001]}, volume 2500 of {\em
  Lecture Notes in Computer Science}, pages 153--167. Springer, 2001.

\bibitem[LS18]{lusin_sierpinski}
Nicholai Lusin and Wacław Sierpiński.
\newblock Sur quelques proprietes des ensembles {(A)}.
\newblock {\em Bull. Acad. Sci. Cracovie}, pages 35--48, 1918.

\bibitem[Mio12]{mio_branching_games}
Matteo Mio.
\newblock Probabilistic modal {$\mu$}-calculus with independent product.
\newblock {\em Log. Methods Comput. Sci.}, 8(4):1--36, 2012.

\bibitem[MM15]{michalewski_comp_measure}
Henryk Michalewski and Matteo Mio.
\newblock On the problem of computing the probability of regular sets of trees.
\newblock In {\em 35th {IARCS} Annual Conference on Foundation of Software
  Technology and Theoretical Computer Science, {FSTTCS} 2015, December 16-18,
  2015, Bangalore, India}, volume~45 of {\em LIPIcs}, pages 489--502. Schloss
  Dagstuhl - Leibniz-Zentrum f{\"{u}}r Informatik, 2015.

\bibitem[Niw88]{niwinski88}
Damian Niwiński.
\newblock Fixed points vs. infinite generation.
\newblock In {\em Proceedings of the Third Annual Symposium on Logic in
  Computer Science {(LICS} '88), Edinburgh, Scotland, UK, July 5-8, 1988},
  pages 402--409. {IEEE} Computer Society, 1988.

\bibitem[Niw97]{niwinski_mu_calc_index}
Damian Niwiński.
\newblock Fixed point characterization of infinite behavior of finite-state
  systems.
\newblock {\em Theor. Comput. Sci.}, 189(1-2):1--69, 1997.

\bibitem[NPS20]{niwinski_measures_wmso}
Damian Niwiński, Marcin Przyby{\l}ko, and Micha{\l} Skrzypczak.
\newblock Computing measures of {Weak-MSO} definable sets of trees.
\newblock In {\em 47th International Colloquium on Automata, Languages, and
  Programming, {ICALP} 2020, July 8-11, 2020, Saarbr{\"{u}}cken, Germany
  (Virtual Conference)}, volume 168 of {\em LIPIcs}, pages 136:1--136:18.
  Schloss Dagstuhl - Leibniz-Zentrum f{\"{u}}r Informatik, 2020.

\bibitem[NPS23]{probabilistic-rabin}
Damian Niwiński, Paweł Parys, and Michał Skrzypczak.
\newblock The probabilistic rabin tree theorem.
\newblock In {\em 38th Annual {ACM/IEEE} Symposium on Logic in Computer
  Science, {LICS} 2023, Boston, MA, USA, June 26-29, 2023}, pages 1--13.
  {IEEE}, 2023.

\bibitem[NPS24]{niwinski_measures_wmso_arxiv}
Damian Niwiński, Marcin Przybyłko, and Michał Skrzypczak.
\newblock Computing measures of weak-mso definable sets of trees, 2024.

\bibitem[NW03]{niwinski_gap}
Damian Niwiński and Igor Walukiewicz.
\newblock A gap property of deterministic tree languages.
\newblock {\em Theor. Comput. Sci.}, 1(303):215--231, 2003.

\bibitem[Oxt80]{oxtoby}
John~C.\ Oxtoby.
\newblock {\em Measure and Category: A Survey of the Analogies between
  Topological and Measure Spaces}.
\newblock Graduate Texts in Mathematics. Springer-Verlag, 1971,1980.

\bibitem[Paz71]{Paz1971}
Azaria Paz.
\newblock {\em Introduction to Probabilistic Automata}.
\newblock Academic Press, 1971.

\bibitem[Pin21]{Pin2021}
Jean{-}{\'{E}}ric Pin, editor.
\newblock {\em Handbook of Automata Theory}.
\newblock European Mathematical Society Publishing House, Z{\"{u}}rich,
  Switzerland, 2021.

\bibitem[PP04]{perrin_pin_words}
Dominique Perrin and {Jean-{\'{E}}ric} Pin.
\newblock {\em Infinite Words: Automata, Semigroups, Logic and Games}, volume
  141 of {\em Pure and Applied Mathematics}.
\newblock Elsevier Morgan Kaufmann, 2004.

\bibitem[PS20]{przybylko_simple_sets}
Marcin Przyby{\l}ko and Micha{\l} Skrzypczak.
\newblock The uniform measure of simple regular sets of infinite trees.
\newblock {\em Inf. Comput.}, 2020.

\bibitem[Rab69]{rabin_s2s}
Michael~Oser Rabin.
\newblock Decidability of second-order theories and automata on infinite trees.
\newblock {\em Transactions of the American Mathematical Society}, 141:1--35,
  1969.

\bibitem[Sah80]{saheb-powerdomain}
Nasser Saheb{-}Djahromi.
\newblock Cpo's of measures for nondeterminism.
\newblock {\em Theor. Comput. Sci.}, 12:19--37, 1980.

\bibitem[Tar51]{tarskiDecision1951}
Alfred Tarski.
\newblock {\em A Decision Method for Elementary Algebra and Geometry}.
\newblock University of California Press, 1951.

\bibitem[Tho96]{thomas_languages}
Wolfgang Thomas.
\newblock Languages, automata, and logic.
\newblock In {\em Handbook of Formal Languages}, pages 389--455. Springer,
  1996.

\bibitem[Var85]{Vardi85}
Moshe~Y. Vardi.
\newblock Automatic verification of probabilistic concurrent finite-state
  programs.
\newblock In {\em 26th Annual Symposium on Foundations of Computer Science,
  Portland, Oregon, USA, 21-23 October 1985}, pages 327--338. {IEEE} Computer
  Society, 1985.

\bibitem[Var99]{Vardi99-proba}
Moshe~Y. Vardi.
\newblock Probabilistic linear-time model checking: An overview of the
  automata-theoretic approach.
\newblock In {\em Formal Methods for Real-Time and Probabilistic Systems, 5th
  International {AMAST} Workshop, ARTS'99, Bamberg, Germany, May 26-28, 1999.
  Proceedings}, volume 1601 of {\em Lecture Notes in Computer Science}, pages
  265--276. Springer, 1999.

\end{thebibliography}

\end{document}